\documentclass[]{article}
\usepackage{setspace}
\usepackage{amsmath,amssymb,latexsym,amsopn,amsfonts,amsthm}
\usepackage[utf8]{inputenc}

\usepackage{mathtools}
\theoremstyle{plain}
\newcounter{thm} \numberwithin{thm}{section}
\newtheorem{lemma}[thm]{Lemma}
\newtheorem{corollary}[thm]{Corollary}
\newtheorem{definition}[thm]{Definition}
\newtheorem{claim}[thm]{Claim}
\theoremstyle{definition}
\newtheorem{remark}[thm]{Remark}
\newtheorem{example}[thm]{Example}

\usepackage{amsmath}
\usepackage{bm}

\usepackage{wrapfig}
\usepackage{caption}
\usepackage{subcaption}

\usepackage{hyperref}
\hypersetup{
    colorlinks=false,
    pdfborder={0 0 0},
}

\usepackage{amssymb}
\usepackage{framed}
\usepackage{verbatim}
\usepackage[normalem]{ulem}
\usepackage{centernot}
\usepackage{graphicx}
\usepackage{color}
\usepackage[boxed,ruled,vlined,linesnumbered]{algorithm2e}
\usepackage{epstopdf}
\usepackage{enumerate}

\newtheorem{myTheorem}{Theorem}[section]
\newtheorem{requirement}{\textbf{Requirement}}
\newtheorem{property}{\textbf{Property}}
\newcommand{\PARAgraph}[1]{\subparagraph{#1.}\space}
\newcommand{\Paragraph}[1]{\paragraph{#1.\space}}
\newcommand{\Subparagraph}[1]{\subparagraph{#1.}\space}

\newcommand{\smallSpaceBetweenEquations}{}
\newcommand{\smallestSpaceBetweenEquations}{}

\newcommand{\remove}[1]{}
\newcommand{\N}{\mathbb{N}}
\newcommand{\Z}{\mathbb{Z}}
\newcommand{\la}{\langle}
\newcommand{\ra}{\rangle}
\newcommand{\lea}{\prec_{b}}
\newcommand{\leqa}{\preceq_{b}}

\newcommand{\leqlb}{\preceq_{lb}}
\newcommand{\nled}{\nprec_{lb}}
\newcommand{\pinf}{\mathcal{L}_\mathcal{S}}
\newcommand{\tempinf}{\mathcal{\ell}_\mathcal{S}}
\newcommand{\LE}{\mathsf{LE}}

\newcommand{\nll}{\centernot\ll}

\newcommand{\capacity}{\mathcal{C}}
\newcommand{\msg}{M}

\newcommand{\pairInvariants}{\textit{pairInvar}}
\newcommand{\localInvariants}{\textit{localInvariants}}
\newcommand{\globalInvariants}{\textit{globalInvariants}}

\newcommand{\lb}{\prec_{lb}}

\newcommand{\cP}{\mathcal{P}}
\newcommand{\cS}{\mathcal{S}}
\newcommand{\cT}{\mathcal{T}}

\newcommand{\cX}{\mathcal{X}}

\newcommand{\bigO}{\mathcal{O}}
\newcommand{\getLabel}{getLabel}
\newcommand{\creator}{creator}
\newcommand{\MI}{MAXINT}
\newcommand{\mmi}{(\bmod~\MI)}
\newcommand{\idv}{idV(i)}
\newcommand{\locCurr}{loc.curr}
\newcommand{\locPrev}{loc.prev}
\newcommand{\arrCurr}{arr.curr}
\newcommand{\arrPrev}{arr.prev}
\newcommand{\pivot}{\textit{pivot}}
\newcommand{\oc}{\out.curr}

\newcommand{\labelCheck}{legitPairs}
\newcommand{\revive}{revive}
\newcommand{\out}{output}

\newcommand{\compLbls}{\textit{comparableLabels}}

\newcommand{\sizeOfInvariantsFig}{\normalsize} 
\newcommand{\mirroredLocal}{\textit{mirroredLocalLabels}}
\newcommand{\isStored}{\textit{isStored}}
\newcommand{\existsOverlap}{\textit{existsPivot}}  
\newcommand{\resetLocal}{restartLocal}

\newcommand{\legitArriving}{legitMsg}
\newcommand{\cancelPairLabels}{\textit{cancelPairLabels}}
\newcommand{\lblOrdrd}{\textit{labelsOrdered}}
\newcommand{\causalPrecedence}{\textit{causalPrecedence}}
\newcommand{\newEvents}{newEvents}

\newcommand{\seg}{\sqsubseteq}
\newcommand{\oft}{P_{often}}
\newcommand{\noft}{P_{notOften}}
\newcommand{\isCanceled}{\textit{isCanceled}}

\newcommand{\reqs}{~\ref{req:2act}}
\newcommand{\pgraph}{\mathcal{G}}
\newcommand{\F}{\mathcal{F}}
\newcommand{\FF}{\mathcal{F}_{focused}}
\newcommand{\newLabel}{newLabel}
\newcommand{\T}{\mathcal{T}}
\newcommand{\clone}{clone}
\newcommand{\lB}{labelBookkeeping()}
\newcommand{\lbnl}{labelBookkeeping()\circ\newLabel()}

\newcommand{\pairs}{pairs}
\newcommand{\receiverLocal}{rcvdLocal}
\newcommand{\equalStatic}{equalStatic}
\newcommand{\initializeToLocal}{initToLoc}
\newcommand{\s}{\mathcal{S}}

\SetKw{Let}{let}
\SetKw{Upon}{upon}
\SetKw{Send}{send}
\SetKw{To}{to}

\definecolor{shadecolor}{rgb}{0.9,0.9,0.9}
\definecolor{heraldBlue}{rgb}{0.0,0.0,0.8}
\definecolor{heraldRed}{rgb}{0.8,0.0,0.0}
\definecolor{heraldGray}{rgb}{0.4,0.4,0.4}
\definecolor{heraldBlack}{rgb}{0.0,0.0,0.0} 
\definecolor{heraldGreen}{rgb}{0.0,0.4,0.0} 

\definecolor{cinnabar}{rgb}{0.89, 0.26, 0.2}

\newcommand{\reviewCommands}[2]{#1}

\reviewCommands{
\newcommand{\IS}[1]{\textcolor{blue}{[\textbf{IS: #1}]}}
\newcommand{\EMS}[1]{\textcolor{heraldGreen}{\small \textbf{[EMS: #1]}}}
\newcommand{\Ver}[1]{\textcolor{heraldBlue}{\small\textsf{[EMS: #1]}}}
\newcommand{\opodis}[1]{\textcolor{cinnabar}{[\textbf{OPODIS: #1}]}}
}{
\newcommand{\IS}[1]{}
\newcommand{\EMS}[1]{}
\newcommand{\Ver}[1]{}
\newcommand{\opodis}[1]{}
}

\newcommand{\short}[1]{#1}
\newcommand{\modified}[2]{#1}\xspace 

\title{Practically-Self-Stabilizing Vector Clocks\\ in the Absence of Execution Fairness\footnote{Department of Computer Science and Engineering, Chalmers University of Technology, G\"oteborg, Sweden. \texttt{\{iosif,elad\}@chalmers.se}}\\\Large{(technical report)}}

\author{Iosif~Salem \and Elad M.\ Schiller}

\date{}

\begin{document} 

\begin{titlepage}


\maketitle

\begin{abstract}
Vector clock algorithms are basic wait-free building blocks that facilitate causal ordering of
events. 
As wait-free algorithms, they are guaranteed to complete their operations within a finite number
of steps. 
Stabilizing algorithms allow the system to recover after the occurrence of transient faults,
such as soft errors and arbitrary violations of the assumptions according to which the system was designed to behave. 
We present the first, to the best of our knowledge, stabilizing vector clock
algorithm for asynchronous crash-prone message-passing systems that can recover in a \textit{wait-free manner} after the occurrence of transient faults. 
In these settings, it is challenging to demonstrate a finite and \textit{wait-free} 
recovery from (communication and crash failures as well as) \textit{transient faults}, bound the message and storage sizes, deal with the removal of all stale information \textit{without blocking}, and deal with counter overflow events (which occur at different network nodes concurrently). 


We present an algorithm 
that never violates safety in the absence of transient faults and provides bounded time recovery during fair executions that follow the last transient fault.
The novelty is that in the absence of execution fairness, the algorithm guarantees a bound on the number of times in which the system might violate safety (while existing algorithms might block forever due to the presence of both transient faults and crash failures).

Since vector clocks facilitate a number of elementary synchronization building
blocks (without requiring remote replica synchronization) in asynchronous systems,
we believe that our analytical insights are useful for the design of other systems that cannot guarantee execution fairness.
\end{abstract}


\end{titlepage}

 
\section{Introduction}
\Paragraph{Context and Motivation}
Vector clocks allow reasoning about causality among events in distributed systems, for example, when constructing distributed snapshots~\cite{DBLP:books/daglib/0020056}. 
Shapiro et al.~\cite{DBLP:conf/sss/ShapiroPBZ11} showed that vector clocks are building blocks of several conflict-free replicated data types (CRDTs). 
CRDTs are distributed data structures that can be shared among many replicas in asynchronous networks. 
All replica updates occur independently and achieve \emph{strong eventual consistency} without using mechanisms for synchronization~\cite{DBLP:conf/sigmod/Skeen81} or roll-back. 

The industrial use of CRDTs includes globally distributed databases, such as the ones of Redis, Riak, Bet365, SoundCloud, TomTom, Phoenix, and Facebook. 
Some of these databases have around ten million concurrent users, ten thousand messages per second, store 
large volumes of data,
and offer very low latency. 
%
%
However, while both the literature and the users demonstrate that large-scale decentralized systems can benefit from the use of CRDTs in general and vector clocks in particular, the relationship between fault-tolerance and strong eventual consistency has not received sufficient attention. 
Providing higher robustness degrees to CRDTs is nevertheless imperative for ensuring the availability and safety of these systems.

Providing robustness in the presence of unexpected failures, i.e., the ones that are not included the fault model, is challenging, especially in the absence of synchrony, mechanisms for synchronization, or roll-back. 
In such systems, it is difficult to: (A) provide unbounded storage and message size, (B) model all possible failures, and (C) guarantee periods in which all nodes are up and connected. 

The goal of this paper is the design of a highly fault-tolerant 
distributed
algorithm for vector clocks in large-scale asynchronous message passing systems. 
In particular, we propose the first, to the best of our knowledge, practically-self-stabilizing algorithm for vector clocks that: 
(I) uses strictly bounded storage and message size, 
(II) deals with a relevant set of failures (i.e., a fault model) as well as with unexpected failures  (i.e., failures that are not considered by the fault model), and 
(III) the algorithm does not require synchronization guarantees, nor uses mechanisms for synchronization or roll-back \emph{even during the period of recovery from unexpected failures}.

\Paragraph{Fault Model}
We consider asynchronous message-passing systems that are prone to the following failures~\cite{georgiou2011cooperative}:
(a) crash failures of nodes (no recovery after crashing), 
(b) nodes that can crash and then perform an undetectable restart, i.e., resume with the same state as before crashing (without knowing explicitly that a crash has ever occurred), but possibly having lost incoming messages in between, and
(c) packet failures, such as omission, duplication, and reordering. 
In addition to these benign failures, we consider \emph{transient faults}, i.e., any temporary violation of assumptions according to which the system and network were designed to behave, e.g., the corruption of the system state due to soft errors.
We assume that these transient faults arbitrarily change the system state in unpredictable manners (while keeping the program code intact). 
Moreover, since these transient faults are rare, the system model assumes that all transient faults occurred before the start of the system run. 


\Paragraph{Design criteria}
Dijkstra~\cite{DBLP:journals/cacm/Dijkstra74} requires self-stabilizing systems, which may start in an arbitrary state, to return to correct behavior within a bounded period. 
Asynchronous systems (with bounded memory and channel capacity) can indefinitely hide stale information that transient faults introduce unexpectedly.
At any time, this corrupted data can cause the system to violate safety. 
This is true for any system, and in particular, for Dijkstra's self-stabilizing systems~\cite{DBLP:journals/cacm/Dijkstra74}, which are required to remove, within a bounded time, all stale information whenever they appear.
Here, the scheduler acts as an adversary that has a bounded number of opportunities to disrupt the system.
However, this adversary never reveals \emph{when} it will disrupt the system.
Against such unfair adversaries, systems cannot specify when they will be able to remove all stale information and thus they cannot fulfill Dijkstra's requirements.

\emph{Pseudo-self-stabilization}~\cite{DBLP:journals/dc/BurnsGM93} deals with the above inability by bounding the number of times in which the system violates safety. 
We consider the newer criteria of \emph{practically-self-stabilizing systems}~\cite{DBLP:journals/jcss/AlonADDPT15,DBLP:journals/jcss/DolevKS10,DBLP:conf/netys/BlanchardDBD14,DBLP:journals/corr/DolevGMS15} that can address additional challenges. 
For example, any transient fault can cause a bounded counter to reach its maximum value and yet the system might need to increment the counter for an unbounded number of times after that overflow event. 
This challenge is greater when there is no elegant way to maintain an order among the different counter values, say, by wrapping around to zero upon counter overflow. 
Existing attempts to address this challenge use non-blocking resets in the absence of faults, as described in~\cite{DBLP:journals/jpdc/AroraKD06}. 
In case faults occur, the system recovery requires the use of a synchronization mechanism that, at best, blocks the system until the scheduler becomes fair.
We note that this assumption contradicts our fault model as well as the key liveness requirement for recovery after the occurrence of transient faults. 

Without fair scheduling, a system that takes an extraordinary (or even an infinite) number of steps is bound to break any ordering constraint, because unfair schedulers can arbitrarily suspend node operations and defer message arrivals until such violations occur. 
Having practical systems in mind, we consider this number of (sequential) steps to be no more than \emph{practically infinite}~\cite{DBLP:journals/jcss/DolevKS10,DBLP:journals/corr/DolevGMS15}, say, $2^b$ (where $b=64$ or an even a larger integer, as long as a constant number of bits can represent it).  
Practically-self-stabilizing systems~\cite{DBLP:journals/jcss/AlonADDPT15, DBLP:conf/netys/BlanchardDBD14, DBLP:journals/corr/DolevGMS15} require a bounded number of safety violations during any practically infinite period of the system run. 
For such systems, we propose an algorithm for vector clocks that recovers after the occurrence of transient faults (as well as all other failures considered by our fault model) without assuming synchrony or using synchronization mechanisms. 
We refer to the latter as a wait-free recovery from  transient faults.
%
We note that the concept of \textit{practically}-self-stabilizing systems is named by the concept of \textit{practically} infinite executions~\cite{DBLP:journals/jcss/DolevKS10}.



To the end of providing safety (and independently of the practically-self-stabilizing algorithm), the application  can use a synchronization mechanism (similar to~\cite{DBLP:journals/jcss/AlonADDPT15, DBLP:conf/netys/BlanchardDBD14, DBLP:journals/corr/DolevGMS15, DBLP:conf/wdag/JehlVM15}). 
The advantage here is that the application can selectively use synchronization only when needed (without requiring the entire system to be synchronous or blocking after the occurrence of transient faults).



\Paragraph{Vector clocks}
%
Logical and vector clocks~\cite{DBLP:journals/cacm/Lamport78,fidge1987timestamps,VirtTimeGlobStates} capture chronological relationships in decentralized systems without accessing synchronization mechanisms, such as synchronized clocks and phase-based commit protocols~\cite{DBLP:conf/sigmod/Skeen81,DBLP:conf/netys/BlanchardDBD14}.
A common (non-self-stabilizing and unbounded) way for implementing vector clocks is to let the nodes maintain a local copy of the vector $V[]$, such that each of the $N$ system nodes has a component, e.g., $V[i]$ is the component of node $p_i$. 
Upon the occurrence of a local event, $p_i$ \textit{increments} $V_i[i]$, and sends an update message $m=\langle V[] \rangle$. 
Upon $m$'s arrival to node $p_j$, the latter \textit{merges} the events counted in $V[]$ and $m.V[]$ by assigning $V[j] \gets \max(V[j],m.V[j])$ for each component $V[j]$.
%
%
One can define the relation $\leq_C$ as a partial order, where $V$ and $W$ are $N$-size integer vectors and $(V \leq_C W) \iff (\forall x\in \{1,\ldots, N\}, V[x] \leq W[x])$. 
The relation $\leq_C$ is used to show causality between two events by checking if the corresponding vector clocks are comparable in $\leq_C$.

%

We note that there exist approaches for improving the scalability and efficiency of vector clocks that offer bounded size vectors (instead of linear) or approximations~\cite[Section 7]{DBLP:books/daglib/0032304}.
These approaches build on, implement, or provide similar semantics to the standard $N$-size vector definition of a vector clock.
Thus, in this paper we focus on the definition of a vector clock as an $N$-size vector.



\Paragraph{The studied question} 
How can non-failing nodes dependably reason about event causality?
We interpret the provable dependability requirement to imply (1) bounded message size and node storage, (2) fault-tolerance independently of synchrony assumptions or synchronization operations, 
%
%
and (3) the system to be practically-self-stabilizing (without fair scheduling).  

\Paragraph{Related work}
Bounded non-stabilizing solutions exist in the literature~\cite{DBLP:conf/wdag/AlmeidaAB04,DBLP:journals/dc/MalkhiT07}.
Self-stabilizing resettable vector clocks~\cite{DBLP:journals/jpdc/AroraKD06} consider distributed applications that are structured in phases and track causality merely within a bounded number of successive phases. Whenever the system exceeds the number of clock values that can be used in one phase, resettable vector clocks use reset operations that allow the system to move to the next phase and reuse clock values. In the absence of faults as presented in~\cite{DBLP:journals/jpdc/AroraKD06}, the system uses non-blocking resets. Nevertheless, the presence of faults can bring the algorithm in~\cite{DBLP:journals/jpdc/AroraKD06} to use a \textit{blocking} global reset that requires fair scheduling (and no failing nodes). Our solution does not use blocking operations even after an arbitrary corruption of the system state. 


The authors of~\cite{DBLP:journals/jpdc/AroraKD06} also discuss the possibility to use global snapshots for the sake of providing better complexity measures. 
They rule out this approach because it can change the communication patterns (in addition to the use of blocking operations during the recovery period). 
Another concern is how to identify a self-stabilizing snapshot algorithm that can deal with crash failures, e.g.,~\cite[Section 6]{DBLP:journals/jpdc/DelaetDNT10} declared that this is an open problem. 

There are practically-self-stabilizing algorithms for solving agreement~\cite{DBLP:journals/jcss/DolevKS10,DBLP:conf/netys/BlanchardDBD14}, state-machine replication~\cite{DBLP:conf/netys/BlanchardDBD14,DBLP:journals/corr/DolevGMS15}, and shared memory emulation~\cite{DBLP:conf/podc/BonomiDPR15}. 
None of them considers the studied problem. 
They all rely on synchronization mechanisms, e.g., quorum systems.   
Alon et al.~\cite{DBLP:journals/jcss/AlonADDPT15} and Dolev et al.~\cite[Algorithm~2]{DBLP:journals/corr/DolevGMS15} consider practically-self-stabilizing algorithms that handle counter overflow events using labeling schemes. 
Both algorithms use these labeling schemes together with synchronization mechanisms for implementing shared counters. 
We solve a different problem and propose a practically-self-stabilizing algorithm for vector clocks that uses a labeling scheme but does not use any synchronization mechanism. 
  
\Paragraph{Our Contributions}
We present an important building block for dependable large-scale decentralized systems that need to reason about event causality. In particular, we provide a practically-self-stabilizing algorithm for vector clocks that does not require synchrony assumptions or synchronization mechanisms. Concretely, we present, to the best of our knowledge, the first solution that:

\Subparagraph{(i) Deals with a wide range of failures}
The studied asynchronous systems are prone to crash failures (with or without  undetectable restarts) and communication failures, such as packet omission, duplication, and reordering failures. 

\Subparagraph{(ii) Uses bounded storage and message size}
Our solution considers $3N$ integers and two labels~\cite{DBLP:journals/corr/DolevGMS15} per vector, where $N$ is the number of nodes. Each label has $\bigO(N^3)$ bits. Since all counters share the same two labels, we propose elegant techniques for dealing with the challenge of concurrent overflows (Section~\ref{s:pair}). We overcome the difficulties of making sure that no counter increment is ever ``lost'' even though there is an unbounded period in which these increments are associated with up to $N$ different versions of the vector clock. 



\Subparagraph{(iii) Deals with transient faults and unfair scheduling}
Theorem~\ref{thm:reqHold} proves recovery within $\bigO(N^8\capacity)$ safety violations in a practically-infinite execution in a wait-free manner after the occurrence of transient faults, which is our complexity measure for practically-self-stabilizing systems, where $N$ is the number of nodes in the system and $\capacity$ is an upper bound on the channel capacity.

We believe that our approaches for providing items (i)--(iii) are useful for the design of other practically-self-stabilizing systems.

\Paragraph{Paper organization}
In Section~\ref{s:systemSettings} we present the design criteria.
In Section~\ref{s:back} we give an overview of relevant labeling schemes and in Section~\ref{s:interfaceDolev} we present an interface to a labeling scheme.
Then, we present novel techniques (Section~\ref{s:pair}), upon which we base our algorithm (Section~\ref{s:algorithms}) and proofs (Section~\ref{s:proofs}).


\section{System Settings}
\label{s:systemSettings}

The system includes a set of processors $P  = \{p_1, \ldots, p_N\}$, which are computing and communicating entities that we model as finite state-machines. 
%
%
Processor $p_i$ has an identifier, $i$, that is unique in $P$.
Any pair of active processors can communicate directly with each other via their bidirectional communication channels (of bounded capacity per direction, $\capacity \in \N$, which, for example, allows the storage of at most one message). 
That is, the network's topology is a fully-connected graph and each $p_i \in P$ has a buffer of finite capacity $\capacity$ that stores incoming messages from $p_j$, where $p_j \in P\setminus \{p_i\}$.
%
%
Once a buffer is full, the sending processor overwrites the buffer of the receiving processor.
We assume that any $p_i, p_j \in P$ have access to $channel_{i,j}$, which is a self-stabilizing end-to-end message delivery protocol (that is reliable FIFO) that transfers packets from $p_i$ to $p_j$.
Note that~\cite{DBLP:journals/ipl/DolevDPT11, DBLP:conf/sss/DolevHSS12} present a self-stabilizing reliable FIFO message delivery protocol that tolerates packet omissions, reordering, and duplication over non-FIFO channels. 
%

\Paragraph{The interleaving model}
The processor's program is a sequence of {\em (atomic) steps}. 
Each \emph{step} starts with an internal computation and finishes with a single communication operation, i.e., packet $send$ or $receive$. 
We assume the interleaving model, where steps are executed atomically; one step at a time. 
Input events refer to packet receptions or a periodic timer that can, for example, trigger  the processor to broadcast a message. 
Note that the system is asynchronous and the algorithm that each processor is running is oblivious to the timer rate. 
Even though the scheduler can be adversarial, we assume that each processor's local scheduler is fair, i.e., the processor alternates between completing send and receive operations (unless the processor's communication channels are empty).
Note that a message that a processor $p_i$ needs to send to its neighbors takes $N-1$ consecutive steps of $p_i$ (the execution might include steps of other processors in between), since each step can include at most one send (or receive) operation.

The {\em state}, $s_i$, of $p_i \in P$ includes all of $p_i$'s variables as well as the set of all messages in $p_i$'s incoming communication channels. 
Note that $p_i$'s step can change $s_i$ as well as remove a message from $channel_{j,i}$ (upon message arrival) or queue a message in $channel_{i,j}$ (when a message is sent). 
We assume that if $p_i$ sends a message infinitely often to $p_j$, processor $p_j$ receives that message infinitely often, i.e., the communication channels are fair.
The term {\em system state} refers to a tuple of the form $c = (s_1, s_2, \cdots, s_N)$, where each $s_i$ is $p_i$'s state (including messages in transit to $p_i$). 
We define an {\em execution (or run)} $R={c_0,a_0,c_1,a_1,\ldots}$ as an alternating sequence of system states $c_x$ and steps $a_x$, such that each system state $c_{x+1}$, except for the initial system state $c_0$, is obtained from the preceding system state $c_x$ by the execution of step $a_x$.
%

\Paragraph{Active processors, processor crashes, and undetectable restarts}
At any point and without warning, $p_i$ is prone to a crash failure, which causes $p_i$ to either forever stop taking steps (without the possibility of failure detection by any other processor in the system) or to perform an undetectable restart in a subsequent step~\cite{georgiou2011cooperative}.
In case processor $p_i$ performs an undetectable restart, it continues to take steps by having the same state as immediately before crashing, but possibly having lost the messages that other processors sent to $p_i$ between crashing and restarting.
Processors know the set $P$, but have no knowledge about the number or 
the identities of the processors that never crash. 

We assume that transient faults occur only before the starting system state $c_0$, and thus $c_0$ is arbitrary.
Since processors can crash after $c_0$, the executions that we consider are not fair~\cite{DBLP:books/mit/Dolev2000}.
We illustrate the failures that we consider in this paper in Figure~\ref{fig:failureModel}.

\begin{figure*}[t!]
   \centering
   \includegraphics[scale=0.8]{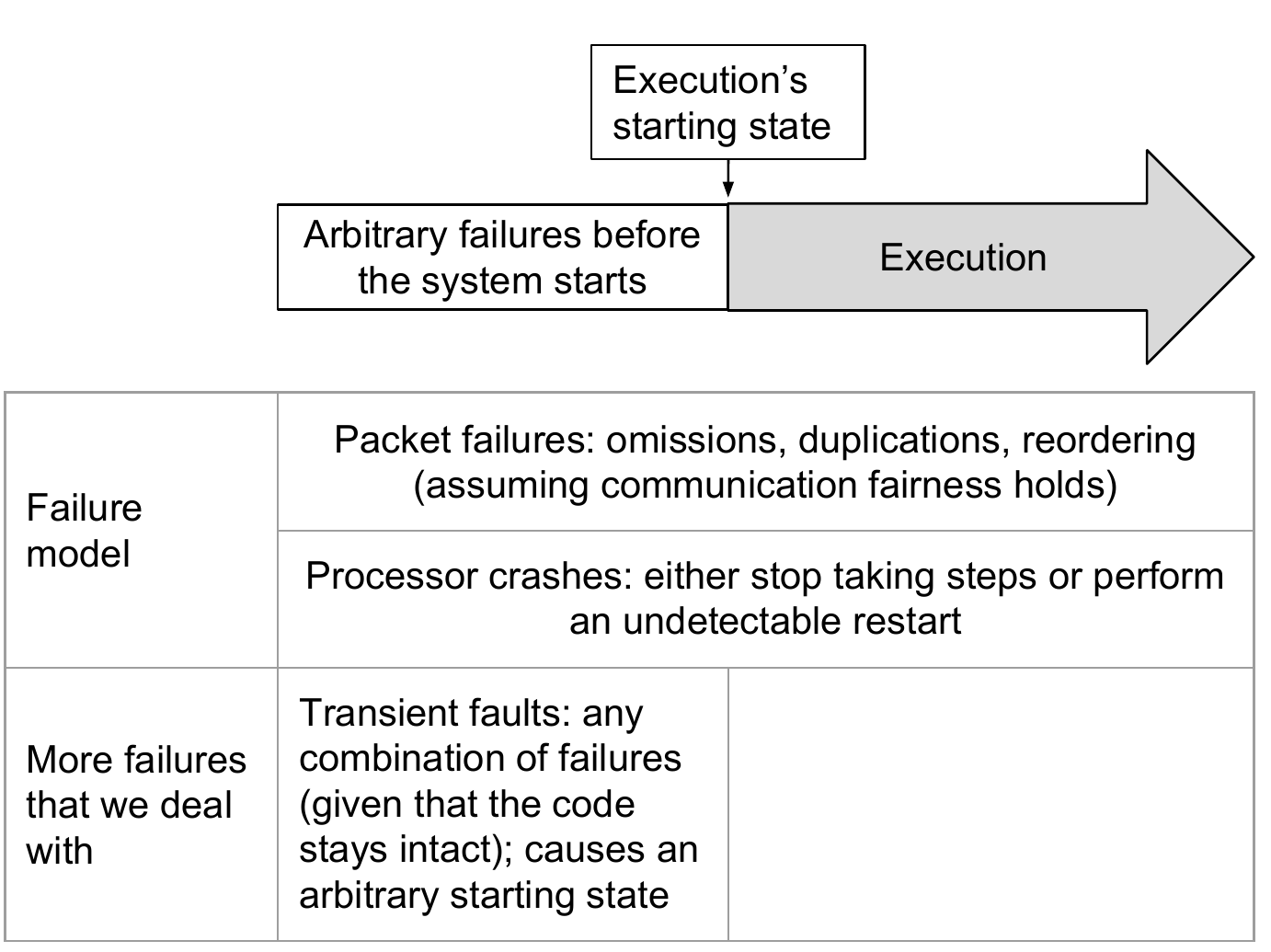}
\caption{Illustration of the failure model and of transient faults.}
\label{fig:failureModel}
\end{figure*}

We say that a processor is active during a finite execution $R'$ if it takes at least one step in $R'$. We say that a processor is active throughout an infinite execution $R$, if it takes an infinite number of steps during $R$. 
Note that the fact that a processor is active during an infinite execution does not give any guarantee on when or how often it takes steps. 
Thus, there might be an arbitrarily long (yet finite)  subexecution $R'$ of $R$, such that a processor is active in $R$ but not in $R'$. 
Therefore, processors that crash and never restart during an infinite execution $R$ are not active throughout $R$.

\Paragraph{Execution operators: concatenation $\circ$ and segment $\seg$} 
%
Suppose that $R'$ is a prefix of an execution $R$, and $R''$ is the remaining suffix of $R$. 
We use the concatenation operator $\circ$ to write that $R=R' \circ R''$, such that $R'$ is a finite execution that starts with the initial system state of $R$ and ends with a step that is immediately followed by the initial state of $R''$.	
We denote by $R' \seg R$ the fact that $R'$ is a subexecution (or segment) of $R$.


\Paragraph{Execution length, practically infinite, and the $\ll$ (significantly less) relation}
To the end of defining the stabilization criteria, we need to compare the number of steps that violate safety in a finite execution $R$ with the length of $R$.
In the following, we define how to compare finite executions and sets of states according to their size.

We say that the length of a finite execution $R=c_0, a_0, c_1, a_1, \ldots, c_{x-1}, a_{x-1}$ is equal to $x$, which we denote by $|R|=x$.
Let $\MI$ be an integer that is considered as a \textit{practically infinite}~\cite{DBLP:journals/jcss/DolevKS10} quantity for a system $\cS$ (e.g., the system's lifetime). 
For example, $\MI$ can refer to $2^b$ (where $b=64$ or larger) sequential system steps (e.g., single send or receive events).
%
In this paper, we use $\ll$  as a formal way of referring to the comparison of, say, $N^c$, for a small integer $c$, and $\MI$, such that $N^c$ is an insignificant number when compared to $\MI$.
Since this comparison of quantities is system-dependent, we give a modular definition of $\ll$ below.

Let $\pinf$ denote a system-dependent quantity that is practically-infinite for a system $\cS$, such that for an integer $z\ll \MI$, we have that $\pinf := z\cdot \MI$.
For a system $\cS$ and $x\in \N$, we denote by $x\ll \pinf$ the fact that $x$ is significantly less than (or insignificant with respect to) $\pinf$.
We say that an execution $R$ is \textit{of $\pinf$-scale}, if there exists an integer 
$y\ll \MI$, such that $|R| = y\cdot \MI$ holds.

\Paragraph{The design criteria of practically-self-stabilizing systems}
\label{s:practicalStabilization}
We define the system's abstract task $\cT$ by a set of variables (of the processor states) and constraints, which we call the \textit{system requirements}, in a way that defines a desired system behavior, but does not consider necessarily all the implementation details. 
We say that an execution $R$ is a \textit{legal execution} if the requirements of task $\cT$ hold for all the processors that take steps during $R$ (which might be a proper subset of $P$).
We denote the set of legal executions with  $\LE$.
We denote with $f_R$ the number of deviations from the abstract task in an execution $R$, i.e., the number of states in $R$ in which the task requirements do not hold (hence $R\in\LE \iff f_R = 0$).
Note that the definition of $\LE$ allows executions of very small length, but our focus will be on finding maximal subexecutions $R^*\seg R$ for a given $\pinf$-scale execution $R$, such that $R^*\in\LE$.
\modified{Definitions~\ref{def:strongSelfStab},~\ref{def:pseudoSelfStab} and~\ref{def:practSelfStab} specify our stabilization criteria.}{Definitions~\ref{def:strongSelfStab} and~\ref{def:practSelfStab} specify our stabilization criteria.}

\begin{definition}[Strong Self-stabilization]
\label{def:strongSelfStab}
For every infinite execution $R$, there exists a partition $R = R' \circ R''$, such that $|R'| = z(N) \in \N$ and $f_{R''} = 0$, where $z(N)$ is the complexity measure.
\end{definition}

\begin{definition}[Pseudo Self-stabilization]
\label{def:pseudoSelfStab}
For every infinite execution $R$, $f_R = f(R,N) \in \N$, where $f_R$ is the complexity measure.
\end{definition}

\begin{definition}[Practically-self-stabilizing System]
\label{def:practSelfStab}
For every infinite execution $R$, and for every $\pinf$-scale subexecution $R'$ of $R$, $f_{R'} = f(R',N) \ll |R'|$, where $f_{R'}$  is the complexity measure.
\end{definition}






\Paragraph{Problem definition (task requirement)}
We present a requirement (Requirement\reqs) which defines the abstract task of vector clocks.
This requirement trivially holds for a fault-free system that can store unbounded values (and thus does not need to deal with integer overflow events). 
The presence of transient faults
can violate these assumptions and cause the system to deviate from the abstract task, which $\LE$ specifies (through Requirement\reqs).
In the following, we present Requirement\reqs\ and its relation to causal ordering (Property~\ref{req:3causality}). 
%

We assume that each processor $p_i$ is recording the occurrence of a new local event by incrementing the $i$-th entry of its vector clock.
%
During a legal execution, we require that the processors count all the events occurring in the system, despite the (possibly concurrent) wrap around events.
Hence, we require that 
\remove{for every vector clock that is received by an active processor, the values of that vector clock or higher ones should be received within the legal execution by all active processors (Requirement~\ref{req:correctCounting}), 
while}
 the vector clock element of each (active) processor records all the increments done by that processor (Requirement~\ref{req:2act}).
As a basic functionality, we assume that each processor can always query the value of its local vector clock. 
We say that an execution $R^*$ is a legal execution, i.e., $R^*\in \LE$, if Requirement~\ref{req:2act} holds for the states of all processors that take steps during $R^*$.


%

\begin{requirement}[Counting all \remove{relevant} events]
\label{req:2act} 
Let $R$ be an execution, $p_i$ be an active processor, and $V^k_i[i]$ be $p_i$'s \remove{\textit{relevant}} value in $c_k \in R$. 
For every active processor $p_i\in P$, the number of $p_i$'s \remove{\textit{relevant}} counter increments between the states $c_k$ and $c_\ell \in R$ is $V^\ell_i[i] - V^k_i[i]$, where $c_k$ precedes $c_\ell$ in $R$.
\end{requirement}

\Subparagraph{Causal precedence}
We explain how faults and bounded counter values affect Requirement\reqs\ and causal ordering.
Let $V$ and $V'$ be two vector clocks, and $\causalPrecedence(V, V')$ be a query which is true, if and only if, $V$ causally precedes $V'$, i.e., $V'$ records all the events that appear in $V$~\cite{tanenbaum2007distributed}.
Then, $V$ and $V'$ are concurrent when $\neg \causalPrecedence(V, V')\land \neg \causalPrecedence(V', V)$ holds.
We formulate the causal precedence property in Property~\ref{req:3causality}.
In a fault-free system with unbounded values, Requirement~\ref{req:2act} trivially holds, since no wrap around events occur.
That is, $V$ causally precedes $V'$, if $V[i] \leq V'[i]$ for every $i\in \{1,\ldots,n\}$ and $\exists_{j\in \{1,\ldots,n\}} V[j] < V'[j]$ hold~\cite{tanenbaum2007distributed}.
However, this is not the case in an asynchronous, crash-prone, and bounded-counter setting, where counter overflow events can occur.
We present cases where Requirement\reqs\ and Property~\ref{req:3causality} do not hold due to a counter overflow event in Example~\ref{eg:overflow}.


\begin{property}[Causal precedence]
\label{req:3causality}
For any two vector clocks $V_i$ and $V_j$ of two processors $p_i, p_j\in P$, $\causalPrecedence(V_i, V_j)$ is true if and only if $V_i$ causally precedes $V_j$.
\end{property}

\begin{example}
\label{eg:overflow}
Consider two bounded vector clocks  $V_i = \la v_{i_1},\ldots, v_{i_N} \ra$ and $V_j = \la v_{j_1},\ldots, v_{j_N} \ra$ of $p_i, p_j\in P$, such that upon a new event $p_k\in P$ increments $V_k[k]$ by adding $1\,\mmi$. 
Assume that $V_i = V_j$ and $V_i[i] = V_j[i] = \MI - 1$ hold (e.g., as an effect of a transient fault).
In the following step $p_i$ increments $V_i[i]$ by 1, thus $V_i[i]$ wraps around to $V_i[i]=0$, while $V_j[i] = \MI-1$ remains.
Then, $V_i[i]$ mistakenly indicates zero events for $p_i$ ($V_i[i] = 0$) instead of $\MI$, i.e., Requirement\reqs\ does not hold.
Also, using the definition of causal precedence in fault-free systems and unbounded counters, $V_i$ appears to causally precede $V_j$, which is wrong, since $V_j$ causally precedes $V_i$ ($p_i$ had one more event than what $p_j$ records).
That is, $V_i[k] = V_j[k]$ for $k\neq i$ and $V_i[i] = 0 < \MI-1 = V_j[i]$, which mistakenly indicates that $p_j$ records $\MI-1$ more events than $p_i$.\qed
%
\end{example}

We remark that Requirement\reqs\ is a necessary and sufficient condition for Property~\ref{req:3causality} to hold.
Suppose that Requirement\reqs\ does not hold, which means that it is not possible to count the events of a single processor between two states (e.g., as we showed in the previous example).
This implies that it is not possible to compare two vector clocks, hence Property~\ref{req:3causality} cannot hold.
Moreover, if Requirement\reqs\ holds, then it is possible to compare how many events occurred in a single processor between two states, and by extension it is possible to compare all vector clock entries for two vector clocks.
The latter is a sufficient condition for defining causal precedence (as in the fault-free unbounded-counter setting~\cite{tanenbaum2007distributed}).

In Section~\ref{s:pair} we present our solution for computing $V^\ell_i[i] - V^k_i[i]$ for Requirement\reqs\ ($c_\ell,\, c_k$ are states in an execution $R$ and $p_i\in P$) and $\causalPrecedence(V_i, V_j)$ for Property~\ref{req:3causality} in a legal execution. 
In Section~\ref{s:algorithms} we present an algorithm for replicating vector clocks in the presence of faults and bounded-counters, which we prove to be practically-self-stabilizing in Section~\ref{s:proofs}.

%

\section{Background: Practically-self-stabilizing Labeling Schemes}
\label{s:back} 
%

In this section we give an overview of labeling schemes that can be used for designing an algorithm that guarantees Requirement\reqs.
%
%
It is evident from Example~\ref{eg:overflow} (Section~\ref{s:systemSettings}) that a solution for comparing vector clock elements that overflow can be based on associating each vector clock element with a timestamp (or label, or epoch).
This way, even if a vector clock element overflows, it is possible to maintain order by comparing the timestamps.

As a first approach for providing these timestamps, one might consider to use an integer counter (or sequence number), $cn$.
We explain why this approach is not suitable in the context of self-stabilization.
Any system has memory limitations, thus a single transient fault can cause the counter to quickly reach the memory limit, say $\MI$.
The event of counter overflow occurs when a processor increments the counter $cn$, causing $cn$ to encode the maximum value $\MI$. 
In this case, the solution often is that $cn$ wraps around to zero.
Thus, this approach faces the same ordering challenges with the vector clock elements.

Existing solutions associate counters with \textit{epochs} $\ell$, which mark the period between two overflow events. 
A non-stabilizing representation of epochs can simply consider a, say, $64$-bit integer. 
Upon the overflow of $cn$, the algorithm increments $\ell$ by one and nullifies $cn$. 
The order among the counters is simply the lexicographic order among the pairs $\la\ell, cn\ra$.  
With this approach, it is a challenge to maintain an order within a set of integers during phases of concurrent wrap around events at different processors. 
In the following we present more elegant solutions for bounded labeling schemes, that tolerate concurrent overflow events, transient faults, and the absence of execution fairness.


\Paragraph{Bounded labeling schemes}
Bounded labeling schemes (initiated in~\cite{israeli1987bounded, dolev1997bounded}, cf.~\cite[Section 2]{DBLP:books/daglib/0020056}) provide labeling of data and denote temporal relations.
Given a bounded set of labels $L$, a bounded labeling scheme usually includes a partial or total order $\prec_L$ over $L$ and a function for constructing locally a new maximal label from $L$ with respect to $\prec_L$, given a set of input labels.
Labeling algorithms handle these labels such that the processors eventually agree, for example, on a maximal label.
Since we consider processor crashes, a suitable labeling scheme should include a garbage collection mechanism that cancels obsolete labels, by possibly using label storage.

\Paragraph{Practically-self-stabilizing bounded labeling schemes}


Alon et al.~\cite{DBLP:journals/jcss/AlonADDPT15} and Dolev et al.~\cite{DBLP:journals/corr/DolevGMS15} present practically-self-stabilizing bounded-size labels. 
Whenever a counter $cn$ reaches $\MI$, the algorithm by Alon et al.~\cite{DBLP:journals/jcss/AlonADDPT15}  replaces its current label $\ell$ with $\ell'$, which at the moment of this replacement is greater than any label that appears in the system state. 
This means that immediately after the counter wraps around, the counter $\la\ell', cn=0\ra$ is greater than all system counters. 
In the remainder of this section, we give an overview of the labeling schemes of Alon et al.~\cite{DBLP:journals/jcss/AlonADDPT15} (Section~\ref{s:Alon}) and Dolev et al.~\cite{DBLP:journals/corr/DolevGMS15} (Section~\ref{s:DolevAbstractTask}). 

\subsection{The case of no concurrent overflow events}
\label{s:Alon}
Alon et al.~\cite{DBLP:journals/jcss/AlonADDPT15} address the challenge of always being able to introduce a label that is greater than any other previously used one. 
They present a two-player guessing game, between a finder, representing the algorithm, and a hider, representing an adversary controlling the asynchronous system that starts from an arbitrary state. 
Let $M$ be the maximum number of labels that can exist in the communication channels, i.e., $M = \capacity N(N-1)$, where $N(N-1)/2$ is the number of bidirectional communication channels in the system and $\capacity$ is the capacity in number of messages per channel (and hence labels).

The hider has a bounded size label set, $\mathcal{H}$, such that $|\mathcal{H}|\leq \msg \in \N$. 
The finder, who is oblivious to $\mathcal{H}$'s content, aims at obtaining a label $\ell$ that is greater than all of $\mathcal{H}$'s labels. 
To that end, the finder generates $\ell$ in such a way that whenever the hider exposes a label $\ell' \in \mathcal{H}$, such that 
$\ell$ is not greater than $\ell'$,
$\mathcal{H}$ has one less label that the finder is unaware of its existence. 
The hider may choose to include $\ell$ in $\mathcal{H}$ as long as it makes sure that $|\mathcal{H}|\leq \msg$ by omitting another label from $\mathcal{H}$ (without notifying the finder).

\Paragraph{Label construction}
%
A label component $\ell = (sting, Antistings)$ is a pair, where $sting \in D$, $D = \{1, \ldots, k^{2}+1\}$, $Antistings \subset D$, 
$|Antistings|=k$, and $k>1$ is an integer.
The order among label components is defined by the relation $\lea$, where $\ell_{i} \lea \ell_{j} \iff (\ell_i.sting \in \ell_j.Antistings) \land (\ell_j.sting \not \in \ell_i.Antistings)$.
%
%
The function $Next_b(L)$ takes a set $L=\{ \ell_1,\ldots, \ell_{\kappa}\}$ of (up to) $k\in \N$ label components, and returns a newly created label component, $\ell_{j}=\la s,A\ra$, such that $\forall \ell_{i} \in L: \ell_{i} \lea \ell_{j}$,
%
%
where $s \in D\setminus \cup_{i=1}^{\kappa} A_i$ and $A = \{s_1, \ldots, s_{\kappa} \}$, possibly augmented by arbitrary elements of $D\setminus A$ when $|A| = \kappa <k$.


\Paragraph{Label cancelation}
%
%
Alon et al.~\cite{DBLP:journals/jcss/AlonADDPT15} use the order $\lea$ for which, during the period of recovery from transient faults, it can happen that $\ell_1$, $\ell_2$, and $\ell_3$ appear in the system and $\ell_1 \lea \ell_2 \lea \ell_3 \lea \ell_1$ holds. 
The finder breaks such cycles by \textit{canceling} these label components so that the system (eventually) avoids using them.
Alon et al.~\cite{DBLP:journals/jcss/AlonADDPT15} implement labels (epochs) as pairs $(ml,cl)$, where $ml$ is always a label component and $cl$ is either $\bot$ when $(ml,cl)$ is legitimate (non-canceled) or a label component for which $cl \not \lea ml$ holds. 
Thus, the finder stores $cl$ as an evidence of $ml$'s cancelation.

\Paragraph{Keeping the number of stored labels bounded} 
Alon et al.~\cite{DBLP:journals/jcss/AlonADDPT15} present a finder strategy that queues the most recent labels that the finder is aware of in a FIFO manner. 
They show a $2\msg$ bound on the queue size by pointing out that, if the finder queues (1) any label that it generates and (2) the ones that the hider exposes, the hider can surprise the finder at most $\msg$ times before the queue includes all the labels in $\mathcal{H}$.

In detail, the algorithm gossips repeatedly its (currently believed) greatest label and stores the received ones in a queue of at most $2\msg$ labels, where $\msg = \capacity N(N-1)$ is the maximum number of labels that the system can ``hide'', i.e., one (currently believed) greatest label that each of the $N$ processors has and $\capacity$ (capacity) in each communication link. 
Upon arrival of label $\ell_{i}$ to $\ell_{j}$ such that $\ell_{i} \not\lea \ell_{j}$, processor $p_j$ queues the arriving label $\ell_{i}$, uses $Next_b()$ to create a new (currently believed) greatest legitimate label $\ell'_{j}$ and queue it as well. 
Alon et al.~\cite{DBLP:journals/jcss/AlonADDPT15} show that, when only $p_j$ may create new labels, the system stabilizes to a state in which $p_j$ believes in a legitimate label that is indeed the greatest in the system. 
Note that the stabilization period includes at most $\msg$ arrivals to $p_j$ of labels $\ell_{i}$ that ``surprise'' $p_j$, i.e., $\ell_{i} \not\lea \ell_{j}$.
\subsection{The case of concurrent overflow events}
\label{s:DolevPaper}
Alon et al.~\cite{DBLP:journals/jcss/AlonADDPT15}'s labels allow, once a single label (epoch) $\ell$ is established, to order the system events using the counter $(\ell, cn)$.
Dolev et al.~\cite{DBLP:journals/corr/DolevGMS15} extend Alon et al.~\cite{DBLP:journals/jcss/AlonADDPT15} to support concurrent $cn$ overflow events, 
by including the label creator identity. 
This information facilitates symmetry breaking, and decisions about which label is the most recent one, even when more than one creator concurrently constructs a new label. 
Dolev et al.~\cite{DBLP:journals/corr/DolevGMS15} make sure that active processors $p_i$ remove eventually obsolete labels $\ell$ that name $p_i$ as their creator (due to the fact that $p_i$ indeed created $\ell$, or $\ell$ was present in the system's  arbitrary starting state). 
Note that the system's arbitrary starting state may include cycles of legitimate (not canceled) labels $\ell_1 \lea \ell_2 \lea \ell_3 \lea \ell_1$ that share the same creator, e.g., $p_k$. 
The algorithm by Dolev et al. guarantees cycle breaking by logging all labels that it observes and canceling any label that is not greater than its currently known maximal one.

\Paragraph{Label construction}
Dolev et al.~\cite{DBLP:journals/corr/DolevGMS15} extend Alon et al.'s label component to $(\creator, sting$, $Antistings)$, where $\creator$ is the identity of the label creating processor, and $sting$ as well as $Antistings$ are as in~\cite{DBLP:journals/jcss/AlonADDPT15} (Section~\ref{s:Alon}). 
They use $=_{lb}$ to denote that two labels, $\ell_{i}$ and $\ell_{j}$, are identical and define the relation $\ell_{i} \lb \ell_{j} \iff (\ell_i.\creator < \ell_j.\creator) \lor (\ell_{i}.\creator = \ell_j.\creator \land ((\ell_i.sting \in \ell_j.Antistings) \land (\ell_j.sting \not \in \ell_i.Antistings)))$. The labels $\ell_i$ and $\ell_j$ are \textit{incomparable} when  $\ell_i \nled \ell_j \land \ell_j \nled \ell_i$ (and comparable otherwise).



\Paragraph{Label cancelation}
Dolev et al. consider label $\ell$ to be obsolete when there exists another label $\ell' \not \lb \ell$ of the same creator. In detail, $\ell_i$ cancels $\ell_j$, if and only if, $\ell_i$ and $\ell_j$ are incomparable, or if $\ell_i.\creator = \ell_j.\creator \land \ell_i.sting \in \ell_j.Antistings \land \ell_j.sting \notin \ell_i.Antistings$, i.e., $\ell_i$ and $\ell_j$ have the same creator but $\ell_j$ is greater than $\ell_i$ according to the $\lea$ order.


\Paragraph{The abstract task of Dolev et al.'s labeling scheme}
\label{s:DolevAbstractTask}
Each processor presents to the system a label that represents the \textit{locally perceived maximal label}. 
During a legal execution, as long as there is no explicit request for a new label, 
all processors refer to the same locally perceived maximal label, which we refer to as the \textit{globally perceived maximal label}. 
Moreover, it cannot be the case that processor $p_i$ has a locally perceived maximal label $\ell_i$ and another processor $p_j\in P$ (possibility $i=j$) stores a label $\ell_j$ that is incomparable to $\ell_i$, greater than $\ell_i$, or that cancels $\ell_i$ (where $\ell_j$ is not necessarily $p_j$'s locally perceived maximal label). 

We note that when the system starts in an arbitrary state, the (active) processors might refer to a globally perceived maximal label that is not the maximal label in the system. 
This is due to the fact that in practically-self-stabilizing systems there could be a non zero number of deviations from the abstract task during any practically infinite execution (Definition~\ref{def:practSelfStab}). 
In detail, Dolev et al.~\cite[Algorithm~2]{DBLP:journals/corr/DolevGMS15} store the locally perceived maximal label of processors $p_i$ at $max_i[i]$ and demonstrate the satisfaction of the above abstract task in~\cite[Theorem~4.2]{DBLP:journals/corr/DolevGMS15}.

\Paragraph{Keeping the number of stored labels bounded} 
%
%
Whenever $p_k$ is active, it will eventually queue all of the labels that it has created, cancel them, and generate a label that is greater than them all. However, in case that $p_k$ is inactive, the algorithm uses the active processors to prevent the asynchronous system from endlessly using labels that belong to cycles. 
Dolev et al.~\cite{DBLP:journals/corr/DolevGMS15} show that $p_k$'s cycle may include at most $\msg+N$ labels, where $\msg$ is the number of labels that can appear in the communication channels and $N$ is the number of processors. Therefore, $p_i \in P$ needs to queue $\msg+N$ labels for any other $p_k$, so that $p_i$ could remove the label cycles once their creator $p_k$ becomes inactive. Moreover, $p_i$ needs a queue of $2(\msg N + 2N^2 - 2N) + 1$ labels $\ell$ (for which $\ell.\creator=i$) until it can be sure to have the maximal label.
These bounds give the maximum number of labels that $p_i$ can either \textit{adopt} (use as its maximal label) or \textit{create} when it does not store a maximal label, throughout any execution.
%
Next, we provide the algorithm details and use these details when justifying our bounds (Section~\ref{s:interfaceDolev}).

\PARAgraph{Variables}
%
%
Each processor $p_i$ maintains an $N$-size vector of labels $max_i$, where $max_i[i]$ is $p_i$'s local maximal $\lb$-label and $max_i[j]$ is the latest \textit{legitimate}, i.e., not canceled, label that $p_i$ received most recently from $p_j$.
Also, $p_i$ maintains an $N$-size vector $storedLabels_i$ of queues that logs the labels that $p_i$ has observed so far, which $p_i$ sorts by their label creator.
That is, $storedLabels_i[j]$ queues label $\ell$, such that 
(i) $p_i$ has received $\ell$ from an arbitrary processor,
(ii) $\ell$'s creator is $p_j \in P$, i.e., $\ell.\creator = j$,
(iii) there are no duplicates of $\ell$ in $storedLabels_i[j]$, and
(iv) $\ell$ is either canceled or every other label in $storedLabels_i[j]$ is canceled.
%


\PARAgraph{The algorithm} 
\remove{
\IS{mention label adoption}}
Processor $p_i$ gossips repeatedly its $\lb$-greatest label, $max_i[i]$, and stores the arriving labels in $storedLabels_i[j]$, where $j=\ell.\creator$.
The algorithm ensures that $storedLabels_i[j]$ stores at most one legitimate label by canceling any label $\ell'$ using label $\ell$ when (1) they are incomparable or (2) they share the same creator and $\ell' \lb \ell$. 
Moreover, $p_i$ makes sure that, for any $j\in [1,N]$, the $\lb$-greater label $max_i[j]$ is indeed greater than any other label in $storedLabels_j[i]$. 
%
%
In case it does not, $p_i$ selects the $\lb$-greatest legitimate label in $storedLabels_i$, and if there is no such legitimate label, $p_i$ creates a new label via $nextLabel()$, which is an extension of $Next_b()$ that also includes the label creator, $p_i$.
Dolev et al.~\cite{DBLP:journals/corr/DolevGMS15} bound the size of $storedLabels_i[j]$, for $j\neq i$ by $N + \msg$ and the size of $storedLabels_i[i]$ by $2(\msg N + 2N^2 - 2N) + 1$.

\remove{ 
\IS{We should add the size of $k$ here too. Here, $k\geq 2|storedLabels_i[i]|$, since any processor $p_i$ may store at most 
$|storedLabels_i[i]|$ labels $\ell$ such that $\ell.\creator = i$ and each label includes an $ml$ and a $cl$ part.} \EMS{Please write the text itself.}
} 



\section{Composing practically-self-stabilizing labeling algorithms and the interface to Dolev et al.~\cite{DBLP:journals/corr/DolevGMS15} labeling scheme}
%
\label{s:interfaceDolev}
%
%

In this section we present a framework for composing any practically-self-stabilizing labeling algorithm (server) with any other practically-self-stabilizing algorithm (client). 
By this composition we obtain a compound algorithm with combined properties.
Then, we discuss the challenges in composing practically-self-stabilizing algorithms, with respect to the composition of strong self-stabilizing algorithms.
Moreover, we present an interface to a labeling algorithm that facilitates our composition approach.
The interface is also used by the client algorithm to query the state of the labeling algorithm, send messages, or to request the labeling algorithm to cancel a label.
We show how this interface is implemented by the practically-self-stabilizing labeling scheme of Dolev et al.~\cite[Algorithm 2]{DBLP:journals/corr/DolevGMS15}, which we use in our solutions (Section~\ref{s:pair}) and algorithm (Section~\ref{s:algorithms}).
%
%
We end the section by discussing the stabilization guarantees of the compound algorithm.

\Paragraph{Composition with a practically-self-stabilizing labeling algorithm} 
%
We follow an approach for algorithm composition in message passing systems that resembles the one in~\cite[Section 2.7]{DBLP:books/mit/Dolev2000}, which considers a composition of two self-stabilizing algorithms. 
Let us name these two algorithms as the server and client algorithms. 
The server algorithm provides services and guaranteed  properties that the client algorithm uses. 
In the composition presented in~\cite[Section 2.7]{DBLP:books/mit/Dolev2000}, once the server algorithm stabilizes, the client algorithm can start to also stabilize. 
This way, the compound algorithm obtains more complex guarantees than the individual algorithms. 

We detail our composition approach which we illustrate in Figure~\ref{fig:composition}. 
In the following, we refer to the computations of a step excluding the send or receive operation, as the step's \textit{invariant check}, which possibly includes updates of local variables.
Our approach for composing practically-self-stabilizing algorithms assumes that the messages of the client algorithm are piggybacked by the ones of the server, and that the server algorithm can send any message independently. 
Also, we assume that the communication among processors relies on a self-stabilizing end-to-end protocol, such as the ones in~\cite{DBLP:journals/ipl/DolevDPT11, DBLP:conf/sss/DolevHSS12}.

\begin{figure}[t!]
   \centering
   \includegraphics[scale=0.65]{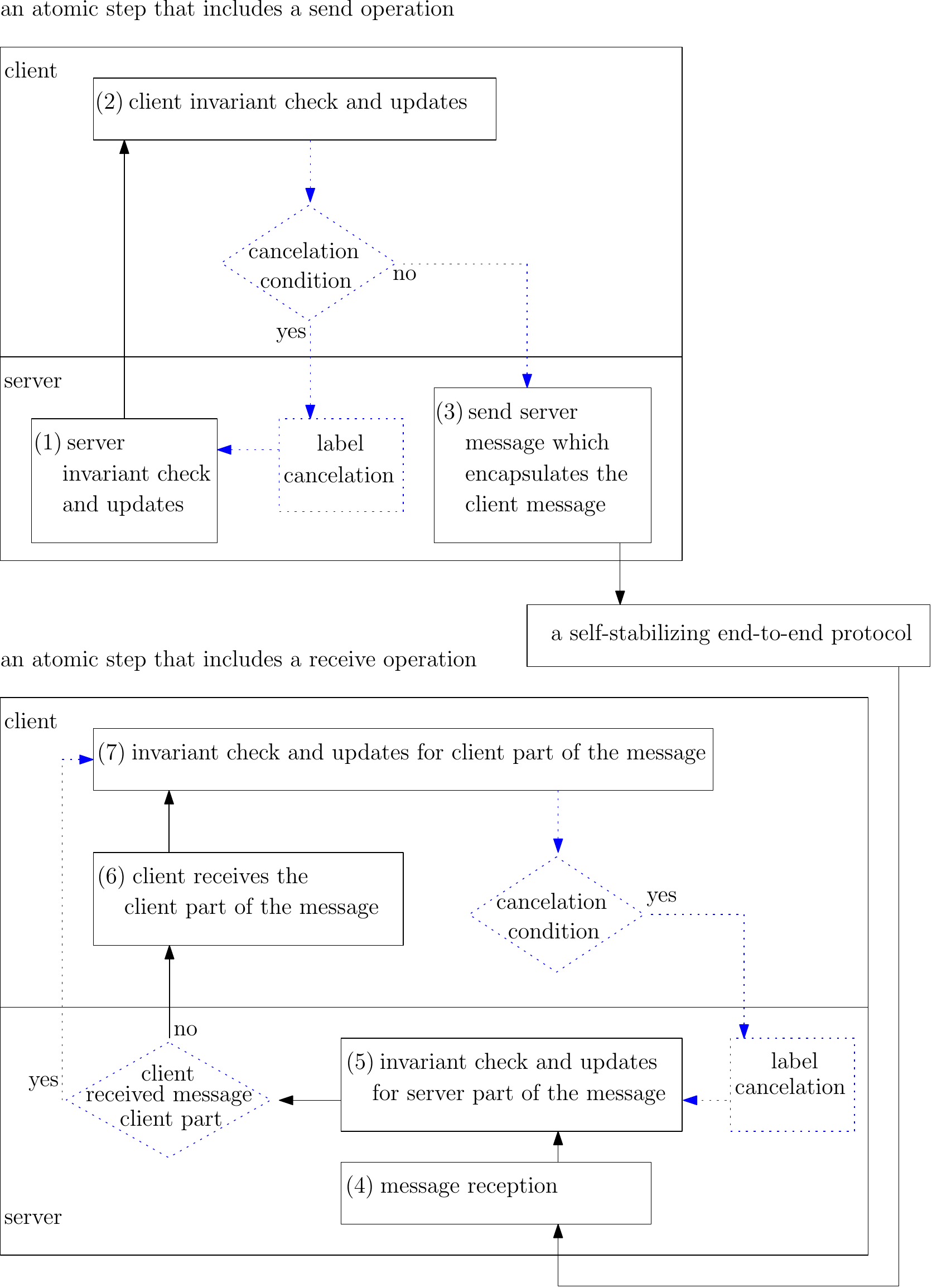}
\caption{Composition of the server and the client algorithms.
The normal lines denote the composition parts that are common in both strong and practically-self-stabilizing algorithms.
The dotted blue lines show the computations in the composition of practically-self-stabilizing algorithms, that are additional to the normal lines.
We refer to the computations done in a step excluding the send or receive operation as the (server or client) invariant check and updates.}
\label{fig:composition}
\end{figure}

\Subparagraph{A step that includes a send operation} We first explain the computations of the compound algorithm during a step that ends with a send operation.
This step starts with the server algorithm's invariant check and updates, which is followed by the client algorithm's invariant check and updates (parts 1 and 2 of Figure~\ref{fig:composition}, respectively).
We assume that the client algorithm can request a change in the labeling (server) algorithm's state, e.g., a label cancelation, but this change is performed by the labeling algorithm (cf. label cancelation in Figure~\ref{fig:composition}).
In case the client algorithm indeed requires a label to be canceled, the labeling algorithm cancels that label, and then the server and client invariant check and updates repeat (cf. Figure~\ref{fig:composition}).
Otherwise, the server encapsulates the client's message, $m_{client}$, and transmits the server message $m_{server}=\langle serverPart, m_{client} \rangle$, which encodes the server and client parts of the message.

\Subparagraph{A step that includes a receive operation} 
Upon the arrival of a message $m$ by the labeling (server) algorithm (part 4 in Figure~\ref{fig:composition}), the server algorithm performs the server invariant check and updates on the server part of the message (part 5 in Figure~\ref{fig:composition}).
Then, the server algorithm raises a message reception event for the client algorithm (part 6 in Figure~\ref{fig:composition}), which delivers the part of the arriving message that is relevant to the client algorithm, i.e., $m_{client}$.
In the following, the client algorithm performs the client invariant check and updates (part 7 in Figure~\ref{fig:composition}), which might include a request to change the state of the server algorithm, e.g., by canceling a label.
If that is the case, the labeling algorithm cancels that label, and then parts 5 and 7 of Figure~\ref{fig:composition} repeat.


\Paragraph{Challenges in composition of practically-self-stabilizing algorithms}
Composing practically-self-stabilizing algorithms is not always identical to composing strong self-stabilizing algorithms (cf.~\cite[Section 2.7]{DBLP:books/mit/Dolev2000}).
An infinite execution is \textit{fair}~\cite{DBLP:books/mit/Dolev2000} if all processors take steps infinitely often (hence no processor crashes).
In this paper we allow processor crashes, i.e., the executions are not fair, in contrast to strong self-stabilization.
Moreover, when composing two self-stabilizing algorithms, we assume that 
the client algorithm does not change the state of the server algorithm.
However, labels can become obsolete (canceled), e.g., due to an overflow event of a counter in the client algorithm.
Thus, a step of the client algorithm might include requesting the labeling (server) algorithm to change its state by canceling a label (cf. Figure~\ref{fig:composition}).
%

\Paragraph{An interface to a labeling algorithm and its implementation by the labeling algorithm of Dolev et al.~\cite[Algorithm 2]{DBLP:journals/corr/DolevGMS15}} 
We detail an interface to a labeling (server) algorithm in order to facilitate composition with a client algorithm.
%
The functions of the interface allow the labeling algorithm to do its invariant check and updates.
They also allow the client algorithm to query the state of the labeling algorithm without changing the labeling algorithm's state, except for the function $cancel()$.
The function $cancel()$ changes the labeling algorithm's state by canceling a label.
Moreover, we explain how the Dolev et al. labeling algorithm~\cite[Algorithm 2]{DBLP:journals/corr/DolevGMS15} (cf. Section~\ref{s:DolevPaper}) implements the functions of this interface.
These functions are also used by the shared counter algorithm in~\cite[Algorithm 3]{DBLP:journals/corr/DolevGMS15}.
We note that~\cite[Algorithm 3]{DBLP:journals/corr/DolevGMS15} relies on synchronization mechanisms, but the labeling algorithm~\cite[Algorithm 2]{DBLP:journals/corr/DolevGMS15} does not rely on synchronization mechanisms, and hence it is suitable for our solution.

\PARAgraph{$\bullet$ $labelBookkeeping()$: server invariant check and updates} 
This function allows the labeling algorithm to perform its invariant check and updates, i.e., the step's computations excluding the send or receive operation (part 1 or parts 4 and 5 in Figure~\ref{fig:composition}).
It is intended to be called in every step of the client algorithm, and thus facilitates the composition of the two algorithms.

In \cite[Algorithm 2]{DBLP:journals/corr/DolevGMS15}, when calling $labelBookkeeping()$ without arguments, \cite[Algorithm 2, lines 21 to 28]{DBLP:journals/corr/DolevGMS15} perform the server invariant check and updates (part 1 of Figure~\ref{fig:composition}). 
When calling $labelBookkeeping(m,j)$, \cite[Algorithm 2, lines 19 to 28]{DBLP:journals/corr/DolevGMS15} process a message $m$ that arrived from a processor $p_j \in P$ (parts 4 and 5 of Figure~\ref{fig:composition}).
The (mutable) function $labelBookkeeping()$ is an alias to $process()$~\cite[Algorithm 3, line 2]{DBLP:journals/corr/DolevGMS15}. 

\PARAgraph{$\bullet$ $isStored()$ and $isCanceled()$: querying whether a label is stored or canceled} 
%
Given a label $\ell$, the (immutable, i.e., its value cannot change) predicate $isStored(\ell)$ checks whether $\ell$ appears in the label storage of the labeling algorithm. 
The (immutable) predicate $isCanceled(\ell)$ checks whether $\ell$ is canceled. 

In \cite[Algorithm 2]{DBLP:journals/corr/DolevGMS15}, $isStored(\ell)$ returns true, if and only if $\ell \in storedLabels[j]$, such that $p_j\in P$ is $\ell$'s creator. 
Also, $isCanceled()$ is an exact alias to $legit()$ in~\cite[Algorithm 2, line 6]{DBLP:journals/corr/DolevGMS15}.  

%
%

\PARAgraph{$\bullet$ $\getLabel()$: retrieving the largest label} The (immutable) function $\getLabel()$ returns the largest locally stored label.

In \cite[Algorithm 2]{DBLP:journals/corr/DolevGMS15}, $\getLabel()$ returns the largest locally stored label with respect to the partial order of labels $\lb$ (cf. Section~\ref{s:DolevPaper})\modified{}{, where $\lb$ is the relation among labels that Dolev et al. define~\cite{DBLP:journals/corr/DolevGMS15}}. 
\short{In detail, that label is stored in $max_i[i]$ (cf. lines 27 and 28 of~\cite[Algorithm 2]{DBLP:journals/corr/DolevGMS15}).} 

\PARAgraph{$\bullet$ $\legitArriving()$ and $encapsulate()$: Token circulation and message encapsulation} 
The $\legitArriving()$ function enables a token circulation mechanism for the labeling algorithm, which is part of the self-stabilizing end-to-end protocol (cf. Figure~\ref{fig:composition}).
The token circulation mechanism guarantees that for two processors $p_i, p_j\in P$, $p_i$ processes an incoming message from $p_j$ only if $p_j$ has received $p_i$'s local maximal label.
To that end, $p_j$ piggybacks the last received maximal label of $p_i$, $sentMax$, to every message $m_{server}$ that it sends to $p_i$.
Moreover, the function $encapsulate()$ facilitates piggybacking of a message of the labeling algorithm with the one of the composed (client) algorithm (facilitating part 3 of Figure~\ref{fig:composition}). 
That is, the (immutable) function $encapsulate(m_{client})$ returns a message $m_{server}$, such that the labeling (server) algorithm's message encapsulates the message $m_{client}$. 

Let $serverPart = \langle sentMax, \bullet \rangle$ and $m_{client}$ be the server, and respectively, the client part of an outgoing message of the compound algorithm.
In \cite[Algorithm 2]{DBLP:journals/corr/DolevGMS15}, the server message is $m_{server}=\langle \langle sentMax, \bullet \rangle, m_{client} \rangle$~\cite[Algorithm 2]{DBLP:journals/corr/DolevGMS15} and $encapsulate(m_{client})$ returns the value $m_{server}$.
Moreover, the (immutable) function $\legitArriving(m_{server}, \ell)$ tests the consistency of an arriving label $\ell$ with $serverPart$ of the server message $m_{server}$. 
That is, the predicate $\legitArriving(m_{server}, \ell)$ returns the value of $\ell=sentMax$.

\PARAgraph{$\bullet$ $cancel()$: canceling a label} 
This is a function that the client algorithm uses to request the labeling algorithm to cancel a label, e.g., upon an overflow event.
In contrast to the functions presented above, $cancel()$ is the only function that the client algorithm can use to change the state of the labeling (server) algorithm (cf. Figure~\ref{fig:composition}).
Let $\ell$ and $\ell'$ be two labels, such that $\ell'$ cancels $\ell$ according to the scheme's label order.
Then, when the client algorithm calls the (mutable) function $\ell.cancel(\ell')$, the labeling algorithm marks $\ell$ as canceled (by $\ell'$).

In \cite[Algorithm 2]{DBLP:journals/corr/DolevGMS15}, in case $\ell' \not \lb \ell$ holds, $p_i$ marks $\ell$ as canceled by $\ell'$  by calling $\ell.cancel(\ell')$ (cf. label cancelation definition in Section~\ref{s:DolevPaper}). 
In detail, the function $cancel()$ is an alias to $cancelExhausted()$~\cite[Algorithm 3, line 10]{DBLP:journals/corr/DolevGMS15}.


\Paragraph{Preserving the stabilization guarantees of the labeling algorithm}
%
We note that during a subexecution in which the client algorithm does not call the function $cancel()$, the approach for algorithm composition of this section is along the lines of the one in~\cite[Section 2.7]{DBLP:books/mit/Dolev2000} (cf. Figure~\ref{fig:composition}).
However, the function $cancel()$ changes the state of the labeling algorithm.
Thus, it is necessary for the stabilization proof of the compound algorithm, i.e., the composition of the labeling and client algorithms, to show that the stabilization guarantees of the labeling algorithm are preserved.
The (client) algorithm that we propose in Section~\ref{s:algorithms} (Algorithm~\ref{alg:SSVC}) for the vector clock problem is composed with the labeling algorithm of Dolev et al.~\cite[Algorithm 2]{DBLP:journals/corr/DolevGMS15} through the interface that we presented in this section.
In Section~\ref{s:proofs} we show that the algorithm that we propose for the vector clock problem preserves the labeling algorithm's stabilization guarantees.



\section{Vector Clock Pairs: operations, invariants, and event counting}
\label{s:pair}

In this section we define a vector clock pair, which is a construction for emulating a vector clock that can tolerate counter overflows.
We define the invariants and conditions that should hold for the vector clock pairs  with respect to Requirement\reqs. 
We show how to merge two (vector clock) pairs (Section~\ref{s:merging}), and use this construction for counting the events of a single processor and computing the query $\causalPrecedence()$, which we defined in Section~\ref{s:systemSettings} (Section~\ref{s:counting}).
In Section~\ref{s:algorithms} we use the vector clock pairs for designing a practically-self-stabilizing algorithm with respect to the abstract task that Requirement\reqs\ defines (cf. Section~\ref{s:systemSettings}).

\Paragraph{The (vector clock) pair}
We say that $I=\langle\ell, m,o\rangle$ is a \textit{(vector clock) item}, where $\ell$ is a label of the Dolev et al.~\cite{DBLP:journals/corr/DolevGMS15} labeling scheme (Section~\ref{s:DolevPaper}), $m$ (main) is an $N$-size vector of integers that holds the processor increments, and $o$ (offset) is an $N$-size vector of integers that the algorithm uses as a reference to $m$'s value upon $\ell$'s creation. We use $(I.m - I.o)\mmi$ for retrieving $I$'s vector clock value. We define a \textit{(vector clock) pair} as the tuple $Z=$ $\langle curr, prev \rangle$, where both $curr$ and $prev$ are vector clock items, such that $Z.curr.o=Z.prev.m$, i.e., two variable names that refer to the same storage (memory cell). We use $VC(Z) := (Z.curr.m - Z.curr.o)\mmi$ for retrieving $Z$'s vector clock. 
We assume that each processor $p_i$ stores a vector clock pair $local_i$ and we explain below how $p_i$ uses $local_i$ for counting local events as well as events that it receives from other processors, even when (concurrent) counter overflows occur.


\Paragraph{Starting a vector clock pair}
%
The first value of a pair $Z$ is $\la \la\ell, zrs, zrs \ra, \la\ell, zrs, zrs\ra  \ra$, where $\ell:=\getLabel()$ is the local maximal label and $zrs := (0,\ldots,0)$ \short{is the zero vector}.
\modified{}{The pseudocode is given in line~\ref{alg:SSVC:resetLocal} of Algorithm~\ref{alg:SSVC}.}
That is, the vector clock value of $Z$ is an $N$-sized vector of zeros, i.e., $VC(Z) = zrs$, that we associate with the local maximal label.

\Paragraph{Exhaustion of vector clock pairs}
We say that a pair $Z$ is \textit{exhausted} when Condition~\ref{eq:notExhausted} holds. 
Condition~\ref{eq:notExhausted} defines exhaustion when the sum of the elements of the vector clock's value $VC(Z)$ is at least $\MI-1$.
Note that defining exhaustion according to the sum of the vector clock's values reduces the exhaustion events, in comparison to defining exhaustion for every vector clock element overflow, i.e., for every $curr.m[i]$, $p_i\in P$.
The latter also justifies the use of one label for a vector clock item $I$, instead of $N$ labels, i.e., one per each element of $I.m$.
Since the size of a label $I.\ell$ in the Dolev et al. labeling scheme~\cite{DBLP:journals/corr/DolevGMS15} is in $\bigO(N^3)$, this linear improvement is significant.

\begin{equation}
exhausted(Z) \iff  \Sigma_{k=1}^N (Z.curr.m[k] - Z.curr.o[k]) \geq \MI - 1
\label{eq:notExhausted}
\end{equation}


\Paragraph{Reviving a (vector clock) pair}
When the (vector clock) pair $Z$ is exhausted (Condition~\ref{eq:notExhausted}), $p_i$ \textit{revives} $Z$ by 
(i) canceling the labels of $Z$, i.e., $Z.curr.\ell$ and $Z.prev.\ell$, and
(ii) replacing $Z$ with $Z'=\la\la getLabel(), Z.curr.m, Z.curr.m\ra, Z.curr\ra$. 
Hence, the value of the new vector clock, $Z'$, is an $N$-sized vector of zeros, i.e., $VC(Z') = Z.curr.m - Z.curr.m = (0,\ldots,0)$ and  $Z'$ has the current offset field, $Z'.curr.o$, that refers to the same main values as the ones recorded by $Z.curr.m$ (and $Z'.curr.o$ alias value, which is $Z'.prev.m$). 
As we show in this section, the fact that $Z'.prev$ stores the value of $Z.curr$ upon exhaustion enables counting local events, as well as, merging (vector clock) pairs even upon concurrent exhaustions in different processors.
%

\remove{ \IS{this text is very confusing at this point}
\short{Note that $Z$ and $Z'$ have the same reference points, which are $Z.curr.m$ and $Z'.prev.o =Z'.prev.m$ (Figure~\ref{fig:oneWrapAround}). Moreover, this property holds also when $Z$ wraps around to another pair, say $Z''$ (Figure~\ref{fig:concurrentWrapAround}). Next, we continue with providing the details needed for using this property when merging two pairs that can, possibly, wrap around concurrently.}
}


\Paragraph{Incrementing vector clock values}
Processor $p_i \in P$ increments its (vector clock) pair, $Z$, by incrementing the $i^\textit{th}$ entry of $Z$'s current item, i.e., it increments $Z.curr.m[i]$ by one.
The new value of the vector clock is $VC(Z) = (Z.curr.m$ $+$ $\idv$ $-$ $Z.curr.o)$ $\mmi$, where $\idv$ is an $N$-size vector with zero elements everywhere, except for the $i^\textit{th}$ entry which is one, and $Z.curr.m$ is the value before the increment. 
In case that increment leads to exhaustion (Condition~\ref{eq:notExhausted}), $p_i$ has to revive the pair $Z$.
We assume that a processor can call $increment()$ only before it starts the computations of a step that ends with a send operation, to ensure that increments are immediately propagated to all other processors.

\subsection{Merging two vector clock pairs}
\label{s:merging}

We present a set of invariants for a single (vector clock) pair as well as for two pairs.
We explain when it is possible to merge two pairs and present the merging procedure.
Our approach is based in finding a common label and offset in the items of the two pairs, which works as a common reference.

\Paragraph{Pair label orderings}
Given a pair $Z = \langle curr, prev \rangle$, we say that its elements are \textit{ordered} when Condition~\ref{eq:labelsOrdered} holds. That is, either the current label of a pair $Z$, $Z.curr.\ell$, is larger than the previous label, $Z.prev.\ell$, and $Z.prev.\ell$ is canceled, or the labels are equal and not canceled (Condition~\ref{eq:labelsOrdered}). 

\begin{small}
\begin{align}
\begin{split}
\lblOrdrd(Z) \iff& \left((Z.prev.\ell \lb Z.curr.\ell \land isCanceled(Z.prev.\ell))\lor\right.\\
& ~\left(Z.prev.\ell = Z.curr.\ell \land \neg isCanceled(Z.curr.\ell)\right)
\label{eq:labelsOrdered} 
\end{split}
\end{align}
\end{small}

\Paragraph{The $\mathbf{=_{\ell,o}}$ and $\mathbf{<_{\ell,o}}$ relations}
\modified{We define the relations $=_{\ell,o}$ and $<_{\ell,o}$ to be able to compare and order vector clock items (and hence pairs). Let $\ell_1 = \la ml_1, cl_1 \ra$ and $\ell_2 = \la ml_2, cl_2 \ra$ be two labels of the Dolev et al. labeling scheme~\cite{DBLP:journals/corr/DolevGMS15}.
Recall that for $\ell = \la ml, cl \ra$, $cl$ indicates if $\ell$ is canceled; if $cl = \bot$ then $\ell$ is not canceled and if $cl \neq \bot$, $cl$ is the label that canceled $ml$, i.e., $\ell$ is canceled.
We say that $\ell_1 =_m \ell_2$, if and only if $ml_1 = ml_2$.
In the sequel we will use $=_m$ and $=$ interchangeably when comparing labels, as the $cl$ part is only used for notifying whether a label is canceled or not. 


Let $\la \ell_1, m_1, o_1 \ra =_{\ell,o} \la \ell_2, m_2, o_2 \ra  \iff \ell_1 = \ell_2 \land o_1 = o_2$. 
We say that two (vector clock) items $z$ and $z'$ \textit{match (in label and offset)}, if and only if, $z =_{\ell, o} z'$. 
We use the order $\la \ell_1, m_1, o_1 \ra <_{\ell,o} \la \ell_2, m_2, o_2 \ra \iff \ell_1 <_{lb} \ell_2$ $\lor$ $(\ell_1 = \ell_2 \,\land\, o_1 <_{lex} o_2)$ for 
comparing between vector clock items,
where 
$<_{lex}$ is the lexicographic order in $\N$.
We define $\max_{\ell, o} \cX$ to be the $<_{\ell,o}$-maximum item in a set of items $\cX$ in which all labels are comparable with respect to $\lb$ and there exists a maximum label among them.
}{We define the relations $=_{\ell,o}$ and $<_{\ell,o}$ to be able to compare vector clock items (and hence pairs). \short{Let $\ell_1 = \la ml_1, cl_1 \ra$ and $\ell_2 = \la ml_2, cl_2 \ra$ be two labels~\cite{DBLP:journals/corr/DolevGMS15}.
We say that $\ell_1 =_m \ell_2$, if and only if $ml_1 = ml_2$.
We use $=_m$ and $=$ interchangeably when comparing labels.} 
Let $\la \ell, m, o \ra =_{\ell,o} \la \ell', m', o' \ra  \iff \ell = \ell' \land o = o'$. 
We say that two (vector clock) items $z$ and $z'$ \textit{match (in label and offset)}, if and only if, $z =_{\ell, o} z'$. 
We use the order $\la \ell_1, m_1, o_1 \ra <_{\ell,o} \la \ell_2, m_2, o_2 \ra \iff \ell_1 \prec_{lb} \ell_2$ $\lor$ $(\ell_1 = \ell_2 \land o_1 <_{lex} o_2)$, where $v_1, v_2$ are $N$-size vectors with elements in $\N$, $\prec_{lb}$ is the label order of the labeling scheme~\cite{DBLP:journals/corr/DolevGMS15},  and $<_{lex}$ is the lexicographic order.
We define $\max_{\ell, o} \cX$ to be the $<_{\ell,o}$-maximum in the item set $\cX$ in which all labels are comparable.}

\Paragraph{Pivot existence}
\modified{Condition~\ref{eq:afterPrInf} tests the pair merging feasibility (Figure~\ref{fig:mergePairs}). 
It considers two pairs $Z$ and $Z'$ and returns true when one of the following holds: 
\begin{enumerate}[(a)]
\item $Z$ and $Z'$ match (in label and offset) in their $curr$ and $prev$, i.e., $Z.itm =_{\ell, o} Z'.itm$, for $itm \in \{curr, prev\}$ (Figure~\ref{fig:noWrapAround}), i.e., there was no vector clock exhaustion,  or 
\item $Z$ and $Z'$ match in their $prev$, i.e., $Z.prev =_{\ell, o} Z'.prev$ (figure~\ref{fig:concurrentWrapAround}), i.e., both vector clocks where exhausted (assuming that case (a) was true before exhaustion),  or 
\item the label and offset in the $prev$ of one equals the label and offset in the $curr$ of the other one, i.e., $Z.curr =_{\ell, o} Z'.prev \lor Z.prev =_{\ell, o} Z'.curr$ (Figure~\ref{fig:oneWrapAround}), i.e., one vector clock was exhausted (assuming that case (a) was true before exhaustion).
\end{enumerate}
}{Condition~\ref{eq:afterPrInf} tests the feasibility of merging the pairs $Z$ and $Z'$ when: (a) $Z$ and $Z'$ match (in label and offset) in their $curr$ and $prev$, i.e., $Z.itm =_{\ell, o} Z'.itm$, for $itm \in \{curr, prev\}$ (Figure~\ref{fig:noWrapAround}),  or 
(b) $Z$ and $Z'$ match in their $prev$, i.e., $Z.prev =_{\ell, o} Z'.prev$ (figure~\ref{fig:concurrentWrapAround}),  or 
(c) the label and offset in the $prev$ of one equals the label and offset in the $curr$ of the other one, i.e., $Z.curr =_{\ell, o} Z'.prev \lor Z.prev =_{\ell, o} Z'.curr$ (Figure~\ref{fig:oneWrapAround}).}
We refer to the common item between $Z$ and $Z'$ as the \textit{pivot} item.


\begin{equation}
\begin{array}{rcr}
\existsOverlap(Z,Z') &\!\iff\!& Z.prev =_{\ell, o} Z'.prev \lor Z.curr =_{\ell, o} Z'.prev \\
&&\lor Z.prev =_{\ell, o} Z'.curr
\end{array}
\label{eq:afterPrInf}
\end{equation}


\begin{figure}[t]
    \centering
  \begin{subfigure}[b]{0.285\textwidth}
        \includegraphics[width=1.0\textwidth]{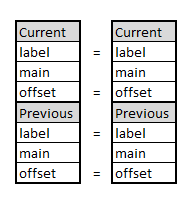}
    \caption{Condition~\ref{eq:afterPrInf} holds because the two pairs differ only by their $curr.main$ fields (no wrap-around).} 
    \label{fig:noWrapAround}
    \end{subfigure}
    \qquad 
    \begin{subfigure}[b]{0.285\textwidth}
    \includegraphics[width=1.0\textwidth]{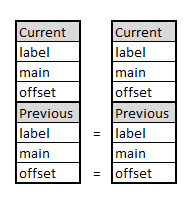}
    \caption{Condition~\ref{eq:afterPrInf} holds because the two pairs match in their $prev$ item (the pairs had wrapped around concurrently).}
    \label{fig:concurrentWrapAround}
    \end{subfigure}  
    \qquad 
    \begin{subfigure}[b]{0.285\textwidth}
        \includegraphics[width=1.0\textwidth]{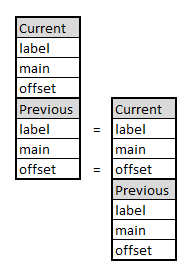}
    \caption{Condition~\ref{eq:afterPrInf} holds because $Z.prev$ and $Z'.curr$ differ only by their main filed ($Z$ has wrapped around).}
    \label{fig:oneWrapAround}
    \end{subfigure}
\caption{Conditions for merging two given (vector clock) pairs; $Z$ (on the left) and $Z'$ (on the right).}
\label{fig:mergePairs}
\end{figure}

\Paragraph{Merging two (vector clock) pairs}
Two vector clocks $Z$ and $Z'$ can be merged when there exists a pivot item, i.e., $\existsOverlap(Z,Z')$ holds (Figure~\ref{fig:mergePairs}).
The $<_{\ell,o}$-maximum pivot item, $pvt$, in $Z$ and $Z'$, provides a reference point when merging $Z$ and $Z'$, because it refers to a point in time from which both $Z$ and $Z'$ had started counting their events. 
We merge $Z$ and $Z'$ to the pair $\out$ in two steps; one for initialization and another for aggregation.

We initialize $\out$ to the $<_{\ell,o}$-maximum pair between $Z$ and $Z'$ (Figure~\ref{fig:mergePairs}), and choose $Z$ (the first input argument) when symmetry exists (figures~\ref{fig:noWrapAround} and~\ref{fig:concurrentWrapAround}).
%
In order to distinguish when we treat numbers and operations in $\N$ or in $\Z_{\MI}$, we denote by $x+_\N y$ the result of adding two numbers $x,y \in \Z_{\MI}$ in $\N$ ($x+_\N y$ can be possibly larger than $\MI$) and $x|_\N$ denotes that $x\in \Z_{\MI}$ is treated as a number in $\N$.

For every $i\in \{1,\ldots, N\}$, let $\newEvents(X,\pivot)[i]$ be the number of new events that the pair $X\in \{Z,Z'\}$ counts since the reference item, $\pivot$.
In Equation~\ref{eq:Xpivot} we compute $\newEvents(X,\pivot)[i]$ depending on whether $\pivot$ matches $X.curr$ or $X.prev$. 
In the former case, we count the number of events in $X.curr.m[i]$ since the offset $X.curr.o[i]$.
In the latter case, we also add the number of events in $X.prev.m[i]$ since the offset $X.prev.o[i]$, because $X.prev.o$ is the common offset of $Z$ and $Z'$.
The aggregation step sets 
$\out[i] = \max\{newEvents(X,\pivot)[i] \,|\, X \in \{Z, Z'\}\} + pivot[i]\mmi$,
for every $i\in \{1,\ldots, N\}$.

\begin{small}
\begin{equation}
\label{eq:Xpivot}
\newEvents(X,\pivot)[i] = \left\{\begin{array}{l}
(X.curr.m[i] - X.curr.o[i]\mmi)|_\N,\\
\hfill\text{ if } \pivot =_{\ell,o} X.curr,\\
\\
(X.curr.m[i] - X.curr.o[i]\mmi)|_\N\,\,\,\, +_\N\\
\hfill (X.prev.m[i] - X.prev.o[i]\mmi)|_\N,\\ 
\hfill\text{ if } \pivot =_{\ell,o} X.prev 
\end{array}\right.
\end{equation}
\end{small}

%

%

\subsection{Event counting and causal precedence}
\label{s:counting}
%

In this section we present our implementation of the queries about counting the events of a single active processor (Requirement\reqs) and about causal precedence (Section~\ref{s:systemSettings}), which is based on the vector clock pair construction.
We explain the conditions under which we compute the query of how many events occurred in a processor $p_i$ between the states $c_x$ and $c_y$ (Requirement\reqs) using $local_i$'s value in these two states, and present the query's computation.
Then, we describe how we compute the query $\causalPrecedence(local_i, local_j)$, for two vector clocks $local_i$ and $local_j$ of active processors $p_i$ and $p_j$, in possibly different states (cf. Section~\ref{s:systemSettings}).

Let $V_i^k[i]$ be the $i^{th}$ entry of $p_i$'s vector clock $V_i$ in state $c_k$, $k\in \{x,y\}$.
Requirement\reqs\ implies that in a legal execution, the query $V_i^y[i] - V_i^x[i]$ returns the number of events that occurred in $p_i$ between the states $c_x$ and $c_y$, where $c_x$ precedes $c_y$.
Let $local_i^k$ be the value of $local_i$ in state $c_k$.
The result of this query depends on the number of calls to $revive_i()$ between $c_x$ and $c_y$.
That is, in case there were two or more calls to $revive_i()$ between $c_x$ and $c_y$, then it is not possible to infer the correct response to the query $V_i^y[i] - V_i^x[i]$ from $local_i^x$ and $local_i^y$, since these two pairs have no common pivot item (cf. Section~\ref{s:merging}). 
Otherwise, in case there was no wrap around (cf. Figure~\ref{fig:noWrapAround}) or one wrap around (cf. Figure~\ref{fig:oneWrapAround}) between $c_x$ and $c_y$, we can use $local_i.curr$, or respectively, $local_i.prev$ as pivot items to count the correct number of events in $p_i$.
Thus, we compute the response to the query $V_i^y[i] - V_i^x[i]$ as follows:

\begin{equation}\small
\label{eq:eventCountQuery}
V_i^y[i] - V_i^x[i] = \left\{ \begin{array}{ll}
VC(local_i^y)[i] - VC(local_i^x)[i],& \text{if $local_i^x$ and $local_i^y$ differ only on} \\
&\text{the field $curr.m$} \text{ (cf. Figure~\ref{fig:noWrapAround})}\\
&\\
\newEvents(local_i^y, local_i^y.prev)[i],& \text{if } local_i^x.curr =_{\ell,o} local_i^y.prev\\
& \text{ (cf. Figure~\ref{fig:oneWrapAround})}\\
 &\\
\bot, & \text{otherwise}
\end{array} \right.
\end{equation}

In Section~\ref{s:algorithms} we propose Algorithm~\ref{alg:SSVC} and in Section~\ref{s:proofs} we show that it is practically-self-stabilizing with respect to Requirement\reqs\ (cf. Section~\ref{s:systemSettings}).
Thus, during a legal execution the return value of $V_i^y[i] - V_i^x[i]$ in Equation~\ref{eq:eventCountQuery} is never $\bot$.




In order to compute the query $\causalPrecedence(Z,Z')$, which is true if and only if $Z$ causally precedes $Z'$ (Section~\ref{s:systemSettings}), we follow a similar approach to merging pairs (Section~\ref{s:merging}).
As in the computation of the query $V_i^y[i] - V_i^x[i]$,
we require that there exists a pivot item, $\pivot$, between two pairs $Z$ and $Z'$ in order to be able to compare them, and 
we use the $\newEvents(X,pivot)$ function to compare these pairs, $X\in \{Z,Z'\}$.
We detail the computation of $\causalPrecedence(Z,Z')$ in Equation~\ref{eq:causal}.

\begin{equation}
\label{eq:causal}
\begin{footnotesize}
\begin{array}{ll}
\!\!\causalPrecedence(Z,Z') \Leftrightarrow&\!\!\!\! \existsOverlap(Z,Z') \land \\
&\!\!\!\!\! (\forall_{i\in \{1,\ldots, N\}} \newEvents(Z,\pivot)[i] \leq \newEvents(Z',\pivot)[i] \land \\
&\!\!\!\!\! \exists_{j\in \{1,\ldots, N\}} \newEvents(Z,\pivot)[j] < \newEvents(Z',\pivot)[j])
\end{array}
\end{footnotesize}
\end{equation}
%
Our approach of including the $prev$ item in a vector clock pair allows to count events from a common reference, even when wrap around events occur.
Hence, in a legal execution of Algorithm~\ref{alg:SSVC} (Section~\ref{s:algorithms}), $\causalPrecedence(Z,Z')$ as we compute it in Equation~\ref{eq:causal} is true if and only if $Z$ causally precedes $Z'$.

\section{Practically-self-stabilizing Vector Clock Algorithm}
\label{s:algorithms}


%
We propose Algorithm~\ref{alg:SSVC} as a practically-self-stabilizing vector clock algorithm that fulfills Requirement~\ref{req:2act} (Section~\ref{s:systemSettings}).
Algorithm~\ref{alg:SSVC} builds on the vector clock pair construction (Section~\ref{s:pair}), which uses a practically-self-stabilizing labeling scheme (cf. Section~\ref{s:back}).
Thus, Algorithm~\ref{alg:SSVC} is composed with a\remove{the} labeling algorithm\remove{ of Dolev et al.~\cite[Algorithm 2]{DBLP:journals/corr/DolevGMS15}} using our composition approach and the interface in Section~\ref{s:interfaceDolev}.
In a nutshell, Algorithm~\ref{alg:SSVC} includes the procedures for 
(i) vector clock increments, 
(ii) checking the invariants of the local (vector clock) pair, e.g., vector clock exhaustion, and sending the local pair of a processor ($local$) to its neighbors (do-forever loop procedure), and 
(iii) merging an incoming vector clock pair with the local one.
To that end, Algorithm~\ref{alg:SSVC} relies on the functions that we defined in Section~\ref{s:pair}.


\newcommand{\algSize}{footnotesize}
\begin{algorithm*}[t!]

\begin{\algSize}

\SetKwInput{KwVariables}{Variables}
\SetKwInput{KwMacros}{Macros}


\textbf{Constants:}
$zrs := (0,\ldots, 0)$: the $N$-size vector of zeros, $\idv$: $N$-size vector, where $\idv[i] = 1$ and $\idv[j] = 0$, for $j\neq i$\;


\textbf{Variables:} 
$\pairs[]$: an $N$-size vector of pairs, where $\pairs[i]$ is the local vector clock pair, i.e., $local$ is an alias to $\pairs[i]$. Also, $\pairs[j]$ is the latest value of $p_j$'s $local$ that $p_i$ received.\label{alg:SSVC:variables}



\textbf{Interface:} 
$\isStored()$, $\getLabel()$, $\legitArriving()$, $encapsulate()$, $cancel()$, $labelBookkeeping()$  (Section~\ref{s:interfaceDolev}), $newEvents()$ (Section~\ref{s:pair}). 



\textbf{Macros:} we use as macros conditions~\ref{eq:notExhausted} to~\ref{eq:afterPrInf}, Equation~\ref{eq:Xpivot} (Section~\ref{s:pair}), and the following:\\
$\mirroredLocal() :=  \isStored(local.prev.\ell) \land local.curr.\ell = \getLabel()$\label{alg:SSVC:mirroredLocalLabels}\;
$\pairInvariants(X)$ $:=$ $\neg exhausted(X) \land (X.prev.\ell \preceq_{lb} X.curr.\ell)$\label{alg:SSVC:pairInvar}\; 

$\compLbls(\cX) := \forall \ell, \ell' \in \{X.curr.\ell, X.prev.\ell \,|\, X\in \cX\}, \ell \leqlb \ell' \lor \ell' \leqlb \ell$\label{alg:SSVC:compLbls}\;
$\labelCheck(X,Y) := \compLbls(\{X,Y\}) \land \existsOverlap(X,Y)$ (Condition~\ref{eq:afterPrInf}, Section~\ref{s:pair})\label{alg:SSVC:legitPairs}\;

$\resetLocal() := \{local \gets \la y, y \ra\}$, where $y = \la \getLabel(), zrs, zrs  \ra$\label{alg:SSVC:resetLocal}\;

$\equalStatic(X,Y) := X.curr.\ell = Y.curr.\ell \land X.curr.o = Y.curr.o \land X.prev = Y.prev$\label{alg:SSVC:equalStatic}\;



\textbf{procedure} $\cancelPairLabels(Z)$\label{alg:SSVC:cancelLabelPairsBegin} \Begin{$\textbf{foreach} {~\ell \in \{Z.curr.\ell, Z.prev.\ell\}} \textbf{ do } \ell.cancel(\ell)$; $labelBookkeeping();$\label{alg:SSVC:resetLabelPairs}}



\textbf{function} $\revive(Z)$ \label{alg:SSVC:reviveBegin}
\Begin{$\cancelPairLabels(Z)$\label{alg:SSVC:reviveCancelLbls}; 
\Return $\la \la \getLabel(), Z.curr.m, Z.curr.m \ra, Z.curr \ra$\;\label{alg:SSVC:reviveEnd}}



\textbf{procedure} $increment()$ \Begin{\label{alg:SSVC:incrementStart}
\Let $local$$=$$\la\la local.curr.\ell, (local.curr.m + \idv) \mmi, local.curr.o  \ra, local.prev\ra$\label{alg:SSVC:increment}\;
\lIf{$exhausted(local)$}{$local \gets \revive(local)$\label{alg:SSVC:checkIncrExh}}
}

\textbf{function} $merge(loc, arr)$ \Begin{\label{alg:SSVC:merge}
%
\leIf{$\exists_{x\in \{curr,\, prev\}} \locCurr =_{\ell,o} arr.x$}{\Let $\pivot$$:= \locCurr.o$}{\Let $\pivot$$:= \locPrev.o$\label{alg:SSVC:pivot}}

\textbf{let} $\initializeToLocal := \arrCurr <_{\ell,o} \locCurr$ $\lor$ $(\arrCurr =_{\ell,o} \locCurr \land \arrPrev$ $\leq_{\ell,o}$ $\locPrev)$\label{alg:SSVC:useLocal}\;
\leIf{$\initializeToLocal$\label{alg:SSVC:mergeInitStart}}
{\Let $\out := loc$\label{alg:SSVC:mergeInitLoc}}
{\Let $\out := arr$\label{alg:SSVC:mergeInit}}

%
\ForEach{$k\in \{1,\ldots,N\}$}{\label{alg:SSVC:newEventsStart}
\Let $maxNewEvents = \max\{newEvents(Z,\pivot)[k] \,|\, Z \in \{loc, arr\} \}$\label{alg:SSVC:newEventsCount}\;
$\oc.m[k] \gets (\pivot[k] + maxNewEvents) \mmi$
\label{alg:SSVC:newEventsEnd}\;
}

\Return $\out$\label{alg:SSVC:mergeReturn}\;
}



\textbf{do forever} \Begin{\label{alg:SSVC:doForeverStart}
$labelBookkeeping()$\label{alg:SSVC:callDolevDoForever}\;
\lIf{$\neg(\mirroredLocal() \land \lblOrdrd(local))$}{$\resetLocal()$\label{alg:SSVC:removeStaleInfo}}
\lIf{$exhausted(local)$}{$local \gets revive(local)$\label{alg:SSVC:doForeverRevive}}
 
\lForEach{$p_k \in P\setminus\{p_i\}$}{\Send $encapsulate(\la local, \pairs[j]\ra)$ \To $p_k$\label{alg:SSVC:doForeverEnd}}
}



\Upon \textbf{message} $m= \langle \bullet, \la arriving, \receiverLocal\ra \rangle$ \textbf{arrival from} $p_j$  \Begin{\label{alg:SSVC:messageReceiveStart}
$labelBookkeeping(m,j)$\label{alg:SSVC:callDolevArriving}\;
$\pairs[j] \gets arriving$\label{alg:SSVC:storeArriving}\; 
%
\If{$\equalStatic(local, \receiverLocal) \land \legitArriving(m,arriving.curr.\ell) \land \pairInvariants(arriving)$\label{alg:SSVC:exhArrCheck}}{
\lIf{$\neg \labelCheck(local, arriving)$\label{alg:SSVC:invarMerge}}{$\resetLocal()$\label{alg:SSVC:receiveResetPairs}}
\Else{\label{alg:SSVC:LegitElse}
$local \gets merge(local, arriving)$\label{alg:SSVC:rvcElseStart}\label{alg:SSVC:mergeArrWithLocal}\;
\lIf{$exhausted(local)$}{$local \gets revive(local)$\label{alg:SSVC:reviveCall}}
}

}
}


\end{\algSize}

\caption{\smaller Practically-self-stabilizing vector-clock replication, code for $p_i$}\label{alg:SSVC}
\end{algorithm*}


\Paragraph{Local variables (line~\ref{alg:SSVC:variables})}
Processor $p_i \in P$ maintains a local (vector clock) pair, $local_i$, such that for any state, $p_i$'s vector clock value is $VC(local_i)$ (cf. Section~\ref{s:pair}).

\Paragraph{Restarting $local$ via $\resetLocal()$ (line~\ref{alg:SSVC:resetLocal})}
The macro $\resetLocal()$ lets $local_i$ have its starting value $\la y, y \ra$, where $y = \la \getLabel(), zrs, zrs  \ra$ and $zrs$ is the $N$-size vector of zeros.
Processor $p_i$ can use $\resetLocal()$ for setting $local_i$ to its initial value, whenever the invariants for $local_i$ do not hold in the do-forever loop and in the message arrival procedures of Algorithm~\ref{alg:SSVC}.

\Paragraph{Token passing mechanism for sending and receiving $local$}
Algorithm~\ref{alg:SSVC} uses a token circulation mechanism for sending and receiving $local$, which is independent of the algorithm's computations on $local$.
This mechanism is necessary for ensuring that (after a constant number of steps) for every two processors $p_i, p_j \in  P$, $p_j$ processes a message from $p_i$ only if $p_i$ has received the latest value of $local_j$.

We remark that without this mechanism, it is possible that 
$p_i$ does not receive (and process) $p_j$'s latest value of $local_j$ for an unbounded number of steps, and yet $p_i$ keeps sending $local_i$ to $p_j$ for an unbounded number of steps.
The latter case can cause an unbounded number of steps that include a call to $\resetLocal()$ at $p_j$, if the pair that $p_j$ received from $p_i$ cannot be merged with $local_j$ (cf. Section~\ref{s:merging} and message arrival procedure in this section).
In Section~\ref{s:proofs}, we show that a call to $\resetLocal()$ in a step of the algorithm (possibly) implies that Requirement\reqs\ does not hold for the state that immediately follows this step. 
Hence, an unbounded number of calls to $\resetLocal()$, imply an unbounded number of states in which Requirement\reqs\ does not hold.
The token circulation mechanism helps the proposed algorithm to avoid this problem.

To implement the token circulation mechanism, each processor $p_i$ maintains an $N$-size vector of pairs, $\pairs_i[]$, where $\pairs_i[j]$, for $j\neq i$, is the last value of $local_j$ that $p_i$ received (from $p_j$), and $\pairs_i[i]$ stores $p_i$'s pair, i.e., $local_i$ is an alias for $\pairs_i[i]$.
We implement the token passing mechanism by augmenting the messages that a processor sends (via $encapsulate()$) in Algorithm~\ref{alg:SSVC} as follows.
A processor $p_i$ sends $\la local_i, \pairs_i[j]\ra$ to a processor $p_j$ by calling $encapsulate(\la local_i, \pairs_i[j]\ra)$ in line~\ref{alg:SSVC:doForeverEnd}.
Hence, a message sent by $p_j$ and received by $p_i$ has the form $m_j= \la \bullet, \la arriving_j, \receiverLocal_j\ra\ra$ (line~\ref{alg:SSVC:messageReceiveStart}).
Processor $p_i$ stores $arriving_j$ in $\pairs_i[j]$ (line~\ref{alg:SSVC:storeArriving}), in order to ensure that $p_i$ has received the latest value of $local_j$.
Thus, processor $p_i$ processes the message $m_j$ if the pairs $local_i$ and $\receiverLocal_j$ are equal or differ only on their $curr.m$, since the merging conditions (cf. Section~\ref{s:merging}) don't depend on $curr.m$.
We detail the exact procedures of sending and receiving messages in Algorithm~\ref{alg:SSVC} in the last part of this section.
In Section~\ref{s:proofs} we show that the token passing mechanism is self-stabilizing (in at most $\capacity N^2$ steps).


\Paragraph{The function $revive()$ (lines~\ref{alg:SSVC:reviveBegin}--\ref{alg:SSVC:reviveEnd})}
\label{s:algDescr:revive}
When the pair $Z$ is exhausted (Condition~\ref{eq:notExhausted}), a call to $revive(Z)$ lets $Z$ to wrap around and return its new version (Section~\ref{s:pair}).
That is, $p_i$ cancels $Z$'s labels, $Z.curr.\ell$ and $Z.prev.\ell$, by calling the labeling algorithm (function $\cancelPairLabels$, lines~\ref{alg:SSVC:cancelLabelPairsBegin}--\ref{alg:SSVC:resetLabelPairs}), and then sets  $Z.curr$ to be the output pair's $prev$ and $\la getLabel(), Z.curr.m, Z.curr.m\ra$ as the output's $curr$.

\Paragraph{The vector clock increment function, $increment()$ (lines~\ref{alg:SSVC:incrementStart}--\ref{alg:SSVC:checkIncrExh})}
When $p_i$ calls $increment()$, it increments the $i^\textit{th}$ entry of $p_i$'s vector clock.
That is, $p_i$ increments $local_i.curr.m[i]$ by $1$ by adding $\idv$ to $local_i.curr.m$, where $\idv$ is an $N$-size vector with zero elements everywhere, except for the $i^\textit{th}$ entry which is 1 (line~\ref{alg:SSVC:increment}).
%
In case that increment leads to a vector clock exhaustion, it calls the function $revive()$ (line~\ref{alg:SSVC:checkIncrExh}). 
%
We assume that a processor can only call $increment()$ 
in the beginning of a step that ends with a send operation (see paragraph on Algorithm~\ref{alg:SSVC}'s do-forever loop below), and this call is part of the step.
This restriction ensures that vector clock increments are immediately sent to all other processors.

\Paragraph{Aggregation of vector clock pairs with the $\mathbf{merge}()$ function (lines~\ref{alg:SSVC:merge}--\ref{alg:SSVC:mergeReturn})}
The function $merge(Z,Z')$ (lines~\ref{alg:SSVC:merge}--\ref{alg:SSVC:mergeReturn}) aggregates two pairs, $Z$ and $Z'$, such as the local one and another one arriving via the network. 
It outputs a pair $\out$ with the $<_{\ell,o}$-maximum items that includes the aggregated number of events of $Z$ and $Z'$ (Section~\ref{s:pair}).
%
%
%

The function uses the $<_{\ell,o}$-maximum pivot item $x$ in $Z$ and $Z'$, from which it counts the new events in $Z$ and $Z'$ (line~\ref{alg:SSVC:pivot}). 
It initializes the output pair, $\out$, with the input pair that is $<_{\ell,o}$-maximum both in $curr$ and $prev$ (lines~\ref{alg:SSVC:useLocal}--\ref{alg:SSVC:mergeInit}).
The algorithm then updates $\out.curr.m$ with the maximum number of new events between $Z.curr.m$ and $Z'.curr.m$ since the pivot item (lines~\ref{alg:SSVC:newEventsStart}--\ref{alg:SSVC:newEventsEnd}), and returns $\out$ (line~\ref{alg:SSVC:mergeReturn}). 
%
That is, 
$output.curr.m[i] = \max\{newEvents(X,\pivot)[i] \,|\, X \in \{Z, Z'\}\} + pivot[i]\mmi$,
for every $i\in \{1,\ldots, N\}$.

\Paragraph{The procedures of the do-forever loop and the message arrival event}
We explain Algorithm~\ref{alg:SSVC}'s do-forever loop (lines~\ref{alg:SSVC:doForeverStart}--\ref{alg:SSVC:doForeverEnd}) and message arrival procedure (lines~\ref{alg:SSVC:messageReceiveStart}--\ref{alg:SSVC:reviveCall}), which follow the algorithm composition of Figure~\ref{fig:composition} (Section~\ref{s:interfaceDolev}).

\Subparagraph{The do-forever loop procedure (lines~\ref{alg:SSVC:doForeverStart}--\ref{alg:SSVC:doForeverEnd})}
The do-forever loop starts by letting the labeling algorithm take a step in line~\ref{alg:SSVC:callDolevDoForever} (part 1 of Figure~\ref{fig:composition}).
Line~\ref{alg:SSVC:removeStaleInfo} refers to the invariants of $local$.
Algorithm~\ref{alg:SSVC} calls $\resetLocal()$ in line~\ref{alg:SSVC:removeStaleInfo}, in case one of the following does not hold: 
(i) $local.curr.\ell$ is not the local maximal label or $local.prev.\ell$ is not stored in the labeling algorithm's storage, i.e., if $\mirroredLocal()$ is false (line~\ref{alg:SSVC:mirroredLocalLabels}), or
(ii) Condition~\ref{eq:labelsOrdered} is false, i.e., $\lblOrdrd(local)$ is false.
In line~\ref{alg:SSVC:doForeverRevive}, the algorithm checks if $local$ is exhausted and in the positive case, $local$ wraps around to the return value of $revive(local)$ (cf. line~\ref{alg:SSVC:resetLocal}).
Lines~\ref{alg:SSVC:removeStaleInfo}--\ref{alg:SSVC:doForeverRevive} refer to part 2 of Figure~\ref{fig:composition}.

In line~\ref{alg:SSVC:doForeverEnd} the processor sends $local$ to every other processor in the system.
The processor sends the message $m_{client} = \la local, pairs[j]\ra$ to every $p_j\in P\setminus\{p_i\}$, by calling $encapsulate(m_{client})$.
The pair $pairs[j]$ is appended due to the token circulation mechanism.
Line~\ref{alg:SSVC:doForeverEnd} refers to part 3 of Figure~\ref{fig:composition}.

\Subparagraph{The message arrival procedure (lines~\ref{alg:SSVC:messageReceiveStart}--\ref{alg:SSVC:reviveCall})}
Upon arrival of a message $m = \la\bullet, \la arriving$, $\receiverLocal\ra\ra$ from processor $p_j$ (part 4 of Figure~\ref{fig:composition}) the labeling algorithm processes its own part of $m$ (part 5 of Figure~\ref{fig:composition}) by the call to $labelBookkeeping(m,j)$ in line~\ref{alg:SSVC:callDolevArriving}.
In lines~\ref{alg:SSVC:storeArriving}--\ref{alg:SSVC:reviveCall} of the message arrival procedure, Algorithm~\ref{alg:SSVC} processes $\la arriving$, $\receiverLocal\ra$ (parts 6 and 7 of Figure~\ref{fig:composition}).
In line~\ref{alg:SSVC:storeArriving} the algorithm stores $arriving$ to $pairs[j]$, i.e., the latest pair that $p_i$ received from $p_j$, to facilitate the token passing mechanism.

Algorithm~\ref{alg:SSVC} proceeds in processing $arriving$ only if $\equalStatic(local, \receiverLocal) \land \legitArriving(m,arriving.curr.\ell) \land \pairInvariants(arriving)$ holds (line~\ref{alg:SSVC:exhArrCheck}).
Let $\receiverLocal$ be the pair that $p_j$ had received from $p_i$ immediately before the step in which it sent the message $m$ to $p_i$.
The predicate $\equalStatic(local, \receiverLocal)$ (line~\ref{alg:SSVC:equalStatic}) is true, if $\receiverLocal$ either equals $local$ or differs from $local$ only in $curr.m$ (in case until the reception of $m$, $p_i$ incremented its vector clock pair, without exhausting it).

Recall that the part of $m$ that refers to the labeling algorithm includes $p_j$'s local maximal label, which should be equal to $arriving.curr.\ell$ (cf. $\mirroredLocal()$ predicate in line~\ref{alg:SSVC:removeStaleInfo}).
The predicate $\legitArriving(m,arriving.curr.\ell)$ (cf. Section~\ref{s:interfaceDolev}) is true if $arriving.curr.\ell$ is equal to $p_j$'s local maximal label as it appears in the part of $m$ that refers to the labeling algorithm.
The predicate $\pairInvariants(arriving)$ (line~\ref{alg:SSVC:pairInvar}) is true if $arriving$ is not exhausted (Condition~\ref{eq:notExhausted}) and $arriving.prev.\ell \preceq_{lb} arriving.curr.\ell$ holds.
Hence, if $\legitArriving(m,arriving.curr.\ell) \land \pairInvariants(arriving)$ is false, $m$ contains stale information and existed in the system in the starting system state.

In case the condition of line~\ref{alg:SSVC:exhArrCheck} holds, the algorithm attempts to merge the arriving pair with the local one.
Merging is feasible if $\labelCheck(local, arriving)$ holds.
The predicate $\labelCheck(X,Y)$ (line~\ref{alg:SSVC:legitPairs}) is true if and only if $\compLbls(\{X,Y\}) \land \existsOverlap(X,Y)$ holds.
That is, all the labels of the pairs $X$ and $Y$ must be comparable with respect to the order of the labeling scheme and there exist a pivot item between $X$ and $Y$ (Condition~\ref{eq:afterPrInf}, Section~\ref{s:pair}).
In case $\labelCheck(local, arriving)$ is false, the algorithm calls $\resetLocal(local)$ (line~\ref{alg:SSVC:receiveResetPairs}), since merging must be possible in a legal execution.
Otherwise, merging $local$ and $arriving$ is feasible, and thus the algorithm lets $local$ to have the return value of $merge(local, arriving)$ (line~\ref{alg:SSVC:mergeArrWithLocal}).
In case the new pair value of $local$ is exhausted, $local$ wraps around to the return value of $revive(local)$ in line~\ref{alg:SSVC:reviveCall} (cf. line~\ref{alg:SSVC:resetLocal}).

\Subparagraph{Remarks on algorithm composition}
Note that in case of pair exhaustion Algorithm~\ref{alg:SSVC} forces the repetition of parts 1 and 2 of Figure~\ref{fig:composition} corresponding to the do-forever loop procedure, as well as, parts 5 and 7 of Figure~\ref{fig:composition} corresponding to the message arrival procedure.
That is, the algorithm requests the cancelation of $local$'s labels by the labeling algorithm, the labeling algorithm cancels these labels, and the call to $labelBookkeeping()$ provides a new local maximal label (cf. lines~\ref{alg:SSVC:cancelLabelPairsBegin}--\ref{alg:SSVC:reviveEnd}).
Then, Algorithm~\ref{alg:SSVC} stores the return value of $revive(local)$ in $local$ (line~\ref{alg:SSVC:doForeverRevive} or~\ref{alg:SSVC:reviveCall}).
Thus, if $local$ is not exhausted during a step (line~\ref{alg:SSVC:doForeverRevive} or~\ref{alg:SSVC:reviveCall}), the composition of the labeling and the vector clock algorithm is along the lines of~\cite[Section 2.7]{DBLP:books/mit/Dolev2000}.
The latter holds, since Algorithm~\ref{alg:SSVC} changes the state of the labeling algorithm only when it calls $cancel()$ and this occurs only upon a call to $revive()$ (due to pair exhaustion).
Also, this repetition of step parts occurs at most once per step, since the output pair of $revive(local)$ is by definition not exhausted (cf. line~\ref{alg:SSVC:reviveEnd} and Section~\ref{s:pair}).
Moreover, in case the invariants for $local$ do not hold in line~\ref{alg:SSVC:removeStaleInfo} or~\ref{alg:SSVC:receiveResetPairs}, the call to $\resetLocal()$ in these lines does not change the state of the labeling algorithm, since it only retrieves the local maximal label through $\getLabel()$ (cf. Section~\ref{s:interfaceDolev}).



\remove{ 
It is important to clarify, as we mentioned in Section~\ref{s:interfaceDolev} (Figure~\ref{fig:composition}), that, as long as the pairs in use are not exhausted, our approach for algorithm composition is along the lines of the one in~\cite[Section 2.7]{DBLP:books/mit/Dolev2000}. 
The event of an exhausted local pair, either due to local increments (line~\ref{alg:SSVC:doForeverRevive}) or merging with an arranging pair (line~\ref{alg:SSVC:reviveCall}), results in 
(a) the cancelation of the current local pair, 
(b) a call for an iteration of the labeling scheme, and then 
(c) the retrieval of the current local maximal label from which the revived label is formed (line~\ref{alg:SSVC:reviveBegin}). 
Another important observation is that all other corruptions of the of $local_i$'s state, result in the restart of the $local_i$ pair (line~\ref{alg:SSVC:resetLocal}), which does not include the cancelation of the current label and thus, in these cases, Algorithm~\ref{alg:SSVC} does not breach the approach for algorithm composition depicted by Figure~\ref{fig:composition}. 
} 

\section{Correctness Proof}
\label{s:proofs}

\subsection{The proof in a nutshell}
\label{s:proofNutshell}

We show that Algorithm~\ref{alg:SSVC} is practically-self-stabilizing (Definition~\ref{def:practSelfStab}).
Recall from Section~\ref{s:systemSettings} that the number of system states in which active processors in an execution $R$ deviate from the abstract task is denoted by $f_R$.
For the vector clock abstract task, $f_R$ denotes the number of system states in $R$, in which Requirement\reqs\ does not hold, with respect to the active processors in $R$.
Thus, in Theorem~\ref{thm:reqHold} we show that for any $\pinf$-scale execution $R$, $f_R \ll |R|$ holds (cf. Section~\ref{s:systemSettings}).

\begin{myTheorem}[Algorithm~\ref{alg:SSVC} is practically-self-stabilizing]
\label{thm:reqHold}
For every infinite execution $R$ of Algorithm~\ref{alg:SSVC}, and for every $\pinf$-scale subexecution $R' \seg R$, $f_{R'} = f(R',N) \ll |R'|$ holds.
\end{myTheorem}

To the end of proving Theorem~\ref{thm:reqHold}, we first present a set of invariants both for the state of a single active processor and also when considering the states of all active processors in an execution (Section~\ref{s:invariants}).
Given these invariants we present the conditions for an execution to be legal (Section~\ref{s:invariants}).
More specifically, 
we show that an execution is legal if, 
(i) there are no steps that include a call to $\resetLocal()$, and 
(ii) for each processor, there is at most one step in which that processor calls the function $revive()$.
In Section~\ref{s:PSVC} we study the functions that \textit{cause} a call to $\resetLocal()$ or $revive()$.
That is, we define a notion of \textit{function causality}, which bases on the interleaving model (cf. Section~\ref{s:systemSettings}).
Then, in Section~\ref{s:boundRestart}, we prove that for every $\pinf$-scale execution $R'$, the number of steps that include a call to either $\resetLocal()$ or $revive()$ is significantly less than $|R'|$, and combine the above to prove Theorem~\ref{thm:reqHold} (Corollary~\ref{cor:practicallyStabilizing}).

Our proof also requires to show that the labeling algorithm by Dolev et al.~\cite{DBLP:journals/corr/DolevGMS15} remains practically-self-stabilizing (Section~\ref{s:labelingAlgStabilizes}), even if we use \modified{a larger, but yet bounded number of}{twice as many} labels\remove{ (since each pair has two labels)}, by extending the size of the label storage, i.e., $storedLabels_i$, for each $p_i\in P$ (cf. Section~\ref{s:DolevAbstractTask}).

\subsubsection{Notation}
We refer to the values of variable $X$ at processor $p_i$ as $X_i$. 
Similarly, $f_i()$ refers to the returned value of function $f()$ that processor $p_i$ executes.
Throughout the proof, any execution is an execution of Algorithm~\ref{alg:SSVC}.
Let $\msg = \capacity N(N-1)$ be the maximum number of messages, and hence pairs, that can exist in the communication channels in any system state, i.e., $N(N-1)/2$ links, where each link is a bidirectional communication channel of capacity $\capacity$ in each direction.
Moreover, recall that $P(R)\subseteq P$ is the set of processors that take steps during an execution $R$.
When referring to a value $Z_x$ that a variable takes, e.g., $local_i$, we treat $Z_x$ as an (immutable) literal, i.e., a value that does not change.


\subsection{Convergence of the labeling algorithm in the absence of wrap around events}
\label{s:labelingAlgStabilizes}

We generalize the lemmas of Dolev et al.~\cite{DBLP:journals/corr/DolevGMS15} \short{(Section~\ref{s:DolevAbstractTask})} that bound the number of label creations and adoptions of their labeling algorithm (cf. Section~\ref{s:DolevAbstractTask}) to accommodate for the extra number of labels of Algorithm~\ref{alg:SSVC}, and show that the labeling algorithm converges when twice as many labels are processed, due to the fact that each pair includes two labels. 
Recall from Section~\ref{s:algorithms} that if there are no calls to $revive()$ (lines~\ref{alg:SSVC:checkIncrExh}, \ref{alg:SSVC:doForeverRevive}, and~\ref{alg:SSVC:reviveCall}) during an execution, then Algorithm~\ref{alg:SSVC} does not change the state of the labeling algorithm (cf. function definition in lines~\ref{alg:SSVC:reviveBegin}--\ref{alg:SSVC:reviveEnd}).
However, in the starting system state of an execution of Algorithm~\ref{alg:SSVC} there exist twice as many labels as in the starting system state of an execution of the labeling algorithm, due to the two labels that each pair consists of.

%
%

We extend~\cite[Lemma 4.3]{DBLP:journals/corr/DolevGMS15}, which bounds the number of labels that were created by $p_j$ and adopted by $p_i$, after $p_j$ stopped adding labels to the system  (Corollary~\ref{cor:dolevLabelAdoptions}).
In Corollary~\ref{cor:dolevLabelCreations}, we extend~\cite[Lemma 4.4]{DBLP:journals/corr/DolevGMS15}, which bounds the number of labels that $p_i$ creates (Corollary~\ref{cor:dolevLabelCreations}).
We then present Corollary~\ref{cor:labelingSchemeConvergence} that is an implication of corollaries~\ref{cor:dolevLabelAdoptions} and~\ref{cor:dolevLabelCreations} and states that the labeling algorithm of Dolev et al.~\cite{DBLP:journals/corr/DolevGMS15} remains practically-self-stabilizing given the generalized bounds presented in those corollaries.
Corollary~\ref{cor:labelingSchemeConvergence} is an extension of~\cite[Theorem 4.2]{DBLP:journals/corr/DolevGMS15}, which shows that the labeling algorithm ~\cite[Algorithm 2]{DBLP:journals/corr/DolevGMS15} is practically-self-stabilizing. 
Hence, we will use Corollary~\ref{cor:labelingSchemeConvergence} for proving that Algorithm~\ref{alg:SSVC} is also practically-self-stabilizing. 
In Section~\ref{s:boundRestart}, we will extend these bounds to accommodate for the extra labels created by Algorithm~\ref{alg:SSVC}, when a processor calls the function $revive()$.

\begin{corollary}[extension of~{\cite[Lemma 4.3]{DBLP:journals/corr/DolevGMS15}}]
\label{cor:dolevLabelAdoptions}
Let $p_i, p_j \in P$ be two processors. Suppose that $p_j$ has stopped adding labels to the system state, and sending these labels during an execution $R$.
Moreover, suppose that at the system state that immediately follows the last step in which $p_j$ stopped adding labels to the system, the number of labels that have $p_j$ as their creator and that were adopted by any of the $N$ processors in the system is at most $2N$, and the maximum number of labels in transit that were created by $p_j$ is at most $2\msg$.
Processor $p_i$ adopts at most $2N+ 2\msg$ labels $\ell$, such that $\ell.\creator = j$ and $\ell \notin storedLabels_i[j]$.
%
\end{corollary}

The bound in \cite[Lemma 4.3]{DBLP:journals/corr/DolevGMS15} is $N+M$, but in the setting of Algorithm~\ref{alg:SSVC} each pair includes two labels, hence the factor of 2.
Thus, setting $|storedLabels_i[j]| = 2N + 2\msg$, $i\neq j$ allows the labeling algorithm to converge.
In the following corollary, we denote with $max_i[i]$ the local maximal label of processor $p_i$, as in the labeling algorithm of Dolev et al.~\cite{DBLP:journals/corr/DolevGMS15} (cf. Section~\ref{s:DolevPaper}).

%

\begin{corollary}[extension of~{\cite[Lemma 4.4]{DBLP:journals/corr/DolevGMS15}}]
\label{cor:dolevLabelCreations}
Let $p_i \in P$ be a processor and $L_i = \ell_{i_0}, \ell_{i_1}, \ldots$ be the sequence of legitimate (not canceled) labels that $p_i$ stores in $max_i[i]$ over an execution $R$, such that 
no counter exhaustions occur during $R$
 and $\ell_{i_k}.\creator = i$, $k\in\N$. 
It holds that $|L_i| \leq 4N^2 + 4N\msg -4N - 2\msg$~\cite{DBLP:journals/corr/DolevGMS15}.
\end{corollary}

The bound in Corollary~\ref{cor:dolevLabelCreations} follows by the proof of \cite[Lemma 4.4]{DBLP:journals/corr/DolevGMS15}, which bounds the number of labels existing either in other processors' states or in transit, for which $p_i$ is the label creator.
These labels are at most $2(\msg + \Sigma_{j\neq i} |storedLabels_i[j]|) = 2\msg + 2\left(N-1)(2N+2\msg)\right) = 4N^2 + 4N\msg -4N - 2\msg$ (the second equality holds by Corollary~\ref{cor:dolevLabelAdoptions}). 
\remove{
Note that counter exhaustions do not occur in $R$, hence Algorithm~\ref{alg:SSVC} does not affect the convergence of the labeling algorithm, since the function $revive()$ is never called (line~\ref{alg:SSVC:reviveBegin}) and the remainder of Algorithm~\ref{alg:SSVC} does not change the state of the labeling algorithm.
That is, except for the function $revive()$ (that includes the function $cancel()$ and $labelBookkeeping()$ (which allows the labeling algorithm to take a step), Algorithm~\ref{alg:SSVC} only interfaces with the labeling algorithm through the functions $\isStored()$, $\getLabel()$, $\legitArriving()$, $encapsulate()$ (Section~\ref{s:interfaceDolev}), which do not change the labeling algorithm's state.
Thus, in the absence of wrap around events setting $|storedLabels_i[i]| = 4N^2 + 4N\msg -4N - 2\msg$ is sufficient for the convergence of the labeling algorithm.
%
} 

Corollary~\ref{cor:labelingSchemeConvergence}  is a straightforward extension of~\cite[Theorem 4.2]{DBLP:journals/corr/DolevGMS15} that also holds for the updated bounds of corollaries~\ref{cor:dolevLabelAdoptions} and~\ref{cor:dolevLabelCreations}, since the proof is based on the bounds' existence, rather than the actual bounds.


\begin{corollary}[{extension of~\cite[Theorem 4.2]{DBLP:journals/corr/DolevGMS15}}]
\label{cor:labelingSchemeConvergence}
Let $R$ be an $\pinf$-scale execution of the labeling algorithm~\cite[Algorithm 2]{DBLP:journals/corr/DolevGMS15}, in which no wrap around events occur.
The labeling algorithm is practically-self-stabilizing in $R$, given the number of label creations and adoptions in corollaries~\ref{cor:dolevLabelAdoptions} and \ref{cor:dolevLabelCreations} as well as the updated queue lengths in $storedLabels_i$, where $p_i \in P$.
\end{corollary}

We remark that in Section~\ref{s:boundRestart} we extend the queue lengths to accommodate for the extra labels that are created due to Algorithm~\ref{alg:SSVC}, i.e., when wrap-around events occur and a processor calls $revive()$.



\subsection{Local and global invariants and their relation to Requirement\reqs}
\label{s:invariants}

In this section we study the local and global invariants that determine if an execution is legal.
We define the predicate $\localInvariants(i)$ (Definition~\ref{def:pairInvariants}), which gives the local invariants for $local_i$ of a processor $p_i$.
That is, if $\localInvariants(i)$ is false in line~\ref{alg:SSVC:removeStaleInfo}, then processor $p_i$ calls $\resetLocal_i()$.
We show that for all functions of Algorithm~\ref{alg:SSVC} that include $local_i$ of a processor $p_i$ in their input, $\localInvariants(i)$ holds (lemmas~\ref{lem:revive}--\ref{lem:restart}).

We also give the conditions for an execution to be legal.
To that end, we show that Requirement\reqs\ is possibly violated in a step where a processor calls $\resetLocal()$ (Remark~\ref{rem:restartBreaksReqs}) and definitely violated when a processor calls $revive()$ in two or more steps in an execution (Remark~\ref{rem:twoReviveBreakReqs}).
Also, we define the predicate $\globalInvariants(R,c)$ for an execution $R$ and a state $c\in R$, which gives the invariants that should hold for every active processor in $R$, so that no step includes a call to $\resetLocal()$ in line~\ref{alg:SSVC:receiveResetPairs}.
Finally, in Lemma~\ref{lem:LE}, we prove that given the bounds $B_{restart}(R)$ and $B_{revive}(R)$ on the number of steps that include a call to $\resetLocal()$ or $revive()$ in an $\pinf$-scale execution $R$, there exists at least one legal subexecution $R^*$ of $R$, such that $|R^*| \nll |R|$ holds under the condition that $B_{restart}(R)\ll |R| \land B_{revive}(R)\ll |R|$ holds.
In sections~\ref{s:PSVC} and~\ref{s:boundRestart}, we prove that the bounds $B_{restart}(R)$ and $B_{revive}(R)$ indeed exist for every $\pinf$-scale execution $R$ (and show that Algorithm~\ref{alg:SSVC} is practically-self-stabilizing).

\begin{definition}[\textbf{The $\localInvariants()$ predicate}]
\label{def:pairInvariants}
Let $R$ be an execution of Algorithm~\ref{alg:SSVC}, $c \in R$ be a system state, and $p_i\in P$.
We say that \textit{the local invariants hold for $p_i$ in $c\in R$}, if and only if, $\localInvariants(i) := \mirroredLocal_i() \land \lblOrdrd_i(local_i)$ (line~\ref{alg:SSVC:removeStaleInfo}) holds.
\end{definition}

\begin{lemma}
\label{lem:revive}
Let $a_x$ be a step in $R$ in which processor $p_i$ calls $\revive_i(local_i)$ when executing line~\ref{alg:SSVC:checkIncrExh}, \ref{alg:SSVC:doForeverRevive}, or~\ref{alg:SSVC:reviveCall} and suppose that $\localInvariants(i)$ holds before $p_i$ executes $\revive_i(local_i)$. 
Then, $\localInvariants(i)$ also holds in the state that immediately follows $a_x$.
\end{lemma}

\begin{proof}
Recall that the function $\revive_i(local_i)$ cancels $local_i$'s labels and returns $\la \getLabel_i()$, $local_i.curr.m$, $local_i.curr.m\ra$, $local_i.curr \ra$ (line~\ref{alg:SSVC:reviveEnd}).
Let $local_i$ $=$ $\la\la \ell_{old}$, $m_{old}$, $o_{old}\ra$, $prev\ra$ be the value of $local_i$ before $p_i$ calls $revive_i()$ and $local_i$ $=$ $\la \la \ell_{new}, m_{old}, m_{old}\ra, \la \ell_{old}, m_{old}, o_{old} \ra \ra$, be the value of $local_i$ after $p_i$ calls $revive_i()$.

Since $\localInvariants(i)$ $=$ $\mirroredLocal_i()$ $\land$ $\lblOrdrd_i(local_i)$ holds for $local_i$ $=$ $\la\la \ell_{old}$, $m_{old}$, $o_{old}\ra$, $prev\ra$ (Condition~\ref{eq:labelsOrdered} and line~\ref{alg:SSVC:mirroredLocalLabels}), 
the following hold for $local_i =$ $\la \la \ell_{new}, m_{old}, m_{old}\ra, \la \ell_{old}, m_{old}, o_{old} \ra \ra$ (i.e., after $p_i$ calls $revive_i()$): 
\begin{enumerate}[(i)]
\item $\isStored_i(local_i.prev.\ell)$ $=$ $\isStored_i(\ell_{old})$ holds, since ($\mirroredLocal_i()$ holds for $local_i$ before calling $\revive_i()$,
\item $local_i.curr.\ell$ $=$ $\ell_{new}$ $=$ $\getLabel_i()$ holds, by $\revive_i()$'s definition (lines~\ref{alg:SSVC:reviveBegin}--\ref{alg:SSVC:reviveEnd}), and
\item $(local_i.prev.\ell$ $\lb$ $local_i.curr.\ell$ $\land$ $\isCanceled_i(local_i.prev.\ell)$ $=$ $\ell_{old}$ $\lb$ $\ell_{new}$ $\land$ $\isCanceled_i(\ell_{old})$ holds, again by $\revive_i()$'s definition.
\end{enumerate}
\end{proof}

\begin{lemma}
\label{lem:merge}
Let $m_j$ $=$ $\la\bullet$, $\la arriving_j$, $\receiverLocal_j \ra \ra$ be a message that $p_i$ received from $p_j$ in step $a_i \in R$,  and $c$, $c'$ are system states in $R$, such that $(c$, $a_i$, $c'$, $\bullet)$ $\seg$ $R$.
Suppose that $\mirroredLocal_i()$ $\land$ $\lblOrdrd_i(local_i)$ $\land$ $\equalStatic_i(local_i, \receiverLocal_j)$ $\land$ $\legitArriving_i(m_j$, $arriving_j.curr.\ell)$ $\land$ $\pairInvariants_i(arriving_j)$ $\land$ $\labelCheck_i(local_i$, $arriving_j)$ hold in $c$.
Then, $\localInvariants(i)$ holds in $c'$, i.e., after the execution of line~\ref{alg:SSVC:mergeArrWithLocal} which calls $merge_i(local_i$, $arriving_j)$ and updates $local_i$.
\end{lemma}

\begin{proof}
Let $m_j = \la\bullet$, $\la arriving_j$, $\receiverLocal_j \ra \ra$ be a message that $p_i$ receives from $p_j$ in step $a_i$ and assume that 
$ \mirroredLocal_i() \land \lblOrdrd_i(local_i) \land \equalStatic_i(local_i, \receiverLocal_j) \land \legitArriving_i(m_j, arriving_j.curr.\ell) \land \pairInvariants_i(arriving_j)$ $\land \labelCheck_i(local_i, arriving_j)$ hold in $c$ with respect to $local_i$ and $m_j$.
For brevity, we denote $\xi_i := \isStored_i(local_i.prev.\ell) \land local_i.curr.\ell = \getLabel_i() \land \lblOrdrd_i(local_i)$.
Observe that $merge_i(local_i, arriving_j)$, initializes $\out_i$ either to $local_i$ or to $arriving_j$ (lines~\ref{alg:SSVC:pivot}--\ref{alg:SSVC:mergeInit}).
Then, $\out_i.curr.m[k]$ is updated with $newEvents$, for each $k\in \{1,\ldots, N\}$, and the result is returned and saved to $local_i$, hence lines~\ref{alg:SSVC:newEventsStart} to \ref{alg:SSVC:newEventsEnd} do not change the value of $\xi_i$. 
Therefore, we show that for each of the three different cases in which a pivot exists (Figure~\ref{fig:mergePairs}), $\xi_i$ holds after $local_i$ is updated with $merge_i(local_i, arriving_j)$.

\Paragraph{Case of Figure~\ref{fig:noWrapAround}}
In this case $local_i.itm$ and $arriving_j.itm$ match in label and offset for each $itm \in \{curr, prev\}$, and $\out_i$ is initialized to $local_i$, for which $\xi_i$ holds.
Also, no new label is processed by the labeling scheme ($labelBookkeeping_i(m_j,j)$ in line~\ref{alg:SSVC:callDolevArriving}), hence the return value of $\getLabel_i()$ remains the same (in $c$ and $c'$) after the execution of line~\ref{alg:SSVC:callDolevArriving} in step $a_i$.
Therefore, $\isStored_i(\out_i.prev.\ell) \land \out_i.curr.\ell = \getLabel_i() \land \lblOrdrd_i(\out_i)$ holds, since in $c$ and before $p_i$ calls $merge_i()$ in step $a_i$, $local_i.itm =_{\ell,o} \out_i.itm$ holds for each $itm \in \{curr, prev\}$, and $\xi_i = \isStored_i(local_i.prev.\ell) \land local_i.curr.\ell = \getLabel_i() \land \lblOrdrd_i(local_i)$ also holds.

\Paragraph{Case of Figure~\ref{fig:concurrentWrapAround}}
Let 
$\ell_{loc} := local_i.curr.\ell$, 
$\ell_{arr} := arriving_j.curr.\ell$, 
$\ell_{prv} := local_i.prev.\ell = arriving_j.prev.\ell$, and
$\ell_{max} := \max_{\lb}\{\ell_{loc}, \ell_{arr}\}$ in state $c$.
Note that $\ell_{max}$ exists, since in $c$ (and thus in $a_i$), $\labelCheck_i(local_i, arriving_j)$ holds, which implies that $\compLbls_i(\cX)$ holds, where $\cX$ includes the labels in $local_i$ and $arriving_j$.
Observe that in $a_i$, $merge_i(local_i$, $arriving_j)$ initializes $\out_i$ to the $<_{\ell,o}$-maximum pair in both $curr$ and $prev$ between $local_i$ and $arriving_j$ (lines~\ref{alg:SSVC:pivot}--\ref{alg:SSVC:mergeInit}), i.e., the pair that stores $\ell_{max}$ in its $curr.\ell$.
Thus, $\isStored_i(\out_i.prev.\ell)$ holds in $c'$, since $\isStored_i(\ell_{prv})$ and $\out_i.prev.\ell = \ell_{prv}$ hold in $c$.
Due to line~\ref{alg:SSVC:callDolevArriving}, $\ell_{arr}$ is stored in the variables of the labeling algorithm in $c$, hence $\ell_{max}$ is also stored in the variables of the labeling algorithm.
Therefore, by $\ell_{max}$'s definition, $\getLabel_i()$ returns $\ell_{max}$ in step $a_i$ and after the execution of line~\ref{alg:SSVC:callDolevArriving}, i.e., $\out_i.curr.\ell = \ell_{max} = \getLabel_i()$.
Moreover, $\mirroredLocal_i()$ holds for $local_i$, after $local_i$ is updated with $merge(local_i,arriving_j)$ in line~\ref{alg:SSVC:mergeArrWithLocal} during step $a_i$ (and hence holds in $c'$).

By the definition of the case of Figure~\ref{fig:concurrentWrapAround}, $local_i.curr.o \neq arriving_j.curr.o$ holds in $a_i$.
If $\ell_{arr} \lb \ell_{loc}$, then $\lblOrdrd_i(local_i)$ holds in $c'$, since $\out_i.curr.\ell = \ell_{loc}$, $\out_i.prev.\ell = \ell_{prv}$, and $\lblOrdrd_i(local_i)$ holds in $c$.
Otherwise, if $\ell_{max} = \ell_{arr}$, then $\out_i.prev.\ell = \ell_{prv} \lb \ell_{max} = \out_i.curr.\ell$ holds in $c'$.
Also $\isCanceled_i(\ell_{prv})$ holds in $c'$, since either $\isCanceled_i(\ell_{prv})$ holds in $c$ or $\ell_{prv}$ is canceled in step $a_i$ by the maximal label $\ell_{arr}$.
Therefore, $\lblOrdrd_i(local_i)$ holds in $c'$


%
%

\Paragraph{Case of Figure~\ref{fig:oneWrapAround}}
In this case we also use the definitions of $\ell_{loc}$, $\ell_{arr}$, and $\ell_{max}$ from the previous case (but here $\ell_{prv}$ is not common for $local_i$ and $arriving_j$).
If $local_i.curr.\ell = \ell_{max}$ in $c$, then $\out_i$ is initialized to $local_i$.
Thus, $\isStored_i(\out_i.prev.\ell) \land \out_i.curr.\ell = \getLabel_i() \land \lblOrdrd_i(\out)$ holds in the end of step $a_i$, since 
(i) $local_i.itm =_{\ell,o} \out_i.itm$, for each $itm \in \{curr, prev\}$, 
(ii) $\isStored_i(local_i.prev.\ell) \land local_i.curr.\ell = \getLabel_i() \land \lblOrdrd_i(local_i)$ holds before $merge_i()$ is called, and
(iii) line~\ref{alg:SSVC:callDolevArriving} does not change the return value of $\getLabel_i()$.
Note that this is the only case where $arriving_j.prev.\ell$ is not processed by the labeling algorithm, since it is a canceled label by $p_j$ that either 
(a) exists already in the variables of the labeling algorithm (hence, $\getLabel_i$ returns a larger label than $arriving_j.prev.\ell$ in $p_i$), or 
(b) it can be reused in case all processors that store it as canceled crash before it is introduced in the system by another processor.
Hence, $\mirroredLocal_i()$ holds in $c'$.

Otherwise, $\out_i$ is initialized to $arriving_j$, since $arriving_j.curr.\ell$ $=$ $\ell_{max}$.
In this case, $\isStored_i(\out_i.prev.\ell) \land \out_i.curr.\ell = \getLabel_i()$ holds, since
(i) $\out_i.prev.\ell = \ell_{loc}$ and $isStored_i(\ell_{loc})$ holds, and 
(ii) $\out_i.curr.\ell = \ell_{arr} = \ell_{max} \land \ell_{loc} \lb \ell_{arr}$, hence $\getLabel_i()$ returns $\ell_{arr}$ after the execution of line~\ref{alg:SSVC:callDolevArriving}.
Also, $\pairInvariants_i(arriving_j)$ holds, which implies that $\ell_{loc}$ $=$ $arriving_j.prev.\ell \preceq_{lb}$ $arriving_j.curr.\ell$ $=$ $\ell_{arr} = \ell_{max}$. 
Thus, it either holds that $arriving_j.prev.\ell$ $=$ $arriving_j.curr.\ell$ $=$ $\ell_{arr}$ and $\neg \isCanceled_i(\ell_{arr})$ (since $\getLabel_i() = \ell_{arr}$), or that $arriving_j.prev.\ell = arriving_j.curr.\ell \land \isCanceled_i(arriving_j.prev.\ell)$.
Therefore, $\lblOrdrd_i(arriving_j)$ holds, hence $\lblOrdrd_i(\out_i)$ holds in the end of step $a_i$, thus $\lblOrdrd_i(local_i)$ holds in $c'$.
\end{proof}


Note that during $increment_i()$, $p_i$ changes only in $local_i.curr.m[i]$ (line~\ref{alg:SSVC:increment}) and the other fields of $local_i$ stay intact.
In case $local_i$ is exhausted after that increment, $p_i$ calls $revive_i()$  (line~\ref{alg:SSVC:checkIncrExh}), hence by Lemma~\ref{lem:revive} the 
value of $\localInvariants(i)$ does not change
whenever $p_i$ calls $increment_i()$. 
By lemmas~\ref{lem:revive} and~\ref{lem:merge} we have the following. 

\begin{corollary}
\label{cor:reviveMergeIncrementAndPivot}
Let $R$ be an execution, $p_i, p_j\in P$, and $c_k\in R$ be a system state, which is followed by a step in which $p_i$ calls $\revive_i(local_i)$, or $increment()$, or $merge_i(local_i,arriving_j)$.
If $\localInvariants(i)$ holds in $c_k$, then $\localInvariants(i)$ holds also in $c_{k+1}$.
\end{corollary}


Lemma~\ref{lem:restart} considers the case in which a processor calls $\resetLocal_i()$ (line~\ref{alg:SSVC:removeStaleInfo} or~\ref{alg:SSVC:receiveResetPairs}).


\begin{lemma}
\label{lem:restart}
Let $local_i$ $=$ $\la y, y\ra$, where $y=\la \getLabel()$, $zrz$, $zrz\ra$, be the value of $local_i$ after a call to $\resetLocal_i()$ in line~\ref{alg:SSVC:removeStaleInfo} or~\ref{alg:SSVC:receiveResetPairs}, in a step $a_k\in R$. 
Then, $\localInvariants(i)$ holds for $local_i$ in the state $c_{k+1}$ that immediately follows $a_k$.
\end{lemma}

\begin{proof}
Since $labelBookkeeping_i()$ is called (lines~\ref{alg:SSVC:callDolevDoForever} and~\ref{alg:SSVC:callDolevArriving}) before each call of $\resetLocal_i()$ (lines~\ref{alg:SSVC:removeStaleInfo} and~\ref{alg:SSVC:receiveResetPairs}) and no other function in the lines between the call to $labelBookkeeping_i()$ and $\resetLocal_i()$ changes the variables of the labeling algorithm (lines~\ref{alg:SSVC:storeArriving}--\ref{alg:SSVC:exhArrCheck}), $\getLabel_i()$ returns the maximal label $\ell_{max}$ stored by the labeling algorithm and $local_i.curr.\ell = local_i.prev.\ell = \ell_{max}$.
Hence, $\isStored_i(local_i.prev.\ell) \land local_i.curr.\ell = \getLabel_i()$ hold, and therefore $\mirroredLocal_i()$ holds (cf. condition~\ref{eq:labelsOrdered}).
Also, $\lblOrdrd_i(local_i)$ holds (condition~\ref{eq:labelsOrdered}), since $(local_i.prev.\ell = local_i.curr.\ell \land \neg \isCanceled_i(local_i.curr.\ell)$ holds.
\end{proof}


\paragraph{Conditions for an execution to be legal}
So far in this section, we have shown that $\localInvariants(i)$ holds for the output of every function of Algorithm~\ref{alg:SSVC} that a processor $p_i$ applies on $local_i$.
However, in order to compute queries about the number of events on a single processor between two states (Requirement\reqs) or to the query $\causalPrecedence(local_i, local_j)$,
 we need to compare two vector clock pairs, i.e., pairs that appear in different processors or different system states.
To that end, we present the conditions under which Requirement\reqs\ breaks (remarks~\ref{rem:restartBreaksReqs} and~\ref{rem:twoReviveBreakReqs}).
Then, in Lemma~\ref{lem:LE} we present the conditions under which an execution is legal (Requirement\reqs\ holds), i.e., we show the conditions under which different vector clock pairs can be compared for computing correctly queries about counting events (and by Property~\ref{req:3causality} causal precedence).


\begin{remark}[$\resetLocal()$ breaks Requirement\reqs]
\label{rem:restartBreaksReqs}
We remark that it is possible that Requirement\reqs\ does not hold immediately after the execution of $\resetLocal()$ (lines~\ref{alg:SSVC:removeStaleInfo} and~\ref{alg:SSVC:receiveResetPairs}).
Since after executing $\resetLocal()$ all values in the main and offset of $local_i.curr$ and $local_i.prev$ are set to zero, it is possible to miscounting events when comparing two pairs in the states of active processors.
That is,
\remove{$p_i$ can miscount the events it records for a processor $p_j$, due to the reset to zero (violation of Requirement~\ref{req:correctCounting}), and} 
$p_i$ can miscount its own events when its entry $local_i.curr.m[i]$ is set to zero after a call to $\resetLocal_i()$\remove{ (violation of Requirement~\ref{req:2act})}, except for the case when $local_i$ remains the same before and after the call to $\resetLocal_i()$.
\qed\end{remark}

Consider a message $m_{j,i} = \la \bullet, \la arriving_j, \receiverLocal_j \ra \ra$ that a processor $p_i$ receives from a processor $p_j$, such that $\chi_{i,j} := \equalStatic_i(local_i, \receiverLocal_j) \land  \legitArriving_i(m_j,arriving_j.curr.\ell) \land \pairInvariants_i(arriving_j)$ does not hold  (line~\ref{alg:SSVC:exhArrCheck}).
Since such messages are not processed by Algorithm~\ref{alg:SSVC}, there is no call to $\resetLocal_i()$ or $revive_i()$ in the step that $p_i$ receives $m_{j,i}$, hence no immediate violation of Requirement\reqs.
However, the fields of $m_{j,i}$ that refer to the labeling algorithm's part of the message are processed by the labeling algorithm in $p_i$ in line~\ref{alg:SSVC:callDolevArriving}.
Hence, it is possible that the maximal label of $p_i$ has changed in the step where $p_i$ receives $m_{j,i}$, and that $p_i$ calls $\resetLocal_i()$ in its next step.
In  Section~\ref{s:boundRestart} (Lemma~\ref{lem:boundRestartLocal}) we show that for every $\pinf$-scale execution $R$, there exist bounds in the number of steps that include a call to $\resetLocal()$ or to $revive()$ that are significantly less than $|R|$.




\begin{remark}[Two calls to $revive()$ by the same processor break Requirement\reqs]
\label{rem:twoReviveBreakReqs}
Let $p_i$ be a processor, and $c_x, c_y$ be two states, such that there exist at least two steps between $c_x$ and $c_y$ in which $p_i$ called $revive_i$.
We remark that we cannot compute correctly the events that occurred between $c_x$ and $c_y$ by comparing $local_i^x$ and $local_i^y$, where $local_i^k$ is the value of $local_i$ in state $c_k$.
We explain why this holds in the following.

\begin{figure*}[t!]
\centering
\includegraphics[scale=1]{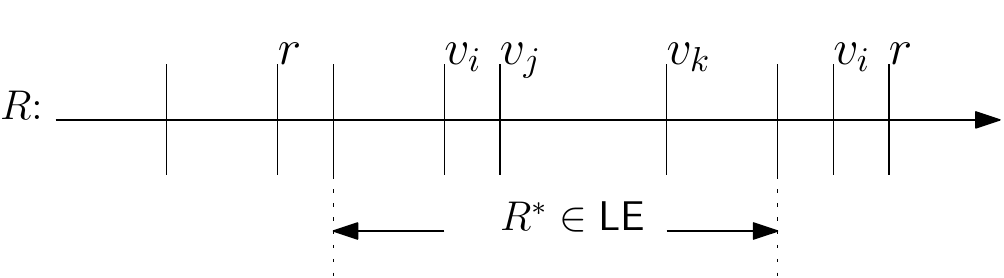}
\caption{Illustration of a legal execution (cf. Lemma~\ref{lem:LE}).
The horizontal line denotes an execution $R$ and the vertical lines highlight specific steps of $R$.
The vertical lines that are marked with $r$, denote a step in which a processor called $\resetLocal()$.
The vertical lines that are marked with $v_x$ denote a step in which a processor $p_x$ called $revive_x()$.
In this example for the segment of $R$, marked as $R^*$, the following hold:
(i) no processor called $\resetLocal()$, and 
(ii) $p_i$ (and every other active processor) called $revive()$ at most once.
Thus, by Lemma~\ref{lem:LE}, $R^*$ is a legal execution, i.e., $R^*\in\LE$.}
\label{fig:restartsLE}
\end{figure*}

Let $c_u$ be the first state after (the step in which) $p_i$ calls $revive_i()$ for the first time after $c_x$, and $c_v$ be the first state after $p_i$ calls $revive_i()$ for the first time after $c_u$.
By the vector clock pair construction and the definition of the function $revive()$ (Section~\ref{s:pair}), $local_i^u.prev$ is the value of $local_i.curr$ immediately before the first call to $revive_i()$ (after $c_x$).
Hence, the pivot item between $local_i^x$ and $local_i^u$ is $local_i^u.prev$, i.e., $(local_i^x.curr.\ell, local_i^x.curr.o) = (local_i^u.prev.\ell, local_i^u.prev.o)$.
Similarly,  $local_i^v.prev$ is the value of $local_i.curr$ immediately before the first call to $revive_i()$ after $c_u$.
Hence, the pivot item between $local_i^u$ and $local_i^v$ is $local_i^v.prev$, i.e., $(local_i^u.curr.\ell, local_i^u.curr.o) = (local_i^v.prev.\ell, local_i^v.prev.o)$.
Thus, the vector clock items $local_i^x.prev$ and $local_i^x.curr$, and specifically the events that $local_i^x.curr.m$ recorded in state $c_x$ do not appear in state $c_v$.
Therefore, irrespective of the number of calls to $revive_i()$ between $c_v$ and $c_y$, it is not possible to count the events in $p_i$ (i.e., the calls to $increment_i()$) between the states $c_x$ and $c_y$ by comparing $local_i^x$ and $local_i^y$.
\qed\end{remark}

In Definition~\ref{def:globalInvariants} we describe the conditions under which $\resetLocal()$ is never called in an execution.
Then, in Lemma~\ref{lem:LE} we give the conditions for an execution to be legal.
We also prove that for every $\pinf$-scale execution $R$, there exists a legal subexecution $R^*$ of $R$, such that $|R^*| \nll |R|$, under the conditions that there exist bounds $B_{restart}(R)$ and $B_{revive}(R)$ on the number of steps that include a call to  $\resetLocal()$, and respectively, $revive()$ in $R$, and $B_{restart}(R) \ll |R| \land B_{revive}(R) \ll |R|$ holds.
We illustrate Lemma~\ref{lem:LE} in Figure~\ref{fig:restartsLE}.

\begin{definition}[\textbf{The $\globalInvariants()$ predicate}]
\label{def:globalInvariants}
Let $R$ be an execution of Algorithm~\ref{alg:SSVC}, $c \in R$ a system state, $\varphi_i \equiv \mirroredLocal() \land \lblOrdrd(local_i)$ (line~\ref{alg:SSVC:removeStaleInfo}), 
and $\psi_{i,j} \equiv \labelCheck(local_i, arriving_j)$ (line~\ref{alg:SSVC:receiveResetPairs}), where $p_i, p_j \in P$ and $m_j= \langle \bullet, \la arriving_j, \receiverLocal_j\ra \rangle$ is a message in the communication channel from $p_j$ to $p_i$.
We define $\globalInvariants(R,c) := \forall_{p_i \in P(R), p_j\in P}\  \varphi_i \land (\neg\chi_{i,j} \lor \psi_{i,j})$.
We say that \textit{the global invariants hold during $R$}, if $\globalInvariants(R,c)$ holds for every $c\in R$.
\end{definition}
 
\begin{lemma}
\label{lem:LE}
Let $R$ be an $\pinf$-scale execution.
(I) For every subexecution $R^*$ of $R$, such that
\begin{enumerate}[(i)]
\item there is no step in $R^*$ in which a processor calls $\resetLocal()$, and

\item for every processor $p_i$ there exists at most one step $a_x \in R^*$ in which $p_i$ calls $revive()$ in $a_x$,

\end{enumerate}
$R^*\in\LE$ holds, i.e., $R^*$ is a legal execution.\\
(II) Moreover, let $B_{restart}(R)$ and $B_{revive}(R)$ be bounds on the number of steps in $R$ that include a call to $\resetLocal()$ and $revive()$, respectively, such that $B_{restart}(R) \ll |R| \land B_{revive}(R) \ll |R|$ holds.
Then, there exists at least one subexecution $R^*$ of $R$ such that $R^*\in\LE \land |R^*| \nll |R|$ holds.
\end{lemma}


\begin{proof}
%
We prove Part I of the lemma using remarks~\ref{rem:restartBreaksReqs} and~\ref{rem:twoReviveBreakReqs}.
For Part II we use the definitions of $\pinf$-scale executions and the $\ll$ relation (Section~\ref{s:systemSettings}), as well as, the pigeonhole principle.

\PARAgraph{Proof of Part I} Let $R^*$ be a subexecution of $R$, such that conditions (i) and (ii) of the lemma hold.
We show that Requirement\reqs\ holds throughout $R^*$.
Since no processor calls $\resetLocal()$ in $R^*$ (condition (i)) no event is ever lost in $R^*$.
That is, there is no step in which for a processor $p_i$, $local_i.curr.m \neq zrs$ holds and $p_i$ sets $local_i.curr.m$ to $zrs$ by calling $\resetLocal_i()$, where $zrs$ is the zero vector  (Remark~\ref{rem:restartBreaksReqs}).
Also, since no processor calls $\resetLocal()$ in $R^*$,  $\globalInvariants(R^*,c)$ holds for every $c\in R^*$.
The latter implies that every message that an active processor receives in $R^*$ is either discarded or contains a pair that is merged with the local one.

Recall that by condition (ii) of the lemma, every processor calls $revive()$ in $R^*$ at most once.
Thus, it is always possible to compute the number of events that occurred in each processor between two states in $R^*$.
That is, the query $V_i^y[i] - V_i^x[i]$ (number of events in $p_i$ from state $c_x$ to state $c_y$), for every $p_i$ that is active in $R^*$, is computed by the first two cases of Equation~\ref{eq:eventCountQuery} (Section~\ref{s:pair}), since there is always a pivot item between $local_i^x$ and $local_i^y$ in $R^*$ (i.e., the return value is never $\bot$).
Moreover, by the fact that $\globalInvariants(R^*,c)$ holds for every $c\in R^*$, we have that it is possible to merge every two pairs in $R^*$.
Hence, Property~\ref{req:3causality} holds (i.e., we can compute correctly the query $\causalPrecedence()$), since $\existsOverlap()$ is true in the computation of the query in Section~\ref{s:pair}.

\PARAgraph{Proof of Part II}
Since $R$ is of $\pinf$-scale, and $B_{restart}(R) \ll |R| \land B_{revive}(R) \ll |R|$ holds, (by the pigeonhole principle) there exists at least one segment $R^*$ of $R$ in which conditions (i) and (ii) of the lemma hold, such that $|R^*| \nll |R|$.
Thus, by Part I of the lemma, $R^*\in\LE$.
In fact, the maximal such $R^*$ is of size at least $|R|/(B_{restart}(R) + B_{revive}(R))$, hence $|R^*|\nll \pinf$ indeed holds (cf. Section~\ref{s:systemSettings}).
 \end{proof}
 
By Lemma~\ref{lem:LE}, we need to show that for every $\pinf$-scale execution $R$ the bounds $B_{restart}(R)$ and $B_{revive}(R)$ of the lemma statement indeed exist and also that $B_{restart}(R) \ll |R| \land B_{revive}(R) \ll |R|$ holds.
 We do that in Sections~\ref{s:PSVC} and~\ref{s:boundRestart}.
 In the end of Section~\ref{s:boundRestart} we complete the proof by showing that Algorithm~\ref{alg:SSVC} is practically-self-stabilizing.


\subsection{Pair evolution graph and function causality}
\label{s:PSVC}

In this section we establish that a call to $\resetLocal()$ in a step $a_x$ of an execution $R$ is \textit{caused} only due to either 
(i) stale information that resided in the system in the starting configuration, or
(ii) a call to $\resetLocal()$ in a step that precedes $a_x$, or
(iii) a call to $revive()$  in a step that precedes $a_x$.
To that end, we define a notion of \textit{function causality} between functions that processors call in $R$, which bases on a graph that relates pairs when they are either in the input or output set of a function that is called during a step of $R$.
We refer to that graph as the \textit{pair evolution graph} and note that it is an illustration of the interleaving model (Section~\ref{s:systemSettings}).

Our aim is to highlight all the changes that occur to any pair during an execution due to functions that processors apply on these pairs in the steps they take, as well as the relations between these functions.
The illustration that we bring resembles Lamport's \textit{happened before} relation~\cite{DBLP:journals/cacm/Lamport78}. 
In our work, we study the events that can \textit{cause} a call to the function $\resetLocal()$, rather than just the order in which the events occur. 
In the following paragraphs, we gradually define the pair evolution graph by identifying the functions that can be applied on a pair, the transition from the input to the output pair when applying a function, and the pairs that appear in the system throughout an execution.
We then define function causality (Definition~\ref{def:cPcauses}), basing on the pair evolution graph.


\paragraph{Functions called during a step.} 
We start by listing the functions that a processor may call during a step.
Let $R$ be an arbitrary execution of Algorithm~\ref{alg:SSVC} and $(c_x, a_x, c_{x+1})\seg R$ be a subexecution of $R$, such that processor $p_i\in P$ takes step $a_x$.
During the step $a_x$ and by the definition of the interleaving model (Section~\ref{s:systemSettings}), $p_i$ can either
\begin{enumerate}[(1)]
\item run lines~\ref{alg:SSVC:doForeverStart}--\ref{alg:SSVC:doForeverRevive} and send one message (out of $N-1$) to another processor due to line~\ref{alg:SSVC:doForeverEnd}, or
\item send one message (out of at most $N-2$ remaining messages) to another processor due to line~\ref{alg:SSVC:doForeverEnd}, or
\item run the message arrival procedure in lines~\ref{alg:SSVC:messageReceiveStart}--\ref{alg:SSVC:reviveCall}.
\end{enumerate}
After giving some insights on the send operation and $labelBookkeeping()$, we detail the functions that $p_i$ can call during $a_x$.

According to the interleaving model (Section~\ref{s:systemSettings}), each step includes a single send or receive operation.
Hence, a complete iteration of Algorithm~\ref{alg:SSVC}'s do-forever loop (lines~\ref{alg:SSVC:doForeverStart}--\ref{alg:SSVC:doForeverEnd}) requires $N-1$ steps (not necessarily consecutive), due to the $N-1$ messages to be send to the processor's neighbors.
In detail, we assume that when a processor $p_i\in P$, runs line~\ref{alg:SSVC:doForeverEnd}, it calls the function $\clone_i(local_i)$, which creates a separate copy of  $local_i$ that is then used in every of the $N-1$ calls of $encapsulate_i(local_i)$ and remains intact during those calls, regardless of the changes that occur to $local_i$ after the first (out of $N-1$) send operation.
We assume $p_i$ automatically discards the output pair of $\clone_i(local_i)$ in the last of the $N-1$ steps of that send operation.
These $N-1$ steps can be interleaved with steps of other processors or with steps in which $p_i$ runs the message arrival procedure (lines~\ref{alg:SSVC:callDolevArriving}--\ref{alg:SSVC:reviveCall}) and possibly changes $local_i$, but not the copy of $local_i$ that is used to complete the send operation.


Thus, during a step a processor can call functions from  $\F_1 := \{increment()$, $revive()$, $labelBookkeeping()$, $\resetLocal()$, $\clone()$, $encapsulate()\}$ in case (1), 
$\F_2 := \{encapsulate()\}$ in case (2), and
$\F_3 := \{labelBookkeeping()$, $merge()$, $revive()$, $\resetLocal()\}$ in case (3).
We define $\F = \F_1\cup \F_2 \cup \F_3$ to be the set of functions that a processor can call during a step. 

\paragraph{Transitions.} 
We define the notion of transitions to denote the application of a single function on a pair during an execution.
We say that $(Z,f,Z')$ is a \textit{transition} in $R$, if there exists a step $a_i \in R$ of a processor $p_i\in P$ and a function $f\in \F$, such that $p_i$ calls $f(Z,\bullet)$ in step $a_i$ with output $Z'$.
In this paragraph, we list all possible transitions for every function $f\in \F$.
In the following paragraphs, we define the pair evolution graph of an execution, basing on the set of all transitions that occurred during that execution.

\subparagraph{Transitions of pairs that stay intact between consecutive steps.}
We define the  transition $(Z,\lambda,Z)$, which denotes that the pair $Z$ (either in a communication channel or in $local_j$ of a processor $p_j\in P$) remained intact between a step $a_x$ and the beginning of consecutive step, $a_{x+1}$.

\subparagraph{Transitions due to a call to $labelBookkeeping()$.}
We define the transition $(Z$, $labelBookkeeping()$, $Z')$ to denote a call to $labelBookkeeping_i()$ that does not change the state of the labeling algorithm, and distinguish the following cases.
When $p_i$ calls $labelBookkeeping_i()$ in the do-forever loop (line~\ref{alg:SSVC:callDolevDoForever}), we consider the transition $(local_i, labelBookkeeping(), local_i)$, since $local_i$ stays intact after the call to $labelBookkeeping_i()$ ends.
Moreover, when $p_i$ calls $labelBookkeeping_i()$ in the message arrival procedure for a message $m = \la\bullet$, $arriving_j\ra$ (line~\ref{alg:SSVC:callDolevArriving}), we consider the transition $(arriving_j, labelBookkeeping(), local_i)$, since information from $m$ is incorporated to the local label storage.


We consider the cases in which $p_i$ possibly changes the state of the labeling algorithm during a call to $labelBookeeping_i()$  by either 
\begin{enumerate}[(i)]
\item canceling a label and creating or recycling another label  during a call to $revive_i()$ (in fact $revive()$ calls $\cancelPairLabels()$ in line~\ref{alg:SSVC:reviveCancelLbls}, which includes a call to $labelBookkeeping_i()$ in line~\ref{alg:SSVC:resetLabelPairs}), or 
\item discovering stale information in the label storage in line~\ref{alg:SSVC:callDolevDoForever}, or
\item receiving a new label during the message arrival procedure in line~\ref{alg:SSVC:callDolevArriving}.
\end{enumerate}
We denote any of the changes in the state of the labeling algorithm that are stated above with the (abstract) function $\newLabel()$ and remark that whenever a processor calls $\newLabel()$, the labeling algorithm deviates from its abstract task (cf. Section~\ref{s:DolevPaper}).


We define transitions for cases (i)--(iii) as follows:
(i) $(local_i$, $labelBookkeeping()\circ\newLabel_i()$, $local_i)$ refers to the case where $p_i$ calls $labelBookkeeping_i()$ during a call to $revive_i()$ (lines~\ref{alg:SSVC:reviveBegin}--\ref{alg:SSVC:reviveEnd})), which includes a call to $\newLabel_i()$, 
(ii)~$(local_i$, $labelBookkeeping()\circ\newLabel()$, $local_i)$ refers to the case where $p_i$ calls $labelBookkeeping_i()$ in line~\ref{alg:SSVC:reviveCancelLbls}, which includes a call to $\newLabel_i()$, and
(iii)~$(arriving_j$, $labelBookkeeping()\circ\newLabel()$, $local_i)$ refers to the case where $p_i$ calls $labelBookkeeping_i()$ in line~\ref{alg:SSVC:callDolevArriving} due to an arriving message $m = \la\bullet$, $arriving_j\ra$, which includes a call to $\newLabel_i()$.
Moreover, a call to $revive()$ on $local_i$ is illustrated by the transition $(local_i$, $labelBookkeeping()\circ\newLabel()$, $local_i)$, denoting the call to $\cancelPairLabels_i()$, 
followed by the transition $(local_i$, $revive()$, $revive_i(local_i))$, to denote the creation of $revive_i()$'s output pair.

\subparagraph{Transitions due to a call to $increment()$ or $revive()$.}
We illustrate a call to the $increment_i()$ function by $p_i$ through a number of transitions, depending on the value of $exhausted_i(local_i)$ in line~\ref{alg:SSVC:checkIncrExh}.
In case $exhausted_i(local_i)$ is false, then $increment_i()$ is only changing $local_i$ to a new value in line~\ref{alg:SSVC:increment}, say $local_i'$.
In this case, the transition $(local_i$, $increment()$, $local_i')$, captures all the changes that occurred to $local_i$ during the call to $increment_i()$.
Otherwise, if $exhausted_i(local_i)$ is true, then a call to $revive_i()$ (line~\ref{alg:SSVC:checkIncrExh}), follows the update from $local_i$ to $local_i'$.
In this case, we illustrate the call to $increment_i()$ with the three following transitions.
The first is the transition $(local_i$, $increment()$, $local_i')$, which indicates the change that occurs to $local_i$ in line~\ref{alg:SSVC:increment}.
The two following transitions are $(local_i'$, $labelBookkeeping()\circ\newLabel()$, $local_i')$ and $(local_i'$, $revive()$, $revive_i(local_i'))$, to denote the call to $revive_i()$ in line~\ref{alg:SSVC:checkIncrExh}.


\subparagraph{All possible transitions.}
We define the set of all transitions that are possible in a step $a_i\in R$ to be the set $\T_i := \cup_{f\in\F}\{ (Z$, $f$, $f(Z,\bullet)\}$ $\cup$ $\{(Z$, $\lambda$, $Z)$, $(Z$, $labelBookkeeping_i()\circ\newLabel_i()$, $Z'\}$, where $f(Z,\bullet)$ denotes the output pair of $f$ when $Z$ is part of its input.
Moreover, we define the set of all transitions that occur during a step to be the set $E_i(R) \subset \T_i$.
Given a transition $e = (Z, f, Z')$, we refer to $f$ as the \textit{tag} of $e$, and
the function $T_R : E(R) \to \F$ returns the tag $f$ of $e$, i.e., $T_R(e) = f$.

\begin{figure}[t!]
\qquad\qquad\includegraphics[scale=0.76]{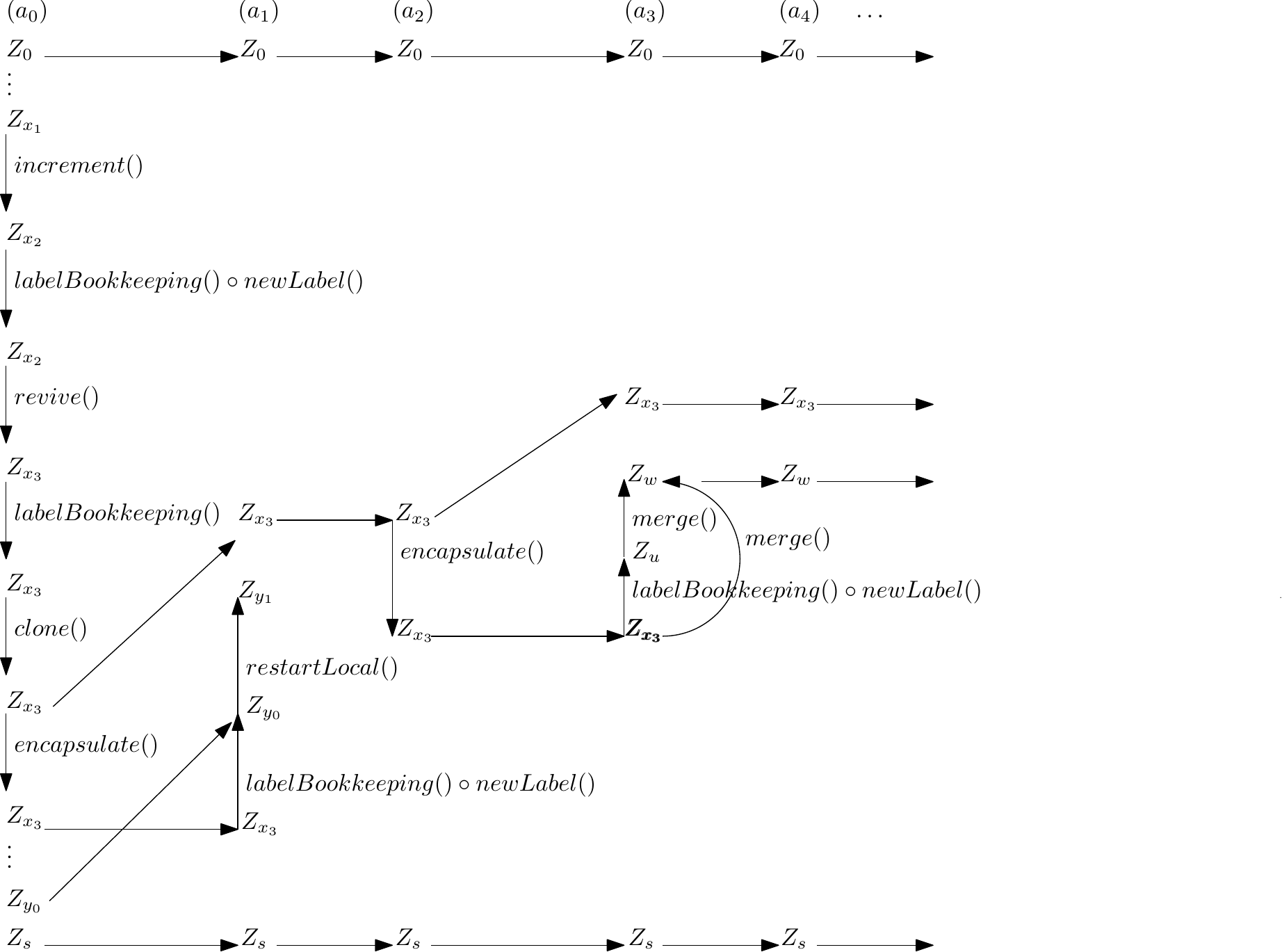}
\caption{\small An example of the pair evolution graph $\pgraph(R)$ of an arbitrary execution $R$.
For simplicity, we illustrate any edge $(Z,\lambda, Z')\in E(R)$ without its tag $\lambda$.
In step $a_0$, processor $p_i$ calls $increment$ on $Z_{x_1} = local_i$, which corresponds to the following three transitions:
$(Z_{x_1}$, $increment()$, $Z_{x_2})$ refers to line~\ref{alg:SSVC:increment} and
 $(Z_{x_2},$ $\lbnl$, $Z_{x_2})$ together with $(Z_{x_2},$ $revive()$, $Z_{x_3})$ refer to the call to $revive()$ in line~\ref{alg:SSVC:checkIncrExh}. 
Then, $p_i$ does an iteration of its do-forever loop (lines~\ref{alg:SSVC:doForeverStart}--\ref{alg:SSVC:doForeverEnd}), which ends with one send operation, say to processor $p_j$.
The latter corresponds to the transitions 
$(Z_{x_3},$ $\lB$, $Z_{x_3})$,
$(Z_{x_3},$ $\clone()$, $Z_{x_3})$, and
$(Z_{x_3},$ $encapsulate()$, $Z_{x_3})$.
In step $a_1$, $p_j$ receives $p_i$'s message (lines~\ref{alg:SSVC:messageReceiveStart}--\ref{alg:SSVC:reviveCall}), which cannot be merged with $Z_{y_0} = local_j$ due to label incomparability, and hence $p_j$ calls $\resetLocal_j()$ (line~\ref{alg:SSVC:receiveResetPairs}).
This step corresponds to the transitions 
$(Z_{x_3},$ $\lbnl$, $Z_{y_0})$ and
$(Z_{y_0},$ $\resetLocal()$, $Z_{y_1})$.
In step $a_2$, $p_i$ does one more send operation (line~\ref{alg:SSVC:doForeverEnd}), say to processor $p_k$, which corresponds to the transition $(Z_{x_3},$ $encapsulate()$, $Z_{x_3})$.
Then, in step $a_3$, $p_k$ receives $p_i$'s message and merges $Z_{x_3}$ to $Z_{u} = local_k$ (lines~\ref{alg:SSVC:messageReceiveStart}--\ref{alg:SSVC:reviveCall}).
This step corresponds to the transitions 
$(Z_{x_3},$ $\lbnl()$, $Z_{u})$, 
$(Z_{u},$ $merge()$, $Z_{w})$, and
$(Z_{x_3},$ $merge()$, $Z_{w})$.}
\label{fig:pairEvolutionGraph}
\end{figure}

\paragraph{All pair values during an execution.}
Our definitions consider all the pair values that appear in the system during $R$. 
These values can appear in the data field of a message that resides in the communication channel, or in the state of a processor. 
Additionally, we consider values of pairs that appear temporarily during a step, because we are interested in the exact values that the algorithm functions compute, and thus we unseal the encapsulation of the step atomicity (Section~\ref{s:systemSettings}).
%
Thus, for a step $a_i\in R$ that follows the state $c_i$, we define $V_i(R)$ to be the collection (with duplicates) of pairs that includes:
(i) all pairs that appear in the state of every processor in $c_i$, and
(ii) all the pairs that are outputs of functions that are called during the step $a_i$.

Consider two pairs that are identical but appear either (a) in different states or (b) in the same state but one appears in the communication channel, and the other one appears either in a different message or in the processor's local pair.
Then these two pairs appear as different elements in $V_i(R)$. 
Moreover, since the output of $\clone()$, $encapsulate()$, and $labelBookkeeping()$ equals the input these pairs appear twice in $V_i(R)$.
At any time, the size of $V_i(R)$ is bounded by $N+ \msg$ plus the number of pairs that are outputs of the functions that a processor calls during step $a_i$.
This bound holds due to the fact that we have $N$ processors and $\msg = \capacity N(N-1)$ is the maximum capacity of pairs in the communication channels. 
%
Note that this definition of $V_i(R)$ makes the definition of pair evolution graphs more intuitive.
%

%



\paragraph{Pair evolution graph.} 
We define a graph that illustrates the evolution of all the pair values that appear in the system during $R$, according to the interleaving model.
In this graph, the vertices are all the pair values that appear in the data field of every message and the states of every processor (including the intermediate stages that steps use for their computation) of an execution $R$. 
The graph's edges are all the transitions 
between couples of pairs that occur during $R$.

We define the \textit{pair evolution graph} of an execution $R$ to be the directed and layered graph with tagged edges $\pgraph(R) = (V(R), E(R))$, where $V(R) = \cup_{a_i\in R} V_i(R)$, $E(R) = \cup_{a_i\in R} E_i(R)$, and an edge $(Z,f,Z') \in E(R)$ is a directed graph's edge $(Z,Z')$ tagged with $f$ (and hence denoted with a triple).
We say that $V_i(R) \subseteq V(R)$ is a \textit{layer} of $\pgraph(R)$, for every step $a_i$ of $R$. 
We illustrate an example of a pair evolution graph (and hence the transitions) in Figure~\ref{fig:pairEvolutionGraph}, and give some insights below. 

We do some observations for pair evolution graphs.
Let $=_{V(R)}$ be the relation that denotes the fact that two pairs are the same node in $\pgraph(R)$.
For example, it might be the case that $Z_1 = Z_2$ but $Z_1 \neq_{V(R)} Z_2$, due to the multiple copies of a single pair that are created in a send operation.
By the definition of $\pgraph(R)$,
all edges that are tagged with $\lambda$, connect only pairs of consecutive layers, i.e., for every $e = (Z, \lambda, Z') \in E_i(R)$, $Z' \in V_{i+1}(R)$ $\land$ $Z = Z'$ $\land$ $Z \neq_{V(R)} Z'$ holds and also $Z$ is not further processed in step $a_i$.
Also, all edges that are not tagged with $\lambda$, include pairs from the same layer, i.e., $\forall  e = (Z,t,Z') \in E(R)$ such that $t\neq \lambda$, there exists a step $a_i\in R$, such that $Z\in V_i(R) \land Z' \in V_i(R)$ holds. 
Moreover, there is no edge in $\pgraph(R)$, that includes pairs from non-consecutive layers, i.e., $\forall  e = (Z,t,Z') \in E_i(R),\, Z'\in V_i(R)\cup V_{i+1}(R)$ holds.  
We note that given a subexecution $R'\seg R$, the pair evolution graph of $R'$ is the subgraph of $\pgraph(R)$ that includes the layers corresponding to the steps of $R'$, i.e., $\pgraph(R')$.


\paragraph{Function causality for $\resetLocal()$.}
As we showed in Section~\ref{s:invariants}, $\resetLocal()$ and $revive()$ are the only two functions of Algorithm~\ref{alg:SSVC} that can possibly violate the conditions for an execution to be legal.
In this paragraph, we define a notion of function causality with respect to $\resetLocal()$, and note that we deal with the case of $revive()$ in the following paragraph.
Our definition determines when a call to a function $f\in\F$ in a step of an execution $R$ \textit{causes} a  function call to $\resetLocal_i()$ in a subsequent step. 

We focus in a subset of $\F$, $\FF := \{\resetLocal(), revive()\}$, since as we will show in Section~\ref{s:boundRestart}, the number of calls to functions in $\FF$ has an effect on the number of subsequent calls to $\resetLocal()$.
In order to define function causality, we first define when two functions are adjacent or connected (i.e., there is a path that connects them in $\pgraph(R)$) in Definition~\ref{def:TRe}.

\begin{definition}[Adjacent and connected functions]
\label{def:TRe}
For two functions $f,g\in \F$, we say that $f$ \textit{is adjacent to} $g$ in $R$, if and only if,
$\exists i\in \N$, $e_1 =$ $(Z_1$, $f$, $Z_1') \in E(R)$, $e_2 = (Z_2$, $g$, $Z_2') \in E(R):\, T_{R}(e_1) = f$ $\land$ $T_{R}(e_2) =$ $g$ $\land$ $Z_1' =_{V(R)} Z_2$ holds, i.e., $(e_1, e_2)$ is a path in $\pgraph(R)$.
%
%
For $f,g\in\FF$, we say that $f$ \textit{is connected to} $g$ in $R$, if and only if, there exist $e_1, e_2, \ldots, e_x \in E(R)$, such that $T_{R}(e_1) = f$ $\land$ $T_{R}(e_x) = g$ $\land$ $\left(\land_{i=1,\ldots, x-1} (e_i \right.$ is adjacent to $\left. e_{i+1})\right)$. 
\end{definition}

Recall that we refer to all the variables of $Z$ except for $Z.curr.m$ as the static part of $Z$, since increments to $Z$ affect only $Z.curr.m$.
We say that a function $f$ leaves the static part of a pair $Z$ intact, if the static part of $f(Z,\bullet)$ equals the static part of $Z$.
Lemma~\ref{rem:causesOfRestart} shows which functions leave the static part of their input pairs intact and which don't.

\begin{lemma}
\label{rem:causesOfRestart}
The functions in $\F\setminus \FF$ leave the static part of at least one of their input pairs intact.
For each function in $\FF$, the input and output pairs may differ in their static parts.
\end{lemma}

\begin{proof}
The lemma statement holds for all functions in $\F\setminus (\FF\cup \{merge()$, $increment()\})$ $=$ $\{clone(), encapsulate(), labelBookkeeping()\}$, since these functions leave their input intact (recall that $labelBookkeeping()$ does not operate on a pair).
Since the output of $merge()$ equals, in its static part, to one of the two input pairs, the claim holds also for $merge()$ (cf. Section~\ref{s:pair}).
Note that a call to $increment()$ can include a call to $revive()$ (we study the case of $revive()$ below), however line~\ref{alg:SSVC:increment} does not change the static part of the input pair.

The functions in $\FF = \{revive(), \resetLocal()\}$, by their definitions, can possibly output pairs that have different static part from their input (cf. lines~\ref{alg:SSVC:resetLocal} and~\ref{alg:SSVC:reviveBegin}--\ref{alg:SSVC:reviveEnd}).
For the case of $revive()$, let $p_i$ be a processor, $Z_1 = local_i$, and $Z_2 = revive_i(Z_1)$.
Since $p_i$ cancels the labels of $Z_1$ and uses the new maximal label that the labeling algorithm returns as the label in $Z_2.curr.\ell$ (lines~\ref{alg:SSVC:reviveBegin}--\ref{alg:SSVC:reviveEnd}), the static parts of $Z_1$ and $Z_2$ are different.

We now study the case of $\resetLocal()$.
Recall from the definition of $\resetLocal()$ (line~\ref{alg:SSVC:resetLocal}) that immediately after a processor $p_i$ calls $\resetLocal_i()$, $local_i$ has the form $\la y, y\ra$, where $y = \la\getLabel_i(), zrs, zrs\ra$ and $zrs$ is the $N$-size zero vector.
Thus, the only case where $\resetLocal_i()$ leaves the static part of $local_i$ intact, is when $local_i = \la\la \ell_x, \bullet, zrs\ra, \la \ell_x, zrs, zrs\ra\ra$ holds immediately before $p_i$ calls $\resetLocal_i()$, where $\ell_x$ is $p_i$'s local maximal label.
The latter holds because $local_i$ equals $\la\la \ell_x, zrs, zrs\ra, \la \ell_x, zrs, zrs\ra\ra$ after $p_i$ calls $\resetLocal_i()$.
This is the case when $p_i$ receives a pair $arriving_j$ from a processor $p_j$, such that $\existsOverlap_i(local_i, arriving_j)$ does not hold and $local_i.curr.\ell$ remains $p_i$'s maximal label even after $p_i$ processes the labels in $arriving_j$.

For any other case except for the one described above one of the following is true for $local_i$ immediately before $p_i$ calls $\resetLocal_i()$:  $local_i.curr.\ell$ $\neq$ $local_i.prev.\ell$ or $s \neq zrs$, for at least one vector $s$ in $\{local_i.curr.m$, $local_i.prev.m$, $local_i.prev.o\}$.
For any of these cases and by the definition of $\resetLocal()$, immediately after $p_i$ calls $\resetLocal_i()$, $local_i$ has a different static part than immediately before $p_i$ called $\resetLocal_i()$.
\remove{
A processor $p_i$ calls the function $\resetLocal_i()$ in line~\ref{alg:SSVC:removeStaleInfo}, when $\mirroredLocal_i() \land \lblOrdrd_i(local_i)$ is false.
Thus, either if $\mirroredLocal_i()$ is false ($local_i.curr.\ell$ is not the maximal label in $p_i$) or if $\lblOrdrd_i(local_i)$ is false (the labels in $local_i$ are not ordered as in Equation~\ref{eq:labelsOrdered}), the output pair of $\resetLocal_i()$ is different than the input (at least) in $local_i.curr.\ell$.
Moreover, when $p_i$ calls the function $\resetLocal_i()$ in line~\ref{alg:SSVC:receiveResetPairs} when processing $arriving_j$, it is either the case that  $local_i.curr.\ell$ is larger than $arriving_j.curr.\ell$ or not.
In case  $arriving_j.curr.\ell \lb local_i.curr.\ell := \ell_x$ and $local_i = \la \la\ell_x, \bullet, zrs\ra, \la\ell_x, zrs, zrs\ra\ra$, then the output pair of $\resetLocal_i()$ in line~\ref{alg:SSVC:receiveResetPairs} will be $local_i = \la \la\ell_x, zrs, zrs\ra, \la\ell_x, zrs, zrs\ra\ra$, an will thus have the same static part with the input.
In case $local_i.curr.o \neq zrs$ then $\resetLocal_i()$ will indeed create a pair with different static part than its input.}
\end{proof}


In Definition~\ref{def:cPcauses} we define function causality between a call to a function in  $\{revive(), \resetLocal()\}$ and a subsequent call to $\resetLocal()$.
Let $p_{\resetLocal}(i,k) := \neg(\mirroredLocal() \land \lblOrdrd(local_i))$ (line~\ref{alg:SSVC:removeStaleInfo}) and $q_{\resetLocal}(i,j,k) := \neg\labelCheck(local_i, arriving_j)$ (line~\ref{alg:SSVC:receiveResetPairs}) be the predicates such that whenever either of them is true $p_i$ calls $\resetLocal_i()$ in a step $a_k\in R'$, where $arriving_j$ is the pair received by $p_i$ in a message from $p_j$ in step $a_k$.

\begin{definition}[$f$ causes $\resetLocal()$, for $f\in \{revive(), \resetLocal()\}$]
\label{def:cPcauses}
Let $|R'| \leq \MI$ be an execution, $p_i, p_j\in P$, $a_k\in R'$, and $\cP(i,j,k):=$ $p_{\resetLocal}(i,k)$ $\lor$ $q_{\resetLocal}(i,j,k)$.
Moreover, let $e$, $e'$ be two edges in $\pgraph(R')$, such that $T_{R'}(e) \in \{revive()$, $\resetLocal()\}$, $T_{R'}(e') = \resetLocal()$, $e\in E_r(R')$, $e'\in E_k(R')$, $r\leq k$, and $\cP(i,j,k)$ is true (in step $a_k$).
We say that $f := T_{R'}(e)$ \textit{causes} $\resetLocal()$ in $R'$, if and only if, the following hold:
\begin{enumerate}[(a)]
\item $f$ is connected to $\resetLocal()$ in $\pgraph(R')$ through a path $P = (e_1, \ldots, e_x)$, such that $e_1 = e$ and $e_x = e'$, 

\item the value of a predicate $\pi$ in $\cP(i,j,k)$ depends on a vector clock item $I_f$ in $f$'s output, and

\item for every edge $e_t\in P$, such that $e_t\notin\{e,e'\}$, it holds that the function $T_{R'}(e_t)$ does not change the label and the offset of $I_f$.
\end{enumerate}
\end{definition}

For example, in Figure~\ref{fig:pairEvolutionGraph} $revive()$ in step $a_0$ causes $\resetLocal()$ in step $a_1$, since the pair (and hence the vector clock items) that $p_i$ sent to $p_j$, was incomparable (no pivot existed) with $p_j$'s local pair, was created when $p_i$ called $revive()$ in $a_0$.

\paragraph{Function causality for $revive()$.}
In the following lemma, we show which functions in $\F$ can change the value of the predicate $exhausted(local)$ (Equation~\ref{eq:notExhausted} and lines~\ref{alg:SSVC:checkIncrExh}, \ref{alg:SSVC:doForeverRevive}, and~\ref{alg:SSVC:reviveCall}), and thus cause a call to $revive()$ (Lemma~\ref{cl:causesForRevive}).
We note that the analysis for $revive()$ (Lemma~\ref{cl:causesForRevive}) is much simpler than the one for $\resetLocal()$, because the condition for calling $revive()$, $exhausted(local)$, depends only on vector clock increments.
On the contrary, the conditions for calling $\resetLocal()$ (lines~\ref{alg:SSVC:removeStaleInfo} and~\ref{alg:SSVC:receiveResetPairs}) depend on every field of $local$, as well as on whether an arriving pair can be merged with the local one.

\begin{lemma}
\label{cl:causesForRevive}
Let $p_i\in P$.
(i) The value of the predicate $exhausted_i(local_i)$ (cf. Section~\ref{s:pair}) can change to true only due to a call to $increment_i()$ or $merge_i()$, or
it can be true in the starting system state due to stale information in $p_i$'s state. 
(ii) The value of $exhausted_i(local_i)$ does not change after $p_i$ calls a function in $\{labelBookkeeping_i()$, $clone_i()$, $encapsulate_i()\}$.
(iii) The value of $exhausted_i(local_i)$ is false after $p_i$ calls a function in $\FF = \{\resetLocal(), revive()\}$.
\end{lemma}

\begin{proof}
Recall that $exhausted(Z)$ $\Leftrightarrow$  $\Sigma_{k=1}^N (Z.curr.m[k] - Z.curr.o[k]) \geq$ $\MI - 1$ (Equation~\ref{eq:notExhausted}).
For part (i) of the claim, first note that the starting system state, $c_0$, of any execution, $R$, is arbitrary.
Hence, it can be the case that $exhausted_i(local_i)$ is true for $local_i$ in $c_0$.
The lemma statement holds for $increment()$, since by its definition (lines~\ref{alg:SSVC:incrementStart}--\ref{alg:SSVC:checkIncrExh}) it increases $local_i.curr.m$.
Similarly, $merge_i(local_i, arriving_j)$ outputs a pair that possibly includes more events that than $local_i$ and $arriving_j$ (cf. lines~\ref{alg:SSVC:merge}--\ref{alg:SSVC:mergeReturn} and Section~\ref{s:pair}).
Hence, it might be the case that $exhausted_i(local_i)$ is false before a call to $merge_i(local_i, arriving_j)$, but true for the new value of $local_i.curr.m$, when $p_i$ stores in $local_i$ the output of $merge_i()$.

For part (ii) of the claim, note that $labelBookkeeping_i()$, $clone_i()$, and $encapsulate_i()$ do not change $local_i.curr.m$.
Finally, part (iii) of the claim is true since by the definitions of $\resetLocal()$ (line~\ref{alg:SSVC:resetLocal}) and $revive()$ (lines~\ref{alg:SSVC:reviveBegin}--\ref{alg:SSVC:reviveEnd}), $local_i.curr.m = local_i.curr.o$ holds for their outputs, hence $exhausted_i(local_i)$ is false.
\end{proof}


\subsection{Bounding the number of deviations from the abstract task in an $\pinf$-scale execution}
\label{s:boundRestart}
%
In Lemma~\ref{lem:boundRestartLocal}, we show that the number of steps in which a processor calls $revive()$ or $\resetLocal_i()$ during an execution $R'$, such that $|R'| \leq \MI$, is significantly less than $|R'|$.
We focus in these two functions, because due to Section~\ref{s:PSVC}, only these two functions can cause a call to $\resetLocal()$ (cf. Lemma~\ref{rem:causesOfRestart} and Definition~\ref{def:cPcauses}).
Then, in Corollary~\ref{cor:practicallyStabilizing} we show that Algorithm~\ref{alg:SSVC} is practically-self-stabilizing (i.e., Theorem~\ref{thm:reqHold} holds).

\begin{lemma}
\label{lem:boundRestartLocal}
Let $R$ be an execution of Algorithm~\ref{alg:SSVC} and $R'$ be a subexecution of $R$, such that $|R'| \leq \MI$.
Then, the number of steps in which a processor calls either $revive()$ or $\resetLocal()$ in $R'$ is significantly less than $\MI$.
\end{lemma}

\begin{proof}
%
%
The proof focuses on giving a bound on the number of steps in which a processor calls $\resetLocal()$ in $R'$ and showing that the bound is significantly less than $\MI$.
As a by-product of this goal, Claim~\ref{cl:reviveBound} shows that the number of steps in $R'$ in which a processor calls $revive()$ is significantly less than $\MI$.

A processor can call $\resetLocal()$ either in line~\ref{alg:SSVC:removeStaleInfo} or in line~\ref{alg:SSVC:receiveResetPairs}. 
The proof considers both cases.
We first show that there can be at most one call per processor to $\resetLocal()$ during any execution due to line~\ref{alg:SSVC:removeStaleInfo}.
To prove this statement, first observe that the condition in line~\ref{alg:SSVC:removeStaleInfo} can be false due to stale information that resided in the processor's state in the starting state. 
However, by Corollary~\ref{cor:reviveMergeIncrementAndPivot}, for any function that changes $local$, it holds that the condition in line~\ref{alg:SSVC:removeStaleInfo} is false for the updated value of $local$.

In the remainder of this proof, we show that the number of steps in $R'$ that include a call to $\resetLocal()$ due to line~\ref{alg:SSVC:receiveResetPairs} is significantly less than $\MI$.
We first bound the maximum number of steps that include a call to $revive()$ (Claim~\ref{cl:reviveBound}), 
as well as the maximum number of labels that can exist during $R'$ (Claim~\ref{cl:newLabelBound}).
In claims~\ref{cl:linkStabilization} and~\ref{cl:mechanismBound} we bound the number of steps that include a call to $\resetLocal()$ in line~\ref{alg:SSVC:receiveResetPairs} due to the recovery of the link-layer algorithm~\cite{DBLP:conf/sss/DolevHSS12}, and respectively, the token-passing mechanism (Section~\ref{s:algorithms}).
Moreover, in Claim~\ref{cl:pairRecyclingBound} we bound the number of calls to $\resetLocal()$ in line~\ref{alg:SSVC:receiveResetPairs} that occur due to a single pair static part that appears in $R'$.
Finally, in Claim~\ref{cl:receiveRestartsBound} we show that these bounds imply that the number of steps that include a call to $\resetLocal()$ in line~\ref{alg:SSVC:receiveResetPairs} during $R'$ is significantly less than $|R'|$, by showing that the number of pair static parts that appear in $R'$ is significantly less than $\MI$ and combining Claims~\ref{cl:reviveBound}--\ref{cl:pairRecyclingBound}.


\begin{claim}
\label{cl:reviveBound}
The number of steps during $R'$ that include a call to $\revive()$ is at most $N + N^2 + N^3\cdot\capacity$.
\end{claim}

\begin{proof}[\textbf{\emph{Proof of Claim~\ref{cl:reviveBound}}}]
The proof considers the three causes for pair exhaustion during $R$ (cf. Lemma~\ref{cl:causesForRevive}). 
That is, due to calls to $increment()$ and $merge()$, as well as due to stale information that appeared in the starting system state.

Since $|R'| \leq \MI$, the maximum number of increments that can occur in $R$ is less than $\MI$.
Note that for a single vector clock pair exhaustion, at most all processors can wrap around concurrently.
This can occur when processor $p_j$ holds a pair value $Z_j$ in $local_j$ that is close to be exhausted, say, just one increment away (lines~\ref{alg:SSVC:incrementStart}--\ref{alg:SSVC:checkIncrExh}). 
Then, $p_j$ sends $local_j$'s value $Z_j$ to all other processors $p_k \in P$ in the system. Every processor $p_k$ that receives $Z_j$, merges it with $local_k$, and in the following step calls $increment_k()$, which leads to exhausting $local_k$. 
Hence, there can be at most $N$ steps in $R'$ that processors $p_k$ take, that include a call to $\revive_k()$ due to the exhaustion of a pair that was merged with $Z_j$. 
Indeed, $p_k$ can exhaust the output of $merge_k(local_k, Z_j)$ at most once, because the call to $\revive_k()$ produces a pair with a static part (and hence the labels $curr.\ell$ and $prev.\ell$) that is different than the one of $Z_j$ and $local_k$.

The remaining pair exhaustions can be only due to arbitrary values that resided in the starting system state (cf. Lemma~\ref{cl:causesForRevive}). 
At most $N$ such vector clocks have resided in the states of the processors, and at most $\msg\leq N^2\cdot\capacity$ resided in the communication channels.
Since for each of these $N + N^2\cdot\capacity$ pair values can lead to at most $N$ concurrent exhaustions, there can be $N^2 + N^3\cdot\capacity$ exhaustions due to pairs that come from the arbitrary starting state.

Note that we have counted the number of steps that include a call to $\revive()$ in two ways;
(i) calls to $increment()$ or $merge()$, and 
(ii) stale information that appeared in the starting system state. 
Of course, a pair can become exhausted due to a combination of these two causes. 
The arguments above hold for such combinations and the counting is correct because each pair exhausted is counted at least once. 
Therefore, in total, there can be at most $N + N^2 + N^3\cdot\capacity$ pair exhaustions that can occur during $R'$. 
Hence, at most that many calls to $\revive()$ in $R'$.
\end{proof}


\begin{claim}
\label{cl:newLabelBound}
The maximum number of labels that can exist in $R'$ (and hence the number of steps that include a call to the $\newLabel()$ function) is in $\bigO(\capacity N^3)$.
\end{claim}

\begin{proof}[\textbf{\emph{Proof of Claim~\ref{cl:newLabelBound}}}]
Recall that from the proofs of corollaries~\ref{cor:dolevLabelAdoptions}, \ref{cor:dolevLabelCreations} and~\ref{cor:labelingSchemeConvergence}, proving that the labeling algorithm of Dolev et al.~\cite{DBLP:journals/corr/DolevGMS15} is practically-self-stabilizing depends on the existence of a bound on the maximum number of labels, 
rather than the actual value of the bound.
We give a (polynomial) bound on the number of labels that exist during $R'$, which implies that the number of steps that include a call to $newLabel()$ (cf. Section~\ref{s:PSVC}) has the same bound.

By corollaries~\ref{cor:dolevLabelAdoptions} and \ref{cor:dolevLabelCreations} there can be at most $4N^2 + 4N\msg - 4N -2\msg$ 
labels in the system, where $\msg = \capacity N(N-1)$ is the maximum capacity of pairs in the communication channels.
Note that there can be at most $N + N^2 + N^3\cdot\capacity$ additional labels creations, due to calls to the $revive()$ function.
Thus, there can be at most $L := (N + N^2 + N^3\cdot\capacity) + (4N^2 + 4N\msg - 4N -2\msg) \in \bigO(\capacity N^3)$ labels in the system during $R'$, and hence at most that many calls to $\newLabel()$.
\end{proof}


In the labeling algorithm of Dolev et al.~\cite{DBLP:journals/corr/DolevGMS15}, each processor $p_i$ uses an $N$-size array of bounded FIFO queues, $storedLabels_i[]$, for keeping a label history.
The queue $storedLabels_i[j]$ stores the labels that $p_i$ has received that show $p_j$ as their creator, i.e., $\ell.\creator = j$ holds for every $\ell \in storedLabels_i[j]$ (cf. Section~\ref{s:DolevPaper}).
Recall that in Section~\ref{s:labelingAlgStabilizes} we extended the queue lengths for an execution of Algorithm~\ref{alg:SSVC} in which no processor calls $revive()$, to accommodate for the two labels that each pair includes.
By Claim~\ref{cl:newLabelBound}, we are able to extend the size of the label storage of the labeling algorithm, in order to accommodate for the extra label creations due to calls to the function $revive()$ in Algorithm~\ref{alg:SSVC}.
Thus, by Claim~\ref{cl:newLabelBound} and Section~\ref{s:labelingAlgStabilizes} we set $|storedLabels_i[j]| = L$, for every $p_i, p_j\in P$, where $L = (N + N^2 + N^3\cdot\capacity) + (4N^2 + 4N\msg - 4N -2\msg) \in \bigO(\capacity N^3)$.

\begin{claim}
\label{cl:linkStabilization}
There can be at most $(2\capacity + 1)N^2$ steps that include a call to $\resetLocal()$ in line~\ref{alg:SSVC:receiveResetPairs} due to the recovery of the link-layer algorithm~\cite{DBLP:conf/sss/DolevHSS12}.
\end{claim}

\begin{proof}[\textbf{\emph{Proof of Claim~\ref{cl:linkStabilization}}}]
Recall that the self-stabilizing link-layer algorithm of~\cite{DBLP:conf/sss/DolevHSS12}, which we rely on, requires at most $2\capacity + 1$ message arrivals per direction of a communication channel to stabilize.
Therefore, since there are $N(N-1)/2$ links in the system, where each of them is a bidirectional communication channel, there can be at most $2 \cdot (2\capacity + 1) \cdot N(N-1)/2 \leq (2\capacity + 1)N^2$ steps that include a call to $\resetLocal()$ due to stale information that,  at the starting system state, resides in the communication channels.
\end{proof}


\begin{claim}
\label{cl:mechanismBound}
Let $m_{j,i} = \langle \bullet, \la arriving_j, \receiverLocal_j\ra \rangle$ be a message that $p_i$ receives from $p_j$, via their communication channel, $channel_{j,i}$.
There can be at most $\capacity$ steps that include a call to $\resetLocal_i()$ in line~\ref{alg:SSVC:exhArrCheck}, 
due to stale information 
that appears in the field $\receiverLocal_j$ of $m_{j,i}$, 
where $m_{j,i}$ appears in $channel_{j,i}$ in the starting system state 
and $\receiverLocal_j$ sets $\equalStatic_i(local_i, \receiverLocal_j)$ to true. 
Hence, there can be at most $\msg\leq \capacity N^2$ such calls to $\resetLocal()$ in any execution. 
\end{claim}


\begin{proof}[\textbf{\emph{Proof of Claim~\ref{cl:mechanismBound}}}]
Notice that there can be at most $\capacity$ messages in the communication channel from $p_j$ to $p_i$, $channel_{j,i}$, at any time, and specifically in the starting system state. 
Each of these $\capacity$ messages in transit from $p_j$ to $p_i$ can possibly store a value of $\receiverLocal_j$, such that $\equalStatic_i(local_i, \receiverLocal_j)$ is true in line~\ref{alg:SSVC:exhArrCheck}, which leads to a call to $\resetLocal_i()$ in line~\ref{alg:SSVC:receiveResetPairs}.

We provide details about how this can occur. 
Consider a step $a_x \in R$ of $p_i$ that includes a call to $\resetLocal_i()$ due to a message arrival (line~\ref{alg:SSVC:receiveResetPairs}). 
Hence, $\equalStatic_i(local_i, \receiverLocal_j)$ was true during $a_x$. 
The fact that $a_x$ includes a call to $\resetLocal_i()$ does not cause the other $\capacity-1$ stale messages in the channel from $p_j$ to $p_i$ to be omitted. 
Therefore, there could be a subsequent step during which $\equalStatic_i(local_i, \receiverLocal_j)$ is true due to the other $\capacity-1$ messages in $channel_{j,i}$ that have stale information, since those messages appeared in the starting system state.


Such steps can be repeated at most $\capacity$ times for $channel_{j,i}$ and at most $\msg$ in total during $R'$, where $\msg$ is the number of messages in transit at any given time and hence in starting system state, $c_0$.
\end{proof}


\begin{claim}
\label{cl:pairRecyclingBound}
Let $L = (N + N^2 + N^3\cdot\capacity) + (4N^2 + 4N\msg - 4N -2\msg) \in \bigO(\capacity N^3)$ be the maximum number of labels that can appear in the system in $R'$ (Claim~\ref{cl:newLabelBound}).
During $R'$, there can be at most $2N\cdot L$ calls to $\resetLocal_i()$ in line~\ref{alg:SSVC:exhArrCheck} for every pair static part that appears in $local_i$ of a processor $p_i$.
Hence, for each pair static part that appears in the state ($local$) of a processor in $R'$, there can be at most $2N^2L$ calls to $\resetLocal_i()$ in line~\ref{alg:SSVC:exhArrCheck}.
\end{claim}

\begin{proof}[\textbf{\emph{Proof of Claim~\ref{cl:pairRecyclingBound}}}]
Let $p_i,p_j\in P$ be two processors and $m_{j,i} = \langle \bullet$, $\la arriving_j$, $\receiverLocal_j\ra \rangle$ be a message that $p_j$ sends to $p_i$ by adding it to $channel_{j,i}$.
Consider the case where the pair static part of $arriving_j$ sent by $p_j$ to $p_i$ causes $p_i$ in step $a_x$ to call $\resetLocal_i()$ and obtain $local_i = Z$.
Note that when referring to a value $Z$ or $Z_x$ that a variable takes, e.g., $local_i$, we treat $Z$ and $Z_x$ as (immutable) literals, i.e., pair values that do not change.
In Part I of the proof, we show that there can be at most one more call to $\resetLocal_i()$ (i.e., a total of at most two) due to receiving the same pair static part from $p_j$, before $p_j$ stores a pair with a different static part in $local_j$.
The proof relies on the token-passing mechanism (lines~\ref{alg:SSVC:doForeverEnd}, \ref{alg:SSVC:storeArriving}, and~\ref{alg:SSVC:exhArrCheck}).
In the proof of this claim, we assume that the token passing mechanism has stabilized (since Claim~\ref{cl:mechanismBound} has already showed that the token passing mechanism of lines~\ref{alg:SSVC:doForeverEnd}, \ref{alg:SSVC:storeArriving}, and~\ref{alg:SSVC:exhArrCheck} can cause an additive  (bounded) number of calls to $\resetLocal()$.
Then, in Part II of the proof, we show that each pair static part can be created by a processor at most $L$ times in $R'$.
We combine Part I and II to obtain the claim's bound.
Throughout the claim's proof, we denote with $\s(Z) = \la\la\ell_1, \bot, o_1\ra, \la\ell_2,  m_2, o_2 \ra \ra$ the static part of a pair $Z = \la\la\ell_1,m_1, o_1\ra, \la\ell_2,  m_2, o_2 \ra \ra$.

\paragraph{Part I} Recall from line~\ref{alg:SSVC:exhArrCheck} that for any message $m_{j,i} = \langle \bullet$, $\la arriving_j$, $\receiverLocal_j\ra \rangle$ that $p_j$ sends to $p_i$, $\equalStatic_i(local_i, \receiverLocal_j)$  has to be true for $p_i$ to process $arriving_j$.
Let $local_i = Z_{i_1}$ in state $c_x$ and suppose in the step $a_x$ that immediately follows $c_x$, processor $p_i$ calls $\resetLocal_i()$ in line~\ref{alg:SSVC:exhArrCheck} after receiving $m_{j,i} = \langle \bullet$, $\la Z_{j_1}$, $\receiverLocal_j\ra \rangle$, which produces $local_i = Z_{i_2}$.
Due to the token passing mechanism (lines~\ref{alg:SSVC:doForeverEnd}, \ref{alg:SSVC:storeArriving},  and~\ref{alg:SSVC:exhArrCheck}), $p_i$ can process a new message from $p_j$ in a step that follows $a_x$, only after $p_j$ receives $Z_{i_2}$ or a subsequent pair that appeared in $local_i$ after $a_x$ (possibly after receiving other pairs).
Let $a_{x'}$ be the first step after $a_x$ in which $p_j$ stores in $local_j$ a pair with static part different than $\s(Z_{j_1})$.
We show that there can be at most one step between $a_x$ and $a_{x'}$ (different than $a_x$ and $a_{x'}$), in which $p_i$ calls $\resetLocal_i()$ due to receiving a pair with static part equal to the one of $Z_{j_1}$.
Hence, there can be at most two such calls to $\resetLocal_i()$, until a state in which $p_j$ stores a pair in $local_j$ with static part different than $\s(Z_{j_1})$.

Recall that $local_i = Z_{i_2}$ is the value of $local_i$ in the state that immediately follows step $a_x$, in which $p_i$ calls $\resetLocal_i()$, and let $\ell_{i_2} = Z_{i_2}.curr.\ell = Z_{i_2}.prev.\ell$ for brevity.
Observe from the definition of $\resetLocal()$ (line~\ref{alg:SSVC:resetLocal}), that $\getLabel_i()$ returns $\ell_{i_2}$ according to a possible update of the local maximal label in line~\ref{alg:SSVC:callDolevArriving}.
 Then $\ell_{i_2}$ is equal to 
either  $Z_{i_1}.curr.\ell$ (Case a), 
or $Z_{j_1}.curr.\ell$ (Case b), 
or $\ell_{i_3}$ is different than both $Z_{i_1}.curr.\ell$ and $Z_{j_1}.curr.\ell$ (Case c).
The latter case refers to a situation in which $Z_{i_1}.curr.\ell$ and $Z_{j_1}.curr.\ell$ cancel each other and $p_i$ produces a new label in line~\ref{alg:SSVC:callDolevArriving} (which is then used in line~\ref{alg:SSVC:receiveResetPairs}).

\subparagraph{Case a}  
In this case $Z_{j_1}.curr.\ell \lb Z_{i_1}.curr.\ell$ holds (cf. Section~\ref{s:DolevPaper} regarding the $\lb$ relation).
That is, $Z_{i_1}.curr.\ell$ was the value of $local_i.curr.\ell$ in the system state before $a_x$, $Z_{i_1}.curr.\ell$ remains as the maximal label of $p_i$ even after $p_i$ receives $Z_{j_1}$ in $a_x$, and hence $p_i$ uses $Z_{i_1}.curr.\ell$ in the return pair of $\resetLocal_i()$ in $a_x$, $Z_{i_2} = \la\la Z_{i_1}.curr.\ell, zrs, zrs\ra, \la Z_{i_1}.curr.\ell, zrs, zrs\ra\ra$.
We show that in a system state that immediately follows a step $a_{y_a}$ in which $p_j$ receives $Z_{i_2}$ from $p_i$ (hence after $a_x$), $\s(Z_{j_1}) \neq \s(local_j)$ holds.
This is true due to the fact that $p_j$ receives the message $m_{i,j} = \langle \bullet$, $\la Z_{i_2}, \bullet\ra \rangle$ from $p_i$, such that $Z_{i_2}.curr.\ell = Z_{i_1}.curr.\ell$, or a message from $p_i$ with a pair which has a label larger than $Z_{i_2}.curr.\ell$ (due to the token passing mechanism in lines~\ref{alg:SSVC:doForeverEnd}, \ref{alg:SSVC:storeArriving}, and~\ref{alg:SSVC:exhArrCheck}). 
Since $Z_{j_1}.curr.\ell \lb Z_{i_1}.curr.\ell$ (this case's assumption), we have that $Z_{j_1}.curr.\ell$ cannot be the label that appears in $local_j.curr.\ell$ after $a_{y_a}$, because during $a_{y_a}$ line~\ref{alg:SSVC:callDolevArriving} causes $p_j$ to adopt the label $Z_{i_2}.curr.\ell = Z_{i_1}.curr.\ell,$ since we have $Z_{j_1}.curr.\ell \lb Z_{i_1}.curr.\ell$ (or a label larger than $Z_{i_1}.curr.\ell$).

\subparagraph{Case b}  
In this case $Z_{i_1}.curr.\ell \lb Z_{j_1}.curr.\ell$ holds, due to the fact that $p_i$ sets $local_i = Z_{i_2} = \la\la Z_{j_1}.curr.\ell, zrs, zrs\ra, \la Z_{j_1}.curr.\ell, zrs, zrs\ra\ra$ when calling $\resetLocal_i()$ in step $a_x$.
Let $a_{y_b}$ be the step (that follows $a_x$) when $p_j$ receives $Z_{i_2}$ and $c_{y_b}$ be the system state that immediately precedes $a_{y_b}$.
Note that the next pair after $a_x$ that $p_j$ will receive from $p_i$ can possibly have a larger label than $\ell_{i_2} = Z_{i_2}.curr.\ell = Z_{j_1}.curr.\ell$, but this event falls in Case c, which we study below.
Thus, in case $p_j$ indeed receives $Z_{i_2}$ in $a_{y_b}$, either (b-i) $\s(local_j) = \s(Z_{i_2})$ or (b-ii) $\s(local_j) \neq \s(Z_{i_2})$ holds (in $c_{y_b}$).


In case (b-i) $\s(local_j) = \s(Z_{i_2})$, the two pairs can be merged to $local_j = Z_{j_2}$, such that $\s(Z_{j_2}) = \s(Z_{i_2}) = \la\la\ell_{i_2}, \bot, zrs\ra, \la\ell_{i_2}, zrs, zrs\ra\ra$ is the representation of $Z_{j_2}$'s static part.  
Thus, in a step that follows $a_{y_b}$, processor $p_j$ sends $Z_{j_2}$ to $p_i$ and in step $a_{z_b}$, processor $p_i$ receives $Z_{j_2}$.
Consider the case where $p_i$ merges $Z_{j_2}$ with $local_i$ in $a_{z_b}$ to $local_i = Z_{i_3}$.
\begin{itemize}
\item If $\s(Z_{i_3}) = \s(Z_{j_2})$ (due to merge or a $\resetLocal_i()$ that used $\ell_{i_2}$ as $p_i$'s maximal label), then we loop back to the beginning of Case b (without having an additional call to $\resetLocal_i()$).


\item Otherwise, if $\s(Z_{i_3}) \neq \s(Z_{j_2})$, then $\ell_{i_2}$ is not the maximal label in $p_i$ (due to a call to $merge_i()$ (line~\ref{alg:SSVC:merge}) or a call  to $\resetLocal_i()$ in $a_{z_b}$). 
Hence, $Z_{i_3}.curr.\ell$ is either larger than $\ell_{i_2}$ or cancels $\ell_{i_2}$.
Subsequently, once $p_j$ receives $Z_{i_3}.curr.\ell$ from $p_i$ (due to the token passing mechanism in lines~\ref{alg:SSVC:doForeverEnd}, \ref{alg:SSVC:storeArriving}, and~\ref{alg:SSVC:exhArrCheck}), $local_j.curr.\ell$ will change to either $Z_{i_3}.curr.\ell$ or a larger label that resides in $p_j$.
\end{itemize}


Hence, in this subcase (b-i), a single pair static part that was stored in $local_j$, can cause $p_i$ to call $\resetLocal_i()$ at most twice due to the static part of $Z_{j_1}$ before $p_j$ changes the static part of $local_j$ to another one.


In case (b-ii), the fact that $\s(local_j) \neq \s(Z_{i_2})$ holds in the system state immediately before $a_{y_b}$ implies that $Z_{i_2}.curr.\ell = Z_{j_1}.curr.\ell$ is not the maximal label in $p_j$ immediately before $a_{y_b}$.
The latter holds, because $local_j$ used to hold the value $Z_{j_1}$ before $a_{y_b}$ and $p_j$ can only substitute the value of $local_j.curr.\ell$ for a label with a larger label than $Z_{i_2}.curr.\ell = Z_{j_1}.curr.\ell$.
Hence, for the value of $local_j$ in the system state that immediately follows $a_{y_b}$, it holds that $\s(local_j) \neq \s(Z_{j_1})$, since $local_j.curr.\ell$ cannot be equal to $Z_{j_1}.curr.\ell$.

 

\subparagraph{Case c}  
In this case both $Z_{i_1}.curr.\ell$ and $Z_{j_1}.curr.\ell$ are canceled in $a_x$ and $p_i$ creates a larger label $\ell_{i_3}$ to use in $Z_{i_2}$. 
Thus, in the system state that immediately follows step $a_{y_c}$, in which $p_j$ receives $Z_{i_2}$ (or a pair with a $curr.\ell$ that is larger than $Z_{i_2}.curr.\ell$), it holds that $\s(local_j) \neq \s(Z_{j_1})$, since $Z_{j_1}.curr.\ell$ will not be the largest label in $p_j$'s state that immediately precedes $a_{y_c}$.\\ 


By the case analysis above, we conclude that a single pair static part in $p_j$ can cause $p_i$ to call $\resetLocal_i()$ either once (cases a, b-ii, and c) or twice (case b-i), before $p_j$ changes the static part of $local_j$ to another one (different from the static part of $Z_{j_1}$).

\begin{figure*}[t!]
\centering
\includegraphics[scale=0.69]{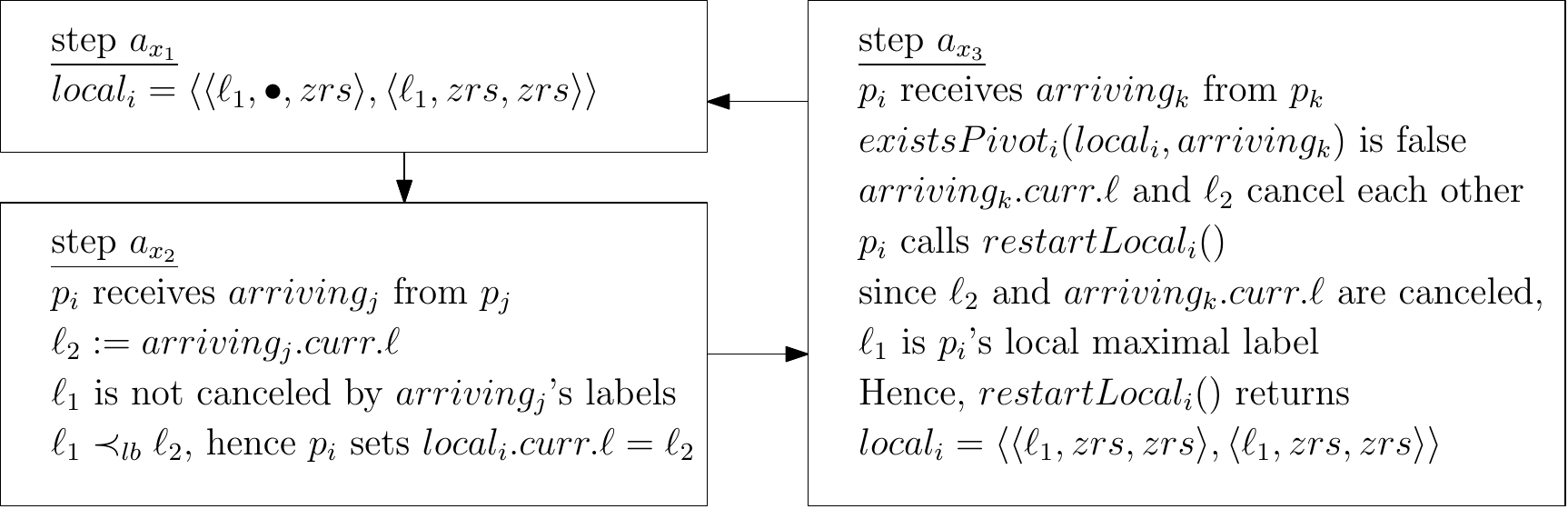}
\caption{Recycling of a pair static part by a processor $p_i$. 
In the end of step $a_{x_1}$, processor $p_i$ stores $local_i = \la \la \ell_1, \bullet, zrs\ra, \la \ell_1, zrs, zrs\ra \ra$.
In step $a_{x_2}$, $p_i$ receives $arriving_j$ from $p_j$, such that $\ell_2 := arriving_j.curr.\ell$ becomes $p_i$'s maximal label without canceling $\ell_1$ (the labels were created by different processors).
In step $a_{x_3}$, $p_i$ receives $arriving_k$ from $p_k$, such that $arriving_k.curr.\ell$ and $\ell_2$ cancel each other.
Hence, $\ell_1$, becomes again $p_i$'s local maximal label and after $p_i$ calls $\resetLocal_i()$, $local_i = \la \la \ell_1, zrs, zrs\ra, \la \ell_1, zrs, zrs\ra \ra$ holds.
That is, $p_i$ \textit{recycled} the same pair static part.
Steps $a_{x_1}$, $a_{x_2}$, and $a_{x_3}$ can be repeated (possibly with other steps in between) at most $L$ times (Claim~\ref{cl:pairRecyclingBound}, Part II).}
\label{fig:pairRecycling}
\end{figure*}

\paragraph{Part II} 

We show that there is a bound on the number of times a processor can create the same pair static part.
Let us consider a scenario in which processor $p_j$ stores the pair $Z_1$ in $local_j$ at some system state $c$, and then stores the pair $Z_2$, such that $\s(Z_1) \neq \s(Z_2)$, at a system state $c'$ that follows $c$, before creating (via $\resetLocal_j()$) the pair $Z_3$, such that $\s(Z_1) = \s(Z_3)$, which results in a subsequent system state $c''$ that follows $c'$.
We illustrate this scenario in Figure~\ref{fig:pairRecycling}.

We argue that $Z_3.curr.\ell = Z_3.prev.\ell$ holds in $c''$. 
Assume, towards a contradiction, that $Z_3.curr.\ell \neq Z_3.prev.\ell$ holds in $c''$.
Since $\s(Z_1) = \s(Z_3)$, we have that $Z_1.curr.\ell \neq Z_1.prev.\ell$ holds in $c$.
Moreover, since line~\ref{alg:SSVC:removeStaleInfo} makes sure that either $Z_1.curr.\ell=Z_1.prev.\ell$ or $Z_1.prev.\ell$ is canceled, it holds that $Z_1.prev.\ell$ is canceled, and hence $Z_3.prev.\ell = Z_1.prev.\ell$ is canceled.
Thus, $p_j$ in a step that follows $c$ used the canceled label $Z_1.prev.\ell$ to create a new pair and store it in $local_j$, which is a contradiction, because $\getLabel_j()$ by its definition never returns a canceled label.
Thus, it can only be the case that $Z_3.curr.\ell = Z_3.prev.\ell$ in $c''$, which means that $p_j$ created $Z_3$ via a call to a $\resetLocal_j()$.
The latter implies that $\s(Z_1) = \s(Z_3) = \la\la \ell_x, \bot, zrs \ra, \la \ell_x, zrs, zrs \ra\ra$, where $\ell_x := Z_k.curr.\ell = Z_k.prev.\ell$, for $k\in \{1,3\}$.


In fact, for this scenario to occur, $\ell_x$ should remain non-canceled in the label storage of $p_j$ between $c$ (where $local_j = Z_1$) and $c''$ (where $local_j = Z_3$), so that $p_j$ can recycle $\ell_x$ via the labeling algorithm and have $\ell_x$ returned through $\getLabel_j()$, as part of a $\resetLocal_j()$.
For that to happen, $Z_2.curr.\ell$ must be canceled by another label (so then $p_j$ recycles $\ell_x$).
Such cancelation scenarios can occur at most $L$ times in $R'$, since there exist at most $L$ labels in $R'$.


We remark that the number of pairs with static part different than $\s(Z_1)$ that $p_j$ stores in $local_j$ between the states $c$ and $c''$ does not change the fact that $p_j$ can (create and thus) store $Z_3$ in $c''$.
That is, the recycling scenario that we describe above with $Z_1$ and $Z_3$, is possible to occur even if $p_j$ stores the pairs $Z_k$, for every $k$ in a set of indices $K$, such that $\s(Z_k) \neq \s(Z_1)$, $k\in K$.
However, for the recycling scenario to occur we require that all the labels of the pairs $Z_k$, $k\in K$, cancel each other and make $\ell_x$ the maximal label in $p_j$'s state in a system state between $c$ and $c''$, which results to $p_j$ using $\ell_x$ to create $Z_3$.


We now show that by combining Part I and II we obtain the claim's bounds.
By Part I, a single pair static part of a pair $Z$ that processor $p_j$ stores can cause another processor $p_i$ to call $\resetLocal_i()$ at most twice before $p_j$ stores a pair with a different static part than $Z$.
By Part II, after $p_j$ stores a pair in $local_j$ that has a static part different than $Z$, it can create again a pair with the same static part as $Z$ at most $L$ times in $R'$.
Hence, a single pair static part can cause $p_i$ to call $\resetLocal_i()$ in at most $2L$ steps in $R'$, hence $2(N-1)\cdot L \leq 2N\cdot L$ for all other processors (since there are $N-1$ choices for $p_i$).
Since, there are $N$ choices for $p_j$ there can be at most $2N^2L$ steps that include a call to $\resetLocal()$ for each pair static part.
\end{proof}

In the proof of Claim~\ref{cl:receiveRestartsBound}, we use the claims of this lemma and Lemma~\ref{rem:causesOfRestart} to show that the number of steps in $R'$ that include a call to $\resetLocal()$ due to line~\ref{alg:SSVC:receiveResetPairs} is significantly less than $|R'|$.

\begin{claim}
\label{cl:receiveRestartsBound}
The number of steps that include a call to $\resetLocal()$ due to line~\ref{alg:SSVC:receiveResetPairs} in $R'$ is in $\bigO(N^8\capacity)$.
\end{claim}

\begin{proof}[\textbf{\emph{Proof of Claim~\ref{cl:receiveRestartsBound}}}]
First, we show that the claim is true when there are no calls to $revive()$ in $R'$.
Then, we extend our arguments to show that the claim holds even when there are calls to $revive()$ during $R'$.
In the following, we denote with $V = N + N^2 + N^3\capacity$ the maximum number of steps that include a call to $revive()$ (Claim~\ref{cl:reviveBound}) and $L = (N + N^2 + N^3\cdot\capacity) + (4N^2 + 4N\msg - 4N -2\msg) \in \bigO(\capacity N^3)$ the maximum number of labels that can be created during $R'$ (Claim~\ref{cl:newLabelBound}).

First, suppose that there are no calls to $revive()$ during $R'$.
Recall that there are at most $N+\msg$ distinct pairs in the starting system state of $R'$, $c_x$.
Since there are no calls to $revive()$ during $R'$ and due to Lemma~\ref{rem:causesOfRestart}, any pair that appears in $R'$ and differs to the ones that appear in $c_x$ with respect to their pair static part (i.e., all pair variables except for $curr.m$), can only be created in a step of $R'$ in which a processor called $\resetLocal()$.
A pair that is an output of $\resetLocal()$ has the form $\la\la\ell, zrs, zrs\ra, \la\ell, zrs, zrs\ra\ra$, where $\ell = \getLabel()$ is the local maximal label and $zrs$ is an $N$-size vector of zeros.
%
Thus, during $R'$ there can be at most $N+M+L$ pairs with respect to their static part.
The latter holds, since 
(i) each of the $N+M$ pairs from the starting state can be the input of a call to $\resetLocal()$ (line~\ref{alg:SSVC:receiveResetPairs}), and
(ii) any further call to $\resetLocal()$ produces a pair of the form $\la\la\ell, zrs, zrs\ra, \la\ell, zrs, zrs\ra\ra$, and there can be at most $L$ such pairs during $R'$ due to Claim~\ref{cl:newLabelBound} (since their static part only differs on $\ell$).

Hence, if there are no calls to $revive()$ during $R'$, there can be at most $(2\capacity + 1)N^2 + \capacity N^2 + N^2 L (N + M + L)$ calls to $\resetLocal()$.
This bound holds, since at most $2\capacity+1$ calls to $\resetLocal()$ can be caused due Claim~\ref{cl:linkStabilization}, $\capacity N^2$ due to Claim~\ref{cl:mechanismBound}, there can be at most $N + M + L$ pair static parts in $R'$, 
and each of them can cause at most $2N^2 L$ calls to $\resetLocal()$ due to Claim~\ref{cl:pairRecyclingBound} (including concurrent calls).

In case there exist steps in $R'$ that include calls to $revive()$, then at most $2V$ more pair static parts are added in the system.
The latter holds, because each of the $V$ pair static parts are added in the system by the output of $revive()$, can be the input to $\resetLocal()$, which in turn creates a new pair static part (hence at most $V$ more pairs with different static parts).
Thus, we update the calculation of the bound as follows: 
$(2\capacity + 1)N^2 + \capacity N^2 + 2N^2 L (N + M + L + 2V) \in \bigO(N^8\capacity)$.
\end{proof}


We are now ready to combine the claims of this proof to prove the lemma statement.
In the beginning of the proof we showed that each processor calls $\resetLocal()$ in line~\ref{alg:SSVC:removeStaleInfo} at most once for any execution and in Claim~\ref{cl:receiveRestartsBound} we showed that during $R'$ each processor calls $\resetLocal()$ in line~\ref{alg:SSVC:receiveResetPairs} in a number of steps that is significantly less than $\MI$. 
Thus, the number of steps in which a processor calls $revive()$ (Claim~\ref{cl:reviveBound}) or $\resetLocal()$ during $R'$ is significantly less than $\MI$. 
%
\end{proof}


\begin{corollary}
\label{cor:practicallyStabilizing}
Let $R$ be an $\pinf$-scale execution of Algorithm~\ref{alg:SSVC}.
By the definition of $\pinf$-scale (Section~\ref{s:systemSettings}), there exists an integer $x\ll \MI$, such that $|R| = x\cdot\MI$ holds.
By Lemma~\ref{lem:boundRestartLocal} the number of steps in which a processor calls $\resetLocal()$ or $revive()$ in every $\MI$-segment $R'$ of $R$ is significantly less than $|R'| = \MI$.
Hence, since $x\ll \MI$, the number of steps in which a processor calls $\resetLocal()$ or $revive()$ in $R$ is also significantly less than $|R|$.
Therefore, by Lemma~\ref{lem:LE} the number of states in $R$ in which Requirement\reqs\ does not hold is significantly less than $|R|$, and thus (by Definition~\ref{def:practSelfStab}) Algorithm~\ref{alg:SSVC} is practically-self-stabilizing.
%
%
\end{corollary}



\section{Conclusion}
Self-stabilization often requires, within a bounded recovery period, the complete absence of stale information (that is due to transient faults). 
This paper studies stabilization criteria that are less restrictive than self-stabilization.
The design criteria that we consider allow  recovery after the occurrence of transient faults (without considering fair execution) and still tolerate crash failures, which we do not model as transient faults. 
We show the composition of two practically-self-stabilizing systems (Section~\ref{s:interfaceDolev}) and present an elegant technique for dealing with concurrent overflow events (Section~\ref{s:pair}). 
We believe that the proposed algorithm \modified{(Section~\ref{s:algorithms})}{(Section~\ref{s:pair})} and its techniques can be the basis of other practically-self-stabilizing algorithms.
\bibliographystyle{plain}
\bibliography{stabilizingCRDTbib}



\end{document}